\documentclass[11pt]{article}
\usepackage[letterpaper,margin=1.00in]{geometry}
\usepackage{amsmath, amssymb, amsthm, amsfonts}
\usepackage{bbm}
\usepackage{amscd}
\usepackage{mathrsfs}
\usepackage{mathtools}

\usepackage{comment} 
\usepackage{ifthen}
\usepackage{tikz}
\usetikzlibrary{positioning,decorations.pathreplacing}
\usepackage{ bbold }
\usepackage{graphicx}
\usepackage{color}
\usepackage{algorithm}
\usepackage[noend]{algpseudocode}
\usepackage{epstopdf}
\usepackage{wrapfig}
\usepackage{paralist}
\usepackage{wasysym}
\usepackage[textsize=tiny]{todonotes}

\usepackage{thmtools}

\usepackage{listings} 

\usepackage{pdflscape} 

\usepackage[most]{tcolorbox}
\tcbset{colback=gray!10, colframe=gray, boxrule=0.5pt, arc=2pt, left=4pt, right=4pt, top=4pt, bottom=4pt}

\usepackage{framed}
\usepackage[framemethod=tikz]{mdframed}
\usepackage[bottom]{footmisc}
\usepackage{enumitem}
\setitemize{noitemsep,topsep=3pt,parsep=3pt,partopsep=3pt}
\usepackage[font=small]{caption}
\usepackage{xspace}

\usepackage{thmtools} 
\usepackage{thm-restate} 

\usepackage[hypertexnames=false]{hyperref}
\hypersetup{
    unicode=false,          
    colorlinks=true,        
    linkcolor=red,          
    citecolor=darkgreen,        
    filecolor=magenta,      
    urlcolor=cyan           
}

\newtheorem{theorem}{Theorem}[section]
\newtheorem{lemma}[theorem]{Lemma}
\newtheorem{meta-theorem}[theorem]{Meta-Theorem}
\newtheorem{claim}[theorem]{Claim}
\newtheorem{remark}[theorem]{Remark}

\newtheorem{observation}[theorem]{Observation}
\newtheorem{definition}[theorem]{Definition}

\usepackage[capitalize, nameinlink,noabbrev]{cleveref}

\crefname{theorem}{Theorem}{Theorems}
\crefname{proposition}{Proposition}{Propositions}
\crefname{observation}{Observation}{Observations}
\crefname{lemma}{Lemma}{Lemmas}
\crefname{claim}{Claim}{Claims}
\crefname{problem}{Problem}{Problems}
\crefname{conjecture}{Conjecture}{Conjectures}
\crefname{question}{Question}{Questions}
\crefname{example}{Example}{Examples}
\crefname{fact}{Fact}{Facts}

\definecolor{darkgreen}{rgb}{0,0.5,0}

\usepackage{algcompatible}
\algnewcommand\algorithmicswitch{\textbf{switch}}
\algnewcommand\algorithmiccase{\textbf{case}}

\algdef{SE}[SWITCH]{Switch}{EndSwitch}[1]{\algorithmicswitch\ #1\ \algorithmicdo}{\algorithmicend\ \algorithmicswitch}%
\algdef{SE}[CASE]{Case}{EndCase}[1]{\algorithmiccase\ #1}{\algorithmicend\ \algorithmiccase}%
\algtext*{EndSwitch}%
\algtext*{EndCase}%

\newcommand{\eps}{\varepsilon}

\newcommand{\poly}{\operatorname{poly}}

\renewcommand{\phi}{\varphi}

\renewcommand{\paragraph}[1]{\vspace{0.15cm}\noindent {\bf #1}}

\newcommand{\FullOrShort}{full}
\ifthenelse{\equal{\FullOrShort}{full}}{
  
  \newcommand{\fullOnly}[1]{#1}
  \newcommand{\shortOnly}[1]{}

  }{

    \newcommand{\fullOnly}[1]{}
    \newcommand{\IncludePictures}[1]{}
   
  }

\usepackage{appendix}
\title{True Work-Efficiency in Parallel Derandomization}

\begin{document}
\date{}
\author{Mohsen Ghaffari \\ \small MIT \\ \small ghaffari@mit.edu \and Cheng Jiang \\ \small MIT \\ \small chengj36@mit.edu }
\maketitle

\begin{abstract} A longstanding limitation of known techniques for parallel derandomization was that they increased the work bound by at least polylogarithmic factors. For instance, for basic and frequently used problems such as maximal independent set, maximal matching, and $\Delta+1$ coloring where $\Delta$ denotes the graph's maximum degree, in $n$-node $m$-edge graphs, polylogarithmic-depth \textit{deterministic} parallel algorithms required $\Omega((m+n)\poly(\log n))$ work, e.g., Luby [FOCS'98]. This meant one needs at least $\poly(\log n)$ processors to be faster than the naive single-processor algorithm for these problems. Recently, Ghaffari and Grunau [FOCS'25] provided a new parallel derandomization method that substantially reduced the overhead from $\poly(\log n)$ to $\poly(\log\log n)$, thus achieving $O((m+n)\poly(\log \log n))$ work bounds. In this paper, we settle this line of research by obtaining linear work bounds of $O(m+n)$, hence exhibiting truly work-efficient parallel derandomization.  
\end{abstract}

   \thispagestyle{empty}

{   
    \newpage
    \hypersetup{linkcolor=blue}
    \tableofcontents
    \setcounter{page}{0}
    \thispagestyle{empty}
}

\newpage
\setcounter{page}{1}
\section{Introduction}
\subsection{Context}
\paragraph{State of the Art.} Derandomization has been one of the classic topics studied in parallel algorithms, since the 1980s---see, e.g.,~\cite{karp1984fast,luby1985simple,alon86,luby1988removing, berger1989efficient,motwani1989probabilistic, berger1989simulating}. Despite that, known techniques increased the work by at least a $\poly(\log n)$ factor.
We give a brief overview of the methods in \Cref{subsec:method}. In particular, for classic graph problems---such as maximal independent set (MIS), maximal matching, and coloring, which were the primary targets in the area~\cite{karp1984fast,luby1985simple,alon86,luby1988removing}---this meant deterministic parallel algorithms are slower than the naive sequential algorithms for these problems, unless one has at least $\poly(\log n)$ processors. Recently, Ghaffari and Grunau~\cite{GG25} introduced a new parallel derandomization framework based on gradually rounding fractional solutions, inspired by local rounding methods developed in the context of distributed graph algorithms, e.g.,~\cite{GhaffariK21, faour2022local}. They used this to reduce this overhead exponentially, obtaining deterministic parallel algorithms with $O((m+n)\poly(\log\log n))$ work and $\poly(\log n)$ depth, for maximal independent set, maximal matching, and a certain hitting set problem. In this paper, we remove the remaining $\poly(\log\log n)$ overhead and obtain \emph{truly work-efficient} deterministic parallel algorithms with $O(m+n)$ work and $\poly(\log n)$ depth, for the same set of problems. This settles this line of work by achieving deterministic parallel algorithms that have linear work bounds, matching that of the corresponding randomized parallel algorithms. We also broaden the applicability of the method by providing linear-work polylogarithmic-depth deterministic parallel algorithms for minimum set cover approximation and graph coloring. Before stating the results formally, let us recall the setup and the standard terminology.

\paragraph{Parallel model, work, and work-efficiency.}
We work in the standard \emph{work-depth} model~\cite{jaja1992introduction, blelloch1996programming}. For an algorithm $\mathcal A$, its depth $D(\mathcal A)$ is the length of the longest dependency chain, and its work $W(\mathcal A)$ is the total number of operations. On $p$ processors, the running time satisfies $
T_p(\mathcal A)\ge \max\{D(\mathcal A), W(\mathcal A)/p\}$. In addition, by Brent's principle~\cite{brent1974parallel}, we have $T_p(\mathcal A)\le D(\mathcal A)+W(\mathcal A)/p.$ Thus, to obtain meaningful parallel speedups with a moderate number of processors, one needs low depth and a work bound close to that of the best sequential algorithm. Algorithms with work bound matching the best sequential algorithm are called \emph{work-efficient}; if there are extra $\poly(\log n)$ factors, the algorithm is often called \emph{nearly} work-efficient. \footnote{We comment that, while early parallel work often emphasized depth and tolerated large polynomial work, more recent research has increasingly focused on obtaining work-efficient or nearly work-efficient algorithms; see, e.g., \cite{fineman2018nearly, jambulapati2019parallel, blelloch2020parallelism, li2020faster, andoni2020parallel, cao2020efficient,dhulipala2021theoretically,anderson2021parallel,rozhovn2022undirected,rozhovn2022deterministic, ghaffari2023work,ghaffari2024work}.}

\subsection{Our Results.}
We show that the general approach of Ghaffari and Grunau~\cite{GG25} can be strengthened and sharpened to achieve deterministic parallel algorithms with \emph{linear work} and polylogarithmic depth, for a range of problems central to this area. Next, we state our concrete results. We comment that for each statement, the formal statement needs to clarify the representation of the input, concretely, that the graph is provided in a compact representation where different computations through different endpoints of each edge can be coordinated. To simplify things, we give informal variants of the theorem statements here, and we reference their formal variants appearing in the technical sections.
\bigskip


\paragraph{Maximal Independent Set.} Our first result is for the maximal independent set problem. We present the informal statement below; the precise statement appears in \Cref{{thm:mis}}.

\begin{restatable}{theorem}{mis}\textnormal{(\textbf{Maximal Independent Set Algorithm})}
\label{thmI:MIS}
There is a deterministic parallel algorithm that, given any $n$-node $m$-edge graph $G=(V,E)$, computes a maximal independent set using $O(m+n)$ work and $\poly(\log n)$ depth.
\end{restatable}

\paragraph{Maximal Matching.} We obtain an algorithm with similar bounds for maximal matching. We present the informal statement below; the precise statement appears in \Cref{thm:maximal_matching}.

\begin{restatable}{theorem}{mm}\textnormal{(\textbf{Maximal Matching Algorithm})}
\label{thmI:MM}
There is a deterministic parallel algorithm that, given any $n$-node $m$-edge graph $G=(V,E)$, computes a maximal matching using $O(m+n)$ work and $\poly(\log n)$ depth.
\end{restatable}

\paragraph{Hitting Set.} As in Ghaffari and Grunau~\cite{GG25}, a central ingredient in the above results is a deterministic parallel algorithm for a hitting set problem -- this is an abstraction that captures many basic usages of randomness in parallel graph algorithms. The setup is a bipartite graph $G=(U\sqcup V,E)$ together with fractional values $p_v\in[0,1]$ such that each $u\in U$ sees total mass $\sum_{v\in N_G(u)}p_v\in\Theta(1)$. The goal is to deterministically select a set $S\subseteq V$ that hits a large fraction of $U$, while ensuring each hit node is hit only $O(1)$ times. We present the informal statement below; the precise statement appears in \Cref{thm:main-hitting-set}.

\begin{restatable}{theorem}{hittingset}\textnormal{(\textbf{Hitting Set Algorithm})}
\label{thmI:hittingSet}
Consider an $n$-node $m$-edge bipartite graph $G=(U\sqcup V,E)$ and suppose that each $v\in V$ has a probability $p_v\in[0,1]$ such that $\forall u\in U$, we have $\sum_{v\in N_G(u)} p_v\in\Theta(1)$. Suppose also that each node $u\in U$ has a real-valued importance $imp_u\in\mathbb{R}_{+}$. There is a deterministic parallel algorithm that, using $O(m+n)$ work and $\poly(\log n)$ depth, computes two subsets $S\subseteq V$ and $U_{good}\subseteq U$ such that
\begin{itemize}
    \item $\sum_{u\in U_{good}} imp_u \ge 0.99 \sum_{u\in U} imp_u$,
    \item $\forall u\in U_{good}$, we have $|N_G(u)\cap S|\in\Theta(1)$.
\end{itemize}
\end{restatable}

\paragraph{Defective Coloring.} An additional ingredient, crucial to our linear-work bound in all these algorithms, is a linear-work subroutine that given a graph with edge weights and a parameter $\eps\geq 1/\poly(\log n)$, computes a coloring of the vertices so that at most $\eps$ fraction of the weight is monochromatic. This is used in all our derandomizations. We present the informal statement below; the precise statement appears in \Cref{thm:near-optimal-defective-coloring}.

\begin{restatable}{theorem}{defectivecoloring}\textnormal{(\textbf{Weighted Defective Coloring Algorithm})}
\label{thmI:defectiveColoring}
There is a deterministic parallel algorithm that, given any $n$-node $m$-edge graph $G=(V,E)$ with nonnegative edge weights and any parameter $\eps \in (0,1]$, computes a coloring of the vertices using $O(1/\eps)$ colors such that the total weight of monochromatic edges is at most an $\eps$ fraction of the total edge weight. The algorithm uses $O(m+n)$ work and $\poly(\log n)$ depth.
\end{restatable}

\paragraph{$(\Delta+1)$-Coloring.} We also expand the reach of these techniques by addressing two other classic problems in the area. The first such problem is coloring. We obtain a work-efficient deterministic parallel algorithm for $(\Delta+1)$-coloring. We present the informal statement below; the precise statement appears in \Cref{thm:delta+1_coloring}.

\begin{restatable}{theorem}{coloring}\textnormal{(\textbf{$(\Delta+1)$-Coloring Algorithm})}
\label{thmI:coloring}
There is a deterministic parallel algorithm that, given any $n$-node $m$-edge graph $G=(V,E)$ of maximum degree $\Delta$, computes a proper $(\Delta+1)$-coloring using $O(m+n)$ work and $\poly(\log n)$ depth.
\end{restatable}

\paragraph{Minimum Set Cover Approximation.} Finally, as our other result expanding the applications, we give a work-efficient deterministic parallel approximation algorithm for set cover, achieving a $\tilde{O}(\log S)$ approximation, where $S$ denotes the maximum size of a set in the input instance. We present the informal statement below; the precise statement appears in \Cref{thm:set-cover}.

\begin{restatable}{theorem}{setcover}\textnormal{(\textbf{Set Cover Approximation Algorithm})}
\label{thmI:setcover}
Suppose we are given a set cover instance, as a bipartite graph with one side for the sets and the other side for the elements, where an element node is connected to all sets that contain it. Let $n$ and $m$ be the number of nodes and edges in this graph, and let $S$ be the size of the largest set. There is a deterministic parallel algorithm that computes an $O(\log S \log\log n)$ approximation of minimum set cover, using $O(m+n)$ work and $\poly(\log n)$ depth.
\end{restatable}

\subsection{Methods}\label{subsec:method}
\paragraph{Classic work--limited independence and polynomial work.} To simplify the discussion, we focus on MIS, which has become a de facto benchmark problem for parallel derandomization. Karp and Wigderson~\cite{karp1984fast} gave the first deterministic parallel algorithm for MIS. Their method can be viewed as replacing a huge randomized search space by a smaller explicitly enumerable family of combinatorial objects, and then checking all possibilities in parallel. This yields polylogarithmic depth, but with very large work. Luby~\cite{luby1985simple,luby1988removing} made this paradigm substantially more general. His randomized MIS algorithm has expected $O(m)$ work and $\poly(\log n)$ depth, and its analysis uses only pairwise independence. A direct derandomization, therefore, restricts attention to a pairwise independent sample space of size $O(n^2)$ and checks all seeds in parallel, thus giving a deterministic parallel algorithm with $\tilde{O}(mn^2)$ work. Alon, Babai, and Itai~\cite{alon86} obtained similar algorithms and more generally showed how limited independence yields deterministic parallel derandomizations with overhead $n^{\Theta(d)}$ for $d$-wise independence.
\medskip

\paragraph{Classic work--binary search through a pairwise independent space.} In a subsequent paper focused on the work-bound overhead~\cite{luby1988removing}, Luby introduced a method to perform a binary search through the pairwise space. Instead of enumerating all seeds, one fixes the seed gradually, one bit at a time, maintaining (a pessimistic estimator on) the conditional expectation of a suitable potential function. This is the most generic parallel derandomization method known to date for obtaining near-linear work bounds, and it underlies several fundamental deterministic parallel algorithms, including those on maximal independent set, maximal matching, and set cover. See e.g., \cite{luby1988removing, berger1989efficient}. However, it inherently incurs at least a $\poly(\log n)$ work overhead: there are $\Omega(\log n)$ stages, and each stage already costs at least linear work (and indeed more, since $O(\log n)$ bits are used already for $1/2$ sampling and there are more stages for different probabilities).
\medskip

\paragraph{Local rounding.} Ghaffari and Grunau~\cite{GG25} showed that one can adapt a variant of the \textit{local rounding} technique recently developed in the context of \textit{distributed graph algorithms}~\cite{GhaffariK21, faour2022local} to substantially improve the work efficiency of the derandomization, achieving $O((m+n)\poly(\log\log n))$ work bounds. At a high level, one starts with a fractional solution that abstracts the probabilistic assignments in (a step of) the randomized parallel algorithm, and one gradually rounds this fractional solution, getting it closer to integrality in stages, while preserving certain bounds on the objective functions. 
\medskip

\paragraph{Our algorithms.} We follow the same general idea of iteratively rounding fractional solutions in gradual steps, while approximately preserving some quadratic objective functions stemming from pairwise analysis~\cite{GhaffariK21, faour2022local, GG25}. However, the outline of how we round various probabilities, and even some key ingredients per step of rounding---most notably, the defective coloring used to ignore certain parts of the objective function, so as to admit a fast schedule for the derandomization ---are quite different here. We next expand on this. 

Consider the toy problem of computing an approximate max cut in an edge-weighted graph $G=(V, E)$, which can be modeled as finding a $\mathbf{x}\in \{0,1\}^{V}$ that maximizes $Z(\mathbf{x}) = \sum_{\{v, v'\}\in E} w(e)\cdot (x_v + x_{v'} - 2x_v x_{v'})$. The fractional relaxation considers assignments $\mathbf{x}\in [0,1]^{V}$, and the natural fractional solution sets $x_v=1/2$ for all $v\in V$ and yields $Z(\mathbf{x})=\frac{1}{2}\sum_{e\in E} w(e)$. Approximate rounding, while preserving this quadratic objective function, means finding $\mathbf{x}\in \{0,1\}^{V}$ with $Z(\mathbf{x})\geq (\frac{1}{2}-\eps)(\sum_{e\in E} w(e))$ for some small $\eps>0$. This is a $(1/2-\eps)$ approximation of the max-cut, but actually a $(1-\Theta(\eps))$ approximation of the fractional solution's score. It gives a simple example for one step of rounding from $(1/2)$-fractional solutions to $0$ or $1$ integral solutions, while approximately preserving one quadratic objective function.

\medskip
\paragraph{Linear-work defective coloring}. Defective coloring is a key ingredient that enables the gradual rounding with low parallel depth, throughout the entire rounding scheme. Consider a coloring of the vertices with $Q$ colors such that the weight of the monochromatic edges is at most $\eps(\sum_{e\in E} w(e))$. This is called \textit{$\eps$-relative defect}. We can ignore monochromatic edges. Then, we process the vertices by going through the color classes, and per color $q\in\{1,2,\dots, Q\}$, we decide for each vertex of color $q$ greedily: put it on the side such that at least half of its edge weight to vertices in colors $[q-1]$ is across the cut. This gives an $O(m+n)$ work algorithm with depth $Q$, which obtains a $(1/2-\eps)$ approximation of max cut. The only aspect not spelled out is how one obtains such a defective coloring, which is the key algorithmic ingredient in such a single step of rounding.

One of the most basic sources of the $(\log\log n)$ overhead factors in the work of Ghaffari and Grunau~\cite{GG25} was their defective coloring. They computed a coloring with $O(1/\eps)$ colors and $\eps$-relative defect, using $O((m+n)\log\log n)$ work and $\poly(\log n)$ depth. We improve this to $O(m+n)$. This ingredient is used in all of our results. The algorithm of Ghaffari and Grunau~\cite{GG25} relied on a Linial-style color reduction~\cite{linial1987LOCAL}, starting from the trivial bound of $n$ colors. To have efficient parallel computations, they could reduce the number of colors only polynomially per step, spending $O(m+n)$ per step, and this necessitated $O(\log\log n)$ iterations and thus $O((m+n)\log\log n)$. We show that (after a small constant number of such reductions), we can perform an exponential step of color reduction, essentially from $O(n^{\alpha})$, for some small constant $\alpha>0$, to $O(\log n)$, with just $O(m+n)$ work. This involves building a certain combinatorial design efficiently and deterministically and then searching for the right color per node with the help of this design, via carefully structuring the word RAM operations. Afterward, we are able to reduce the colors to $O(1/\eps)$ with more standard techniques. See \Cref{subsec:Linial}. This already improves one of the overhead factors, and indeed, given our $O(1/\eps)$-color $\eps$-relative defective coloring, which is computed in $O(m+n)$ work and $\poly(\log n)$ depth for any $\eps>1/\poly(\log n)$, an $O(m+n)$ work and $\poly(\log n)$ depth deterministic parallel algorithm for $(1/2-\eps)$ approximation of (weighted) max cut is immediate.
\medskip

\paragraph{Complexities in general rounding.} The general rounding framework, in its applications for various problems considered in this paper, is much more complex than the nice example of max cut discussed above, in a number of ways: (1) Often, instead of a single direct objective function, there are many constraints that we would like to satisfy (and we will do so for most of them, in an importance-weighted sense). For instance, in the hitting set problem (\Cref{thmI:hittingSet}), we would want to control the number of deterministically sampled $V$-nodes neighboring each $U$-node---this is $|U|$ many different constraints. To simultaneously push all these constraints, the idea is to use a Chebyshev-like quadratic measure of deviations around each node. If done naively, the size of such an objective function would be quadratically large, hence defeating the linear-work goal. The workaround is to write such functions over small bundles of nodes, which allows one to control the losses to within an inverse polynomial of the bundle size, though also making the problem larger roughly by the bundle size. One needs careful control of the bundle sizes over time to ensure that the work bound and the accumulative losses remain small. (2) Usually, the rounding has to deal with fractional values that are much smaller, e.g., $1/n$, and thus we need $\tilde{\Theta}(\log n)$ steps of gradual rounding. Informally, the reason is that with any $\eps$-relative defective coloring, since we must have $\eps\geq 1/\poly(\log n)$ to allow us to process the colors one by one in $\poly(\log n)$ depth, there is $\eps\geq 1/\poly(\log n)$ fraction of the objective that is not really controlled, and a large $\poly(\log n)$-factor jump in fractionality could blow up this part. However, we need to ensure that the work bound over these $\tilde{\Theta}(\log n)$ steps does not add up proportionately to the number of steps, and for that, it is critical to ensure that certain ``size" measures decay quickly over the steps. (3) Often, the rounding has to deal with non-uniform fractional solutions where different variables have different fractional values (probabilities). Controlling the losses over different probability classes, as they get rounded, is not straightforward, especially since we need to ensure the size decay outlined before. We do not open the discussion of how all these aspects are handled. However, we point out one general aspect where the approach used in \cite{GG25} gives rise to $\poly(\log\log n)$ overhead factors and outline how we circumvent them, and we also point out a number of novel subroutines that we build here to avoid the overhead factors.
\medskip

\paragraph{Controlling the losses over different steps, and the work overhead.} An important source of the $O(\log\log n)$ overheads in \cite{GG25} is for the rounding of relatively high probabilities, which are above $1/\poly(\log n)$. Since one can roughly group these in factor-$2$ bands, there are $\Theta(\log \log n)$ such probability groups, and also $\Theta(\log \log n)$ steps of $2$-factor rounding until reaching integral solutions. To ensure that the objective function loses at most a constant factor overall, we need to enforce that most steps lose only an $O(1/\log\log n)$ fraction of the objective. This means writing the Chebyshev-like quadratic constraints mentioned in item (1) above in bundles of size at least $\Omega(\log\log n)$, and this implies the size of the objective function is at least $\Omega(\log\log n)$ larger than the number of elements. Since the scheme is dealing with the rounding of different probabilities, it is not clear how to argue or ensure that the problem size has had enough shrinkage over the steps to leave room for the $\Omega(\log\log n)$ overhead stemming from these bundle sizes. Indeed, the work bound in \cite{GG25} for this portion is at least $\Omega(m+n)$ per step (even assuming ideal bounds for other ingredients such as defective coloring and some subproblem selection and sorting subroutines), and thus $\Omega((m+n)\log \log n)$ overall. We have a different outline of rounding for different probabilities. In a short and informal sense, we round each probability class $p$ to roughly $\sqrt{p}$, and we show that after this, the problem size overall has had enough shrinkage to allow us to use relatively large bundles, and thus achieve fairly small errors per step, without blowing up the overall work bound, until reaching fairly high probabilities where we can gradually shrink the bundle sizes (the latter is because, unfortunately, large bundles themselves come with a certain loss, due to the bundle irregularities and nodes left out of bundles, and this needs a tight control once we deal with relatively large probabilities). The actual scheme is somewhat involved, and we refer to \Cref{sec:HittingSet}.
\medskip

\paragraph{Novel subroutines.} Besides the aspects discussed above, to remove $\poly(\log\log n)$ overheads, our algorithms rely on a number of novel parallel algorithmic subroutines. We mention some of them here. (I) An important subroutine is a deterministic parallel sorting algorithm that allows us to sort $n$ integers in a $[\poly(\log n)]$ range in $O(n)$ work and $\poly(\log n)$ depth. See \Cref{subsec:sort}. The best-known prior solution for this required $\Omega(n\log\log\log n)$ work~\cite{bhatt1991improved}. We believe this sorting itself may find wider applications. (II) Another subroutine allows us to iteratively extract, out of the original graph, the subgraphs ``incident" on the largest degree class, and maintain this under batch removal of vertices, in a total work that is linear in the graph size. This scheme is frequently used in our algorithms, as we solve various graph problems gradually, by each time focusing on the largest degree class. It is crucial that the work bound overall (and not just per iteration) is linear in the original graph size, and doing this via a deterministic parallel algorithm requires nontrivial ideas. See \Cref{subsec:extractor}. (III) To control the losses mentioned above, we often need careful bundling of the elements to be sampled, and we need to efficiently and explicitly build an auxiliary graph that captures the dependencies in these bundles (e.g., all pairwise terms per bundle) with a compact representation. Doing this is not straightforward since the natural approaches via integer sorting (and graph representation switches) would incur $O(\log\log n)$ overhead terms~\cite{bhatt1991improved}. See \Cref{subsec:12-sampling} for our solution.

\medskip
\paragraph{New problems.} Finally, we note that our algorithms for $(\Delta+1)$-coloring and minimum set cover approximation are the first deterministic ones to avoid $\poly(\log n)$ overhead factors in the work bound, and they show how the method's applicability can be broadened to some of the other important problems studied in the area. For instance, the set cover problem has to deal with $\poly(\log n)$ different stages of randomized sampling, and thus derandomizing them. See, e.g., the classic work of Berger, Rompel, and Shor~\cite{berger1989efficient} for a derandomization with $\poly(\log n)$ overhead for this problem. Doing all these iterations while ensuring that the total work bound overall remains linear requires additional ideas. 
\medskip


\section{Preliminaries}
\subsection{PRAM Model and Setting}

\paragraph{Model.} We work in the integer ARBITRARY Concurrent Read Concurrent Write (CRCW) PRAM model. In the case of multiple simultaneous write attempts, an \textit{arbitrary} one takes place. The word size is assumed to be $w=\log n$ bits, and all arithmetic operations are performed on $w$-bit integers. The supported RAM operations include indexing, bitwise operations, addition, subtraction, multiplication, and division on $w$-bit integers.

\paragraph{Graph notations.} Throughout, $n$ and $m$ refer to the size of the \emph{input} graph
$G_0=(V_0,E_0)$, where $|V_0|=n$ and $|E_0|=m$.
When we consider subgraphs $G=(V,E)$ of $G_0$, we use $|V|$ and $|E|$ for their
sizes. Note that the machine word size remains $w=\log n$, determined by the
input size, even when we process subgraphs.


\paragraph{Graph representation format.} We assume the input graph is given in compact adjacency-array form: each vertex stores an array containing all of its neighbors and knows its degree. The relative order of neighbors may be arbitrary. In addition, we assume that for each edge $(u,v)$, the endpoint $u$ knows the position of $u$ in the adjacency list of $v$. Notice that, even if this format is not given, it can be achieved if we have \(O(n^2)\) shared memory, of which only \(O(m)\) locations are ever used: for each edge $(u,v)$, the vertex \(v\) can simply write the relevant index into the shared cell labeled by the pair \((u,v)\). Alternatively, it could be obtained by sorting the edges by their endpoints, though at the moment the best known deterministic parallel sorting, for integers in a polynomial range, uses $O((n+m)\log\log n)$ work and $\poly(\log n)$ depth~\cite{bhatt1991improved} (despite the fact that in sequential computation, $O(n+m)$ work suffices for this, e.g., via radix sort).
\subsection{Subgraph Selection Tools}

In this section, we exploit word-level parallelism to speed up certain
operations by a factor of $O(\log n)$, effectively batching
$O(\log n)$ operations into one word operation.

For example, sorting integers bounded by $\poly(\log n)$ can be done in linear work.
In contrast, sorting $n$ integers of $\poly(n)$ magnitude still requires
$O(n\log\log n)$ work.

Another important application is subgraph splitting.
If a graph is colored using $\poly(\log n)$ colors,
we can construct compact representations of all induced subgraphs
in $O(n+m)$ work.
This requires nontrivial data movement and typically relies on sorting.

The following classical claim is stated explicitly because
we will reuse its structure later.

\begin{lemma}[Subgraph Selection]\label{clm:subgraph-splitting}
Given a graph $G$ in compact representation,
and a marking indicating which nodes and edges are selected,
we can construct the compact representation of the induced subgraph
in $O(|V|+|E|)$ work and $\poly(\log n)$ depth, where $|V|$ and $|E|$ are the numbers of vertices and edges of the graph, with $|V| \leq n$.
\end{lemma}

\begin{proof}
For each node, mark each incident edge with weight $1$ if selected
and $0$ otherwise. These marks can be computed in $O(1)$ depth
with $O(1)$ work per edge.

For each node, perform a prefix sum over its incident marks.
This takes $O(\log n)$ depth and $O(d_i)$ work,
where $d_i$ is the degree of the node.
The prefix sums determine the new positions of selected edges
in the compact adjacency array.

The same technique applies to selecting nodes and allocating the arrays for edges. 
\end{proof}

\begin{remark}
Note that the work here is charged to the edges of the input graph to be extracted, not to the extracted subgraph.
In later sections, we will show how to split the graph into up to $\poly(\log n)$ subgraphs within $O(|V|+|E|)$ work, where $|V| \leq n$.
\end{remark}

\begin{remark}
    To actually mark which edges are deleted usually requires the communication property of the compact representation. In other words, for each edge $(u,v)$, we need to make sure that both $(u,v)$ and $(v,u)$ are marked. We do not need this property if we are only deleting some vertices and their incident edges, as we can only mark the vertices in this case. 
\end{remark}

Notice that in some applications, we need more than extracting a single subgraph. A common case is when we split a graph into multiple subgraphs. This is generally a sorting-like task, but we can do it for $\poly(\log n)$ parameters, as stated below.

\begin{lemma}[$\log n$ subgraph splitting]\label{lem:logn-subgraph-splitting}
Suppose we color the vertices of a graph using at most $\log n$ colors.
Then we can construct compact representations of the subgraphs induced by each color
using $O(|V|+|E|)$ work and $\poly(\log n)$ depth, where $|V| \leq n$.
\end{lemma}

\begin{proof}
We first perform subgraph selection (\cref{clm:subgraph-splitting}) to keep only the edges that belong to the induced subgraphs. Then we use $\log n$ sorting (\cref{logn-sorting}) to sort the vertices by color. Finally, for each color class, we use the kept edge list as the edge list of its induced subgraph.
\end{proof}

The next subsection shows how $\log n$ sorting (\cref{logn-sorting}) is done.

\subsection{Linear-work $(\log n)$ Sorting}
\label{subsec:sort}
We now introduce a compact representation of multiple integers that will be used extensively to speed up certain calculations by a factor of $\log n$. This representation packs $\log n$ integers of $k$ bits into $k$ words and uses bitwise operations to perform simultaneous operations on them in $O(k)$ time rather than $O(\log n)$ time.
\begin{definition}[$k$-compact (packed) form]
Assume the word size is exactly $w=\log n$ bits.
Given $\log n$ integers
\[
x_0,x_1,\dots,x_{\log n-1}\qquad\text{with }x_t\in[0,2^k-1]\ \text{for all }t,
\]
their \emph{$k$-compact (packed) form} is an array of exactly $k$ words
\[
(W_0,W_1,\dots,W_{k-1}), \qquad W_j\in[0,2^{\log n}-1],
\]
defined by viewing the $k$ words as one concatenated $k\log n$-bit string and
packing the $x_t$ consecutively into this string. Formally, let
\[
X \;=\; \sum_{t=0}^{\log n-1} x_t\,2^{tk},
\]
and define, for each $j\in\{0,1,\dots,k-1\}$,
\[
W_j \;=\; \left\lfloor \frac{X}{2^{j\log n}} \right\rfloor \bmod 2^{\log n}.
\]
Equivalently, in the concatenation $W_0 \parallel W_1 \parallel \cdots \parallel W_{k-1}$,
the integer $x_t$ occupies bits $[tk,\,(t+1)k-1]$ (i.e., $k$ consecutive bits).
\end{definition}

\begin{lemma}[Effective sum of $k$-compact forms]\label{lem:effective-sum-packed}
Let $k=O(\log n)$ be an integer. With $O(n^{2/3}\log n)$ preprocessing, we can:
\begin{itemize}
\item compute the sum of two $k$-compact forms in $O(k)$ work, and
\item convert a $k$-compact form into a $(k+1)$-compact form in $O(k)$ work.
\end{itemize}
The sum of two $k$-compact forms is a $(k+1)$-compact form. Both operations have $O(1)$ depth, and the preprocessing has $\poly(\log n)$ depth.
\end{lemma}

\begin{proof}
Let $w=\Theta(\log n)$ be the word size. We divide each word into blocks of
$b=\frac{1}{3}\log n$ bits. Then each word contains $O(w/b)=O(1)$ blocks.

If $k>b$, then $k=\Theta(\log n)$ and we can handle the operation by brute force.
Otherwise, we preprocess a lookup table that specifies how to:
(i) add two blocks consisting of $\left\lfloor \frac{b}{k}\right\rfloor$ packed $k$-bit integers (producing the result in $(k+1)$-compact form), and
(ii) convert such a block from $k$-compact form into $(k+1)$-compact form.
\end{proof}

Using the compact representation above, we can accumulate over $\log n$ classes and thereby sort integers of size $\log n$ using linear work. Note that bucket sort from the sequential setting does not carry over directly to parallel computation, since it is unclear how to deal with several elements being inserted into the same bucket at the same time. More generally, obtaining a linear-work parallel algorithm for sorting integers of size $O(n)$ remains a major open problem. Here, we only leverage bitwise operations to address a limited special case.

\begin{lemma}[$\log n$ accumulation]\label{lem:logn-accumulation}
Let $k \leq \poly(n)$ be an integer. With $O(n^{2/3}\log^2 n)$ preprocessing,
we can add $k$ binary arrays of length $\log n$
(in $1$-compact form)
using $O(k)$ work and $O(\log k)$ depth.
The result is in $\lceil \log k\rceil$-compact form.
\end{lemma}

\begin{proof}
We recursively split the $k$ arrays into two halves,
sum each half,
extend if necessary,
and add the two resulting compact forms
using \cref{lem:effective-sum-packed}.

The recurrence is
\[
T(k)=2T(k/2)+O(\log k),
\]
which solves to $T(k)=O(k)$.
The depth satisfies
\[
D(k)=D(k/2)+O(1)=O(\log k).
\]
\end{proof}

\begin{lemma}[$\log n$ sorting]\label{logn-sorting}
Suppose we have $k$ pairs $(a_i,b_i)$ with key $a_i \in [\log n]$ and $b_i \in [n]$.
Then we can sort these pairs by $a_i$ (the order among pairs with the same $a_i$ may be arbitrary) using $O(k)$ work and $\poly(\log n)$ depth.
This requires $O(n)$ preprocessing.
\end{lemma}

\begin{proof}

We split into two cases: $k< \frac{1}{9}\log n$ and $k \geq \frac{1}{9}\log n$.

\paragraph{Case 1.} Assume $k \geq \frac{1}{9}\log n$.

It is safe to assume that $k$ is a multiple of $\log n$, by appending dummy pairs
$(a_i,b_i)=(\log n+1 ,n)$ with the maximum possible values. We then split the input
into groups of size $\log n$.

Within each group, we run $\log n$ accumulation (\cref{lem:logn-accumulation}) to compute, for each
$t\in[\log n]$, the number of indices $i$ in that group with $a_i=t$. After that, we unpack the compact forms and take a prefix sum over groups to compute, for each
group and each $t$, the total number of indices $i$ with $a_i=t$ in all previous
groups (and thus the starting offset for bucket $t$ inside this group). We also
obtain the global total count for each $t$, which lets us allocate the output array
and partition it into $\log n$ contiguous buckets (one per key value $t$), each of
the correct size.

Finally, within each group, we scan its $\log n$ elements once and maintain a
\emph{local} counter for each $t$ that records how many elements with key $t$ have
already been seen in this group. This local counter gives each element its rank
among the elements in the same group with the same key. We then write the element at position \(\text{pos}=\text{(global offset for bucket $t$ before this group)}+\text{(local rank within this group)}\).
The order among equal keys is arbitrary, so any consistent local ranking works.

\paragraph{Case 2.} Assume $k< \frac{1}{9}\log n$.

We handle this case using lookup tables, so that every operation costs $O(1)$ work
and we never allocate an array of length $\ell=\log n$ at runtime.

As before, we encode any count vector $(c_1,c_2,\dots,c_\ell)$ with
$\sum_{t=1}^{\ell} c_t < \frac{1}{9}\log n$ into a single word. The number of
such vectors is at most $k\binom{\ell+k}{k}<n^{0.9}$, so we can preprocess tables
over all encoded states in $O(n)$ work.

In addition to the update/extraction tables (increment/decrement, extract $c_t$,
and find the first nonzero $c_t$), we preprocess one more table that \emph{packs
the nonzero buckets}:
for each encoded state $S$ and each key $t\in[\ell]$, the table stores
\(\mathrm{loc}(S,t)\in\{0,1,2,\dots,k\}\),
where $\mathrm{loc}(S,t)=0$ iff $c_t=0$, and otherwise $\mathrm{loc}(S,t)$ is the
index of key $t$ among the nonzero keys of $S$ (in increasing order).
We also preprocess, for each state $S$ and each index $j\in[k]$, the quantity
\(\mathrm{start}(S,j)\), the starting position (prefix sum) of the $j$-th nonzero bucket in the output.

Now the algorithm is:
(1) scan the input once and update the encoded count state $S$ by applying
$\texttt{INC}(S,a_i)$ for all $i$ (total $O(k)$ work);
(2) allocate an output array of length $k$, and allocate a pointer array
$\texttt{ptr}[1..r]$ where $r$ is the number of nonzero buckets (this is at most $k$);
(3) scan the input again, and for each pair $(a_i,b_i)$ compute
$j=\mathrm{loc}(S,a_i)$ in $O(1)$ work, write it to position
$\mathrm{start}(S,j)+\texttt{ptr}[j]$, and increment $\texttt{ptr}[j]$.
This uses only $O(k)$ total work and only $O(k)$ runtime space.
\end{proof}

\begin{remark}
We can extend $\log n$ sorting (\cref{logn-sorting}) to keys in the range $[\log^c n]$, and similarly extend
graph splitting for the same range by iterating the algorithm $c$ times. This does not
immediately apply to $\log n$ accumulation (\cref{lem:logn-accumulation}), because storing the accumulated values may
require more words than the input representation.
\end{remark}

\subsection{Large Degree Extractor}
\label{subsec:extractor}
In many parallel algorithms, a common approach is to work only with vertices whose degree is at least half of the maximum degree. This is easy to maintain sequentially, but not in parallel: updating degrees is nontrivial, since several edges may try to update the same counter simultaneously, and resolving this may itself require sorting the edges. In this section, we provide a solution using $\log n$ sorting (\cref{logn-sorting}).

\begin{lemma}\label{big-degree-ds}
There is a deterministic parallel data structure that, given a simple bipartite graph
\(G=(V_L\cup V_R,E)\) in compact representation with \(|V_L|,|V_R|\le n\), supports
the following operations.

\begin{itemize}
    \item \textbf{Batched updates.} Remove an arbitrary batch of vertices, together with all
    incident edges. Each vertex is removed at most once.
    
    \item \textbf{Large-degree subgraph query.} Let the current maximum degree on the left side
    be \(\Delta\), and let \(k\) be such that \(\Delta\in[2^k,2^{k+1})\).
    The data structure outputs the induced subgraph on the set of left vertices of current
    degree at least \(2^k\), together with all of their neighbors on the right, in compact
    representation. The work of one query is \(O(|V'|+|E'|+R)\), where \(V'\) and \(E'\) are the vertex and edge sets of the output subgraph, and \(R\)
    is a redundancy term that sums to at most $O(|V|+|E|)$ among all queries. 
\end{itemize}

Moreover, the total work of all batched updates is \(O(|V|+|E|)\), and every operation has
\(poly(\log n)\) depth.
\end{lemma}
\begin{proof}
We may assume without loss of generality that batched updates and queries alternate, since consecutive update batches can be merged. Also, a query has no external parameter: it always asks for the current largest degree scale on the left. We index the left buckets so that \(L_i\) corresponds to the degree range \([2^{i-1},2^i)\). Thus, if the first phase whose cleaned left bucket remains nonempty is \(h\), then the query threshold is \(2^{h-1}\).

\paragraph{Stored data.}
For each \(u\in V_L\), we keep a compact stored adjacency list \(A(u)\), initially equal to the full adjacency list of \(u\), and we place \(u\) into the unique bucket \(L_i\) with \(2^{i-1}\le d_0(u)\le 2^i-1\). We do not update the bucket placement eagerly under deletions. The invariant we maintain on the left is that whenever \(u\in L_j\), we have \(|A(u)|\le 2^j-1\).

For each \(v\in V_R\), we partition its neighbors by the \emph{initial} left bucket:
\(B_i(v):=\{u\in N(v): u \text{ initially belongs to } L_i\}\).
All these blocks are built once in \(O(|V|+|E|)\) work and \(\poly(\log n)\) depth using \cref{lem:logn-subgraph-splitting}. We also keep:
\begin{itemize}
    \item a bit-mask \(M(v)\), whose \(i\)-th bit is \(1\) iff \(B_i(v)\) is still unopened;
    \item a current candidate list \(C(v)\);
    \item a deferred list \(D(v)\);
    \item a target level \(p(v)\).
\end{itemize}
Initially \(C(v)=D(v)=\emptyset\), \(M(v)\) marks exactly the nonempty blocks \(B_i(v)\), and \(p(v)\) is the largest index \(i\) with \(B_i(v)\neq\emptyset\). We store the right vertices in buckets \(R_i:=\{v:p(v)=i\}\), built by one \(\log n\)-sorting step.

\paragraph{Batched updates.}
When a batch of vertices is removed, we only set deleted bits for these vertices. We do not eagerly update any stored list. Instead, whenever a list is scanned later, all deleted vertices and incident edges are filtered out at that time. Since each vertex is removed at most once, the total work of all batched updates is \(O(|V|)\).

\paragraph{Phase invariant.}
A query proceeds in descending phases. At the \emph{beginning} of phase \(i\), the following hold.

\begin{enumerate}
    \item \(L_i\) is the largest nonempty left bucket.
    \item For every \(j\) with \(1\le j\le i\) and every \(u\in L_j\), we have \(|A(u)|\le 2^j-1\).
    \item Every right vertex that may still have an unreported live neighbor in \(L_i\) lies in \(R_i\).
    \item For every \(v\in R_i\), we have \(D(v)=\emptyset\). The list \(C(v)\) consists of exactly the already activated neighbors of \(v\) that still need to be tracked in the current and future phases, up to stale entries that will be discarded when \(C(v)\) is next scanned. Every not-yet-activated live neighbor of \(v\) lies in some unopened block \(B_t(v)\) with \(t\le i\). In particular, every live neighbor of \(v\) that lies in \(L_i\) is either already in \(C(v)\) or belongs to \(B_i(v)\) with the \(i\)-th bit of \(M(v)\) set.
\end{enumerate}

These properties hold initially by construction.

\paragraph{Cleaning one left phase.}
Let \(i\) be the largest index such that \(L_i\) is nonempty. We first clean \(L_i\). For each \(u\in L_i\), we scan \(A(u)\), discard deleted right neighbors, compact the surviving entries, and let \(d(u)\) be the resulting current degree. If \(d(u)\ge 2^{i-1}\), then \(u\) stays in \(L_i\). If \(1\le d(u)\le 2^{i-1}-1\), then \(u\) is deferred by exactly one level, namely to \(L_{i-1}\). If \(d(u)=0\), then \(u\) is removed. Since all deferred vertices go to the single bucket \(L_{i-1}\), this redistribution is done by one prefix-sum/compaction step.

If the cleaned bucket \(L_i\) remains nonempty, then this is the phase at which the query stops. Otherwise, the query continues to the next lower phase.

\paragraph{Processing the corresponding right phase.}
We now process exactly the vertices in \(R_i\). Fix \(v\in R_i\). If the \(i\)-th bit of \(M(v)\) is set, we activate \(B_i(v)\) by appending it to \(C(v)\), and then clear that bit.

We then scan the current list \(C(v)\). While scanning, we discard all deleted left neighbors. Every remaining live edge \((u,v)\) in the current list is examined and handled as follows:
\begin{itemize}
    \item if \(u\in L_i\), then \((u,v)\) is reported as a witness edge of the current phase;
    \item otherwise, \(u\) lies in a lower left bucket, so we move \((u,v)\) from \(C(v)\) to \(D(v)\).
\end{itemize}
Thus, after the scan, the edges of \(C(v)\) whose left endpoint lies in a lower bucket have been moved to \(D(v)\). The edges from \(v\) into \(L_i\) are exactly the witness edges of the current phase; if the query continues, then these witness edges are empty because the cleaned bucket \(L_i\) is empty.

\paragraph{Rebuilding the right buckets.}
If the query stops at phase \(i\), then no further right-side rebuilding is needed. Otherwise the cleaned bucket \(L_i\) is empty, and we must prepare the state for the next phase.

There are two cases.

If \(D(v)\neq\emptyset\), then \(v\) must remain active immediately at the next lower level. In this case we set \(C(v):=D(v)\), \(D(v):=\emptyset\), and \(p(v):=i-1\). Thus \(v\) is deferred by exactly one level.

If \(D(v)=\emptyset\), then the already activated part of \(v\)'s neighborhood has been completely exhausted. Hence any future witness of \(v\) must come from an unopened block. Let \(j<i\) be the largest index such that the \(j\)-th bit of \(M(v)\) is still set. If such a \(j\) exists, we set \(C(v):=\emptyset\) and \(p(v):=j\); that is, \(v\) jumps directly to the largest unopened level below \(i\). If no such \(j\) exists, then \(v\) has no remaining possible future witness and is removed.

After all vertices of \(R_i\) are processed, we rebuild the right buckets by one \(\log n\)-sorting step according to the new target levels \(p(v)\). This restores the phase invariant at the beginning of the next phase.

\paragraph{Why the jump is correct.}
The only subtle point is the second case above. If \(D(v)=\emptyset\), then after processing phase \(i\) there is no already activated live neighbor of \(v\) left below level \(i\). Therefore the only possible future witnesses of \(v\) must come from unopened blocks. If \(j<i\) is the largest unopened level, then every live neighbor still hidden in an unopened block \(B_t(v)\) satisfies \(t\le j\), and since left buckets only move downward, such a neighbor can never contribute to any phase above \(j\). Hence \(v\) cannot become relevant again before phase \(j\), so jumping directly to \(R_j\) is correct.

\paragraph{Correctness of the query.}
Let \(h\) be the first phase such that the cleaned bucket \(L_h\) remains nonempty. After cleaning phase \(h\), every vertex still in \(L_h\) has current degree at least \(2^{h-1}\). On the other hand, every vertex in a lower bucket \(L_j\) with \(j<h\) has degree at most \(2^j-1\le 2^{h-1}-1\) by the left-side invariant. Therefore the queried left side is exactly \(V_L':=L_h\).

By the right-side invariant, every right vertex that still has a live neighbor in \(L_h\) lies in \(R_h\). For every \(v\in R_h\), every live neighbor of \(v\) in \(L_h\) is either already in \(C(v)\) or lies in the newly activated block \(B_h(v)\), and both are scanned in phase \(h\). Therefore the witness edges reported in phase \(h\) are exactly the edges between \(V_L'\) and their right neighbors. Let \(V_R'\) be the set of right vertices incident to at least one witness edge in phase \(h\). Then \(V_R'\) is exactly the right side of the desired induced bipartite subgraph.

The compact representation of the output subgraph is now immediate. On the right, we use the witness lists produced in phase \(h\). On the left, after cleaning phase \(h\), every \(u\in V_L'\) has \(A(u)\) equal to its exact current live neighborhood. Every such live neighbor lies in \(V_R'\), since it is adjacent to a vertex of \(V_L'\). Hence the lists \(A(u)\) for \(u\in V_L'\) are exactly the left adjacency lists of the output subgraph. A standard relabeling/compaction step constructs the final compact representation in \(O(|V'|+|E'|)\) additional work and \(\poly(\log n)\) depth.

\paragraph{Work bound.}
The total update cost is \(O(|V|)\).

For the query cost, the only non-output work is the redundancy term \(R\). On the left, whenever a vertex \(u\) drops from \(L_i\) to \(L_{i-1}\), every later non-output rescan of an edge incident to \(u\) is charged to this one-level drop. After the drop, \(u\) has at most \(2^{i-1}-1\) remaining live neighbors, so the total left-side redundancy charged to \(u\) over the whole sequence is \(O(d_0(u))\), and summing over all \(u\in V_L\) gives \(O(|E|)\).

On the right, an edge \((u,v)\) contributes redundancy in two ways. First, it may be scanned once when its block \(B_i(v)\) is activated; this is charged once to the edge itself. Second, after activation, it may survive from one phase to the next through the transition \(C(v)\to D(v)\to C(v)\). Every such future rescan is charged to a one-level drop of the left endpoint \(u\): an edge can continue to survive on the right only while the current left bucket of \(u\) keeps decreasing. Therefore the total right-side redundancy charged to a fixed \(u\) is \(O(2^{i_0}+2^{i_0-1}+2^{i_0-2}+\cdots)=O(d_0(u))\), where \(i_0\) is the initial bucket of \(u\). Summing over all \(u\in V_L\), the total right-side redundancy is \(O(|E|)\).

Hence each query costs \(O(|V'|+|E'|+R)\), where the sum of all redundancy terms \(R\) over the whole sequence of operations is \(O(|V|+|E|)\). All subroutines are scans, compaction steps, prefix sums, and \(\log n\)-sorting on \([\log n]\)-valued keys, so every operation has \(\poly(\log n)\) depth.
\end{proof}
\begin{remark}
We can apply the same structure to an undirected graph in order to find all vertices in the highest degree class. To do so, we replace the undirected graph by a bipartite graph with two copies of the vertex set, \(V_L\) and \(V_R\). Every undirected edge \(\{u,v\}\) is represented by the two bipartite edges \((u_L,v_R)\) and \((v_L,u_R)\). Thus the left degree of \(u_L\) is exactly the degree of \(u\) in the original graph, and querying the highest left degree class returns precisely the vertices of highest degree in the undirected graph. Likewise, deleting a vertex \(v\) in the original graph corresponds to deleting both copies \(v_L\) and \(v_R\).
\end{remark}
\section{Defective Coloring}

\subsection{Definition and Results}

\begin{definition}[Weighted defective coloring]
Let $G=(V,E)$ be a graph with $|V|\le n$ and positive edge weights, and let $S=\sum_{(u,v)\in E} w(u,v)$. An $(N,\epsilon)$ defective coloring is a mapping $f:V\to [N]$ such that $\sum_{(u,v)\in E,\ f(u)=f(v)} w(u,v)\le \epsilon S$. That is, the total weight of defective edges (edges whose endpoints receive the same color) is at most an $\epsilon$-fraction of the total weight.
\end{definition}

\begin{restatable}[Near Optimal Defective Coloring]{theorem}{NearOptimalDefectiveColoring}\label{thm:near-optimal-defective-coloring}
Consider $\epsilon \in (0,1]$. For every graph $G=(V,E)$ with $|V|\le n$ given in compact representation, there exists an algorithm with $O(n+m)$ work and $poly(\log n)$ depth that computes a $\left(\left\lceil \frac{2}{\epsilon}\right\rceil,\epsilon\right)$ defective coloring of $G$.
\end{restatable}

\begin{remark}
These parameters are optimal up to constant factors. For example, in a complete graph with equal weights, any coloring using $k$ colors incurs a defect of at least approximately $1/k$.
\end{remark}

\begin{remark}
All results in this section extend to multigraphs. That is, there may be multiple edges between the same pair of endpoints, and at no point does the algorithm require these parallel edges to be merged.
\end{remark}

\subsection{Method Overview and Parameter Mixing}

In this section, we show how to compute a near-optimal defective coloring in linear work. We restate the theorem here.
\NearOptimalDefectiveColoring*

We approach this through multiple steps. The most important step is \cref{lem:almost_optimal_polylog}, which computes a $\left(2,\frac12+\frac{1}{10\log^2 n}\right)$ defective coloring, and then iterates it to obtain the general parameters. We proceed in several steps, gradually adjusting the parameters to reach the claimed guarantee.

\Cref{lem:linial} resembles one round of Linial's classical deterministic distributed coloring algorithm: we reduce a larger list of colors to a smaller list at the cost of a controlled amount of defect. In such a step, we give each color a small list of candidate colors, and choose the one that minimizes the worst-case defect, assuming that neighbors choose adversarially to maximize it.

\Cref{lem:batched-greedy} uses the previous coloring as a schedule. This allows us to use a sequential greedy algorithm: processing colors one by one, we assign colors to vertices sequentially so as to minimize defect, achieving a small depth. This introduces only an optimal amount of additional defect: we discard only a $1/c$ fraction of the non-defective edges, where $c$ is the number of new colors.

To begin, we present a lemma that will be used to bootstrap the parameters.

\begin{lemma}[Parameter mixing]\label{lem:parameter_mixing}
Suppose we have:
\begin{itemize}
    \item an algorithm $\mathcal{A}$ that, given a graph $G=(V,E)$ with $|V|\le n$ equipped with a \emph{proper} $c_1$-coloring, computes a $(c_2,\epsilon_2)$ defective coloring in $O(|V|+|E|)$ work and $poly(\log n)$ depth; and
    \item an algorithm $\mathcal{B}$ that, on any graph $G=(V,E)$ with $|V|\le n$, computes a $(c_1,\epsilon_1)$ defective coloring in $O(|V|+|E|)$ work and $poly(\log n)$ depth.
\end{itemize}
Then there is an algorithm that, on any graph $G=(V,E)$ with $|V|\le n$, computes a $\bigl(c_2,\epsilon_1+\epsilon_2(1-\epsilon_1)\bigr)$ defective coloring. Moreover, the resulting algorithm also runs in $O(|V|+|E|)$ work and $poly(\log n)$ depth.
\end{lemma}

\begin{proof}
Run $\mathcal{B}$ on the input graph $G=(V,E)$ to obtain a $c_1$-coloring $f$. Let $S=\sum_{e\in E} w(e)$ be the total edge weight, let $E_{\mathrm{bad}}=\{(u,v)\in E : f(u)=f(v)\}$ be the set of defective edges under $f$, and let $\epsilon=\frac{\sum_{e\in E_{\mathrm{bad}}} w(e)}{S}$ denote the \emph{actual} defectiveness of $f$. Then $\epsilon\le \epsilon_1$.

Delete the edges in $E_{\mathrm{bad}}$ and use subgraph selection to construct the compact representation of the remaining graph $G'=(V,E\setminus E_{\mathrm{bad}})$. By construction, $f$ is a \emph{proper} $c_1$-coloring of $G'$.

Now run $\mathcal{A}$ on $G'$ (using the proper $c_1$-coloring $f$) to obtain a $c_2$-coloring $g$. Let $S'=\sum_{e\in E\setminus E_{\mathrm{bad}}} w(e)=(1-\epsilon)S$ be the total edge weight of $G'$. Algorithm $\mathcal{A}$ guarantees that the total weight of defective edges of $g$ \emph{within $G'$} is at most $\epsilon_2 S'$. Measured as a fraction of the original total weight $S$, this contributes at most $(1-\epsilon)\epsilon_2$ additional defect.

Therefore, the total defectiveness of the final coloring (measured in $G$) is at most $\epsilon + (1-\epsilon)\epsilon_2$. Since $\epsilon\le \epsilon_1$ and $\epsilon_2\le 1$, the function $\epsilon + (1-\epsilon)\epsilon_2 = \epsilon_2 + (1-\epsilon_2)\epsilon$ is nondecreasing in $\epsilon$, and hence
\[
\epsilon + (1-\epsilon)\epsilon_2 \le \epsilon_1 + (1-\epsilon_1)\epsilon_2
= \epsilon_1 + \epsilon_2(1-\epsilon_1).
\]

For the work bound, the algorithm runs $\mathcal{B}$ on $G$, performs one subgraph-selection step to build $G'$, and then runs $\mathcal{A}$ on $G'$. Since $G'$ is a subgraph of $G$, all three steps together still use $O(|V|+|E|)$ work, and the depth remains $poly(\log n)$.
\end{proof}

\begin{remark}
As long as $\epsilon_1,\epsilon_2$ are constants strictly smaller than $1$, the resulting constant $\epsilon_1+\epsilon_2(1-\epsilon_1)$ is also strictly smaller than $1$.
\end{remark}

\subsection{Linial's Step}
\label{subsec:Linial}

This subsection shows how we mimic Linial's step of exponentially reducing the number of colors to $O(\log n)$ colors~\cite{linial1987LOCAL}, with a constant defect, in the parallel setting with linear work. Before doing that, we first apply a simpler scheme to reduce the number of colors to a small polynomial. Concretely, the next lemma is provided in the proof of Lemma~2.4 of \cite{GG25}. We use this to bring down the number of colors from $n$---which is the trivial bound implied by the numbering of the graph vertices---to $n^\epsilon$ colors for a desirably small $\epsilon$.

\begin{lemma}[Restatement of Lemma~2.4 in \cite{GG25}]\label{lem:gg25-restatement}
Suppose $G=(V,E)$ is a graph with $|V|\le n$ and a proper $k$-coloring. Then we can compute a \(\bigl(3k^{2/3},\,6k^{-1/3}\bigr)\) defective coloring of $G$ in $O(|V|+|E|)$ work and $poly(\log n)$ depth. The algorithm requires $O(n)$ preprocessing. As a result, for any fixed $\epsilon>0$, we obtain a $(O(n^{\epsilon}),\epsilon)$ defective coloring algorithm for $G$ with $O(|V|+|E|)$ work and $poly(\log n)$ depth.
\end{lemma}

\begin{proof}
We iterate the \(\bigl(3k^{2/3},6k^{-1/3}\bigr)\) defective-coloring algorithm a few times, each time removing the defective edges and feeding the remaining graph into the oracle with the new value of $k$ equal to the current number of colors. In each iteration, $k$ decreases by a polynomial factor, so after a constant number of rounds it drops to $O(n^\epsilon)$. The work and depth bounds follow, and \cref{lem:parameter_mixing} shows that the total defect is still $o(1)$.
\end{proof}

The above two steps are used only to obtain the combinatorial structure needed for the key step below. For this structure, we use a deterministic method to obtain combinatorial objects with concentration guarantees similar to those of randomized constructions.

\begin{definition}[$(k,l,t)$ combinatorial design]
A $(k,l,t)$ combinatorial design is a set family
\[
S_1,S_2,\cdots,S_N \subseteq [k],
\]
such that
\begin{itemize}
    \item $|S_i|=l$ for all $i$, and
    \item $|S_i \cap S_j| \le t$ for all $i \neq j$.
\end{itemize}
The number $N$ is called the size of the design.
\end{definition}

\begin{lemma}\label{lem:comb_design}
There exists a universal constant $c>0$ such that for each sufficiently large $k$ that is a multiple of $100$, there exists a $(k,\frac{k}{100},\frac{k}{5000})$ combinatorial design of size at least $2^{ck}$, and such a design can be found in $O(4^k\,poly(k))$ work and $poly(k)$ depth on a CREW PRAM with $k$-bit words. Each set is represented by a $k$-bit word, with the $i$th bit equal to $1$ if and only if the corresponding element belongs to the set.
\end{lemma}

\begin{proof}
We construct a graph $G=(V,E)$, where $V$ consists of all subsets of $[k]$ of size $\frac{k}{100}$. Thus, $|V|=\binom{k}{k/100}$. Two vertices are adjacent if and only if their intersection has size at least $\frac{k}{5000}$.

By symmetry, all vertices have the same degree. The degree of any vertex is $p(|V|-1)$, where $p$ is the probability that if we sample two sets of size $\frac{k}{100}$ uniformly at random, then their intersection exceeds $\frac{k}{5000}$. It is clear that $p\le 2^{-ck}$ for some $c>0$ by a multiplicative Chernoff bound for hypergeometric variables. As a result, any maximal independent set of this graph has size at least $2^{ck}$, since it is lower bounded by $\frac{|V|}{\Delta +1}$, where $\Delta$ is the maximum degree.

For the work bound, note that Luby gave a deterministic parallel algorithm for computing a maximal independent set in $(|V|+|E|)\,poly(\log |V|)$ work and $poly(\log |V|)$ depth; see \cite{Luby93}. A direct application of that algorithm gives the claimed work bound.
\end{proof}

With this structure, we have prepared the key tool that was obtained by an exhaustive search in Linial's algorithm. Our next step shows that one step of that algorithm can be applied in a work-efficient way, reducing the number of colors to $O(\log n)$. This step is the key tool toward the claimed bounds.

\begin{lemma}\label{lem:linial}
There exist a universal constant $c<1$ and a constant $\epsilon>0$ such that the following holds.
Suppose that a graph $G=(V,E)$ with $|V|\le n$ is given with a proper $n^{\epsilon}$-coloring.
Then we can compute a $(\log n, c)$ defective coloring of $G$ in $O(|V|+|E|)$ work and $poly(\log n)$ depth.
The algorithm requires $o(n)$ preprocessing work and $poly(\log n)$ preprocessing depth.
\end{lemma}

\begin{proof}

\paragraph{Proof overview.}
This step is essentially one step of Linial's classical algorithm, which colors a constant-degree graph using $O(\log n)$ colors. We first show how this idea applies in our defective-coloring setting. We then show how to reduce the work using bit operations, since a direct implementation of the distributed algorithm would incur an extra factor of $O(\log n)$ work when determining which candidate colors are already in use.

\paragraph{Reduce to Linial's algorithm on one vertex.}
We first apply \cref{lem:comb_design} with $k=\frac{1}{10}\log n$, rounding $k$ up to a multiple of $100$ if needed. The preprocessing work is \(O(4^k\,poly(k))=o(n)\), and the preprocessing depth is \(poly(\log n)\).

The size of the combinatorial design is at least \(2^{ck}\ge n^{0.1c}=n^\epsilon\), where we take $\epsilon=0.1c$. This allows us to map each original color injectively to a set in the design.

Now we use Linial's idea. Each vertex chooses a color from its corresponding set in the design. Suppose that this vertex corresponds to the set $S_0$, and suppose that it has $t$ neighbors with edge weights $w_1,w_2,\cdots,w_t$, corresponding to sets $S_1,S_2,\cdots,S_t$. Let $w_i'$ be the smallest power of $2$ greater than or equal to $w_i$, and let $W'=\sum_{i=1}^t w_i'$ be the total approximate weight.

For each element $u\in S_0$, define
\[
W(u)=\sum_{u \in S_i} w_i,
\qquad
W'(u)=\sum_{u \in S_i} w_i'.
\]
Notice that \(\sum_{u \in S_0}W'(u)=\sum_i w_i'|S_0 \cap S_i| \le \frac{k}{5000}W'\).

Since $|S_0|=\frac{k}{100}$, there exists some $u\in S_0$ with \(W'(u)\le \frac{1}{50}W'\). We further relax this to \(W'(u)\le 0.9W'\) to make such a $u$ easier to find. This implies
\[
W-W(u)=\sum_{u \notin S_i} w_i
\ge \frac12 \sum_{u \notin S_i} w_i'
= \frac12\bigl(W'-W'(u)\bigr)
\ge 0.05W'
\ge 0.05W,
\]
and hence \(W(u)\le 0.95W\).

Therefore, if for each vertex we find such a $u$, then we obtain a coloring with defectiveness $0.95$. So it suffices to solve the following subproblem within the required work bound.

\paragraph{Job of a single vertex.}
Let $k \leq \log n$ be an integer. Suppose that we have $t+1$ sets $S_0,S_1,S_2,\cdots,S_t$ with \(|S_0 \cap S_i| \le \frac{k}{5000}\) and \(|S_i|=\frac{k}{100}\). We also have weights $w_i$ that are powers of $2$ and are at most $n$. Let \(W=\sum_{i=1}^t w_i\), and for each \(u \in S_0\) define \(W(u)=\sum_{u \in S_i} w_i\). Then we can find a $u$ with $W(u)\le 0.9W$ in $O(t)$ work and $poly(\log n)$ depth. Note that each set is represented by a word, with the $i$th bit being $1$ if and only if $i$ belongs to the set, using the fact that the ground-set size is $O(\log n)$.

\paragraph{Completing the job by bit operations.}
Since the weights have at most $\log n$ possible values, we can use $\log n$ sorting (\cref{logn-sorting}) to classify the sets $S_i$ by weight. We use the following observation to reduce the number of sets.

\begin{observation}
Suppose that we have three sets $A,B,C$ with equal weight $w$. Then we can reduce them to one set $D$ with weight $w$ and one set $E$ with weight $2w$ in $O(1)$ work and $O(1)$ depth. The total weight of each element is preserved.
\end{observation}

\begin{proof}
For each element, we assign it to the two new sets according to whether its total number of occurrences among the three sets $A,B,C$ is $0,1,2,$ or $3$. Since this is independent across elements, this selection process can be simulated using $O(1)$ bit operations.
\end{proof}

We now start with a list of at most $t$ nonempty arrays, where each array corresponds to one weight class of sets. We may modify the weights using the above observation, but the weights never exceed $n^2$ (the total weight of each element is unchanged), so there are at most $2\log n$ possible weights. Hence, we can use a bitmap to maintain which weight classes are present and find the smallest present class in $O(1)$ work.

Each time, we inspect the smallest weight class and reduce its size in parallel using the previous observation. If we end with at most $2$ sets in that class, we put them aside. The total work is $O(t)$, since each application removes one set. The depth is $O(\log^2 n)$, since the size of the smallest class drops by a constant factor each time.

After the final reduction step, each weight appears at most twice. We choose some $u$ such that the largest weight contributing to $u$ is minimized. Let that largest weight be $w$. Then the total weight of $u$ is at most
\[
w+w+\frac12w+\frac12w+\frac14w+\frac14w+\cdots \le 4w.
\]
On the other hand, the weight of any other element is at least $w$. Since the sum of the total weights over all elements is at most \(\frac{1}{50}|S_0|W\), we get \(w \le \frac{1}{50}W\), so this $u$ satisfies the requirement.

To find such a $u$, we process the remaining sets from larger weights to smaller weights, maintain the union of those sets, and choose an element in the step at which the union becomes $S_0$.
\end{proof}

Combining this result with \cref{lem:parameter_mixing}, we get the following.

\begin{lemma}\label{lem:logn-c-defective}
For every graph $G=(V,E)$ with $|V|\le n$, we can compute a $(\log n,c)$ defective coloring of $G$ in $O(|V|+|E|)$ work and $poly(\log n)$ depth, for some universal constant $c<1$.
\end{lemma}

\subsection{Bootstrapping steps}

This step builds on Linial's step and finally achieves defective coloring with optimal parameters.

First, we note that once we have a graph colored by \(poly(\log n)\) colors, it is easy to obtain an optimal defective coloring on that graph.

\begin{lemma}[Batched Greedy]\label{lem:batched-greedy}
Let $u\leq poly(\log n)$ be an integer. For every graph $G=(V,E)$ with $|V|\le n$ that is $k$-colored, where $k=poly(\log n)$, we can compute a $(u,\frac{1}{u})$ defective coloring of $G$ in $O(|V|+|E|)$ work and $poly(\log n)$ depth.
\end{lemma}

\begin{proof}
We first use repeated $\log n$ sorting (\cref{logn-sorting}) to sort the vertices by color. We can then use the coloring as a schedule. Now we run the greedy algorithm color by color, choosing for each vertex the color that introduces the least defect.
\end{proof}

Using \cref{lem:parameter_mixing,lem:logn-c-defective,lem:batched-greedy} with \(u=2\) in \cref{lem:batched-greedy}, we get the following.

\begin{lemma}\label{lem:two-c-defective}
For every graph $G=(V,E)$ with $|V|\le n$, we can compute a $(2,c)$ defective coloring of $G$ in $O(|V|+|E|)$ work and $poly(\log n)$ depth, for some universal constant $c<1$.
\end{lemma}

We can further bootstrap this lemma to obtain the following.

\begin{lemma}\label{lem:almost_optimal_polylog}
Let $u$ be an integer. For every graph $G=(V,E)$ with $|V|\le n$, we can compute a $(u,\frac{1}{u}+\frac{1}{10\log^2 n})$ defective coloring of $G$ in $O(|V|+|E|)$ work and $poly(\log n)$ depth.
\end{lemma}

\begin{proof}
We first solve the case \((u,\frac{1}{u}+\frac{1}{11\log^2 n})\) for \(u \le 110\log^2 n\). For larger \(u\), we can simply take \(u=110\log^2 n\), which already gives the required defect.

We start with $G_0=G$. Suppose that we have $G_i$. We define $G_i'$ and $G_{i+1}$ as follows. Assign weight $1$ to every edge of $G_i$, run the \((2,c)\) defective-coloring algorithm of \cref{lem:two-c-defective} on it, and remove isolated vertices to obtain $G_i'$. Then restore the edge weights of $G_i'$ to their original values, run the same \((2,c)\) defective-coloring algorithm on it, and remove isolated vertices to obtain $G_{i+1}$. We stop at $i=k$ for some \(k<poly(\log n)\) to be determined later.

This process takes \(O(|V|+|E|)\) total work: after isolated vertices are removed, \(|V_i|\le 2|E_i|\), and the transition \(G_i\to G_i'\) makes \(|E_i|\) decay geometrically.

Putting together all colors produced on \(G_i\) and \(G_i'\), we obtain a \((4^k,c^k)\) defective coloring of \(G\). We choose the smallest \(k\) such that \(c^k < \frac{1}{11\log^2 n}\). Then \(4^k=poly(\log n)\), so \cref{lem:batched-greedy} still applies. Therefore, combining the batched-greedy result with the \((4^k,c^k)\) defective coloring using \cref{lem:parameter_mixing}, we obtain a \(\left(u,\frac{1}{u}+\frac{1}{11\log^2 n}\right)\) defective coloring.
\end{proof}

We now complete the proof of the general \(\left(\left\lceil 2/\epsilon\right\rceil,\epsilon\right)\) defective-coloring theorem within the above framework.

We may assume without loss of generality that \(\epsilon \ge 1/n\). Indeed, if \(\epsilon<1/n\), then assigning distinct colors to all vertices gives defect \(0\) using \(n\) colors, and \(n \le \left\lceil \frac{2}{\epsilon}\right\rceil\).

We are now ready to prove the main theorem, which we restate here.

\NearOptimalDefectiveColoring*

\begin{proof}
Let \(\delta := \frac{1}{10\log^2 n}\). By \cref{lem:almost_optimal_polylog}, for every integer \(u\ge 2\), we can compute a \((u,\,1/u+\delta)\) defective coloring in \(O(|V|+|E|)\) work and \(poly(\log n)\) depth.

We recurse on the defective edges. To make the total work sum correctly over all recursive calls, we use a slightly perturbed weight function for the recursion and work analysis.

\paragraph{Modified weights.}
If \(m=0\), the theorem is trivial, so assume \(m\ge 1\). Scale the original edge weights so that their total weight is \(1\). Then add \(1/(10m)\) to each edge weight, and let \(\widetilde w(e)\) denote the resulting modified weight. The total modified weight is \(\frac{11}{10}\). Since \(\widetilde w(e)\ge w(e)\) for every edge \(e\), any upper bound on defective weight measured with \(\widetilde w\) is also an upper bound for the original weights.

\paragraph{Recursive refinement and parameter recurrence.}
Suppose that after some rounds, we have a coloring of the original graph, and let \(H\) be the subgraph consisting of the edges that are still defective under this coloring. Let
\[
A=\text{(current number of colors)},\qquad
B=\text{(total \(\widetilde w\)-weight of edges in \(H\))}.
\]
Applying a \((u,1/u+\delta)\) defective-coloring algorithm to \(H\) and refining the current colors by the new colors (that is, replacing each color by a pair) changes the parameters as
\[
(A,B)\longmapsto \left(Au,\;B\left(\frac1u+\delta\right)\right).
\]
Thus, after rounds with parameters \(u_1,u_2,\dots,u_r\), we obtain
\[
A_r=\prod_{s=1}^r u_s
\]
colors and residual defective modified weight at most
\[
B_r \le B_0 \prod_{s=1}^r\left(\frac1{u_s}+\delta\right),
\]
where \(B_0\le \frac{11}{10}\).

\paragraph{Choice of parameters \(u_s\).}
We now choose the parameters \(u_s\). They are all equal to \(2\), except that one of them may be a positive integer at most \(1000\), so they are within the parameter range of \cref{lem:almost_optimal_polylog}.

Let
\[
T:=\left\lceil \frac{2}{\epsilon}\right\rceil .
\]
We choose the total number of refinement colors in the form
\[
A=i2^j
\qquad\text{with }1\le i\le 1000,
\]
and implement this using \(j\) rounds with \(u=2\) and, if \(i>1\), one final round with \(u=i\).

If \(T\le 1000\), we take \(j=0\) and \(i=T\), so \(A=T\). Otherwise, choose \(j\) such that \(\frac{T}{2^j}\in [500,1000)\), and set \(i:=\left\lfloor \frac{T}{2^j}\right\rfloor\). Then \(500\le i\le 999\), hence \(i\le 1000\), and
\[
A=i2^j\le T,
\qquad
A\ge \left(1-\frac1{500}\right)T=\frac{499}{500}T.
\]
Using at most \(T\) colors is sufficient for a \((T,\epsilon)\) defective coloring, since we may leave some colors unused. So it suffices to produce an \((A,\epsilon)\) defective coloring.

\paragraph{Defect bound.}
Since all chosen \(u_s\le 1000\), we have
\[
\prod_{s=1}^r\left(\frac1{u_s}+\delta\right)
=
\frac1A \prod_{s=1}^r (1+u_s\delta).
\]
Therefore,
\[
B_r \le \frac{11}{10}\cdot \frac1A \prod_{s=1}^r (1+u_s\delta).
\]
Also, the number of rounds is \(r\le j+1\), with \(j=O(\log T)\), and
\[
\sum_{s=1}^r u_s \le 2j + 1000 = O(\log T) = O(\log n)
\]
using \(\epsilon\ge 1/n\), hence \(T=O(n)\). Thus
\[
\prod_{s=1}^r (1+u_s\delta)
\le
\exp\!\left(\delta\sum_{s=1}^r u_s\right)
=
\exp\!\left(O\!\left(\frac{1}{\log n}\right)\right).
\]
For all sufficiently large \(n\), this factor is at most \(1.1\). Hence
\[
B_r
\le
\frac{11}{10}\cdot \frac{11}{10}\cdot \frac1A
=
\frac{121}{100}\cdot \frac1A
\le
\frac{121}{100}\cdot \frac{500}{499}\cdot \frac1T
<
\frac{2}{T}
\le \epsilon.
\]
So the final coloring has defective weight at most \(\epsilon\) with respect to the modified weights, and therefore also with respect to the original weights.

\paragraph{Work bound.}
Let \(H_r=(V_r,E_r)\) be the residual defective graph after round \(r\), after removing isolated vertices. Each recursive call costs \(O(|V_r|+|E_r|)\) work.

Every residual edge has modified weight at least \(1/(10m)\), so if the residual modified weight is \(B_r\), then \(|E_r|\le 10m\,B_r\). Moreover, in every round we use \(u\ge 2\), so the residual modified weight is multiplied by at most \(\frac12+\delta<1\) for all sufficiently large \(n\). Hence \(B_r\) decays geometrically, so \(\sum_r |E_r| = O(m)\). After removing isolated vertices, we also have \(|V_r|\le 2|E_r|\) for all \(r\ge 1\), so \(\sum_r |V_r| = O(m)\). Including the initial graph contributes an additional \(O(n+m)\) term, and the total work is \(O(n+m)\).

Finally, the number of rounds is \(O(\log T)=O(\log n)\), and each round has \(poly(\log n)\) depth, so the total depth is \(poly(\log n)\cdot O(\log(1/\epsilon))=poly(\log n)\), again using \(\epsilon \geq \frac{1}{n}\).
\end{proof}



\section{Hitting Set}
\label{sec:HittingSet}
In this section, we describe our deterministic parallel algorithm for hitting set, with linear work and polylogarithmic depth.
\subsection{Definition and Results}

\paragraph{Overview.}
We follow the gradual rounding idea of \cite{GG25}, but we organize the treatment of the different probability classes differently, and we also use our improved $\log n$ sorting (\cref{logn-sorting}) and the near-optimal defective-coloring algorithm. These allow us to reduce the work from $O\!\left((n+m)\,\operatorname{poly}(\log\log n)\right)$ to $O(n+m)$. We reproduce parts of the proofs from \cite{GG25} to keep the presentation self-contained.

\begin{definition}[Pre-hitting-set instance]
A pre-hitting-set instance consists of a bipartite graph $G=(V_L \cup V_R,E)$.
Each left node $u\in V_L$ has a sampling probability $p_u=2^{-k_u}$, where
$k_u\in[2\log n]$.
For each right node $v\in V_R$, define its expected hit count
\[
c_v \coloneqq \sum_{u\in N(v)} p_u,
\]
where $N(v)\coloneqq \{u\in V_L : (u,v)\in E\}$.
Each right node $v\in V_R$ also has a weight $w_v>0$.
\end{definition}

\begin{definition}[Hitting set and happiness]
A hitting set is a subset $H\subseteq V_L$.
For each $v\in V_R$, let $h_v\coloneqq |N(v)\cap H|$.
We say that $v$ is \emph{happy} if
$\frac{1}{2}\,(c_v-1) < h_v \;\le\; C\,(c_v+1),$
for a sufficiently large absolute constant $C$.
\end{definition}

\begin{remark}
We briefly explain why the range $(\frac{1}{2}(c_v-1),\, C(c_v+1)]$ is used. The true expected hit count is $c_v$, and the left endpoint reflects that we allow zero hits when the expected hit count is small. The right endpoint simply reflects that we allow the hit count to increase by both an additive and a multiplicative constant. Moreover, the left endpoint being strict ensures that a happy right node is hit at least once when the expected hit count is at least $1$.
\end{remark}

The main theorem we prove is the following.

\begin{theorem}\label{thm:main-hitting-set}
Fix a constant \(\epsilon>0\). There exists an algorithm that, given a pre-hitting-set
instance \(G=(V_L\cup V_R,E)\) in compact representation, computes a hitting set
\(H\subseteq V_L\) such that the total weight of happy right nodes is at least $(1-\epsilon)\sum_{v\in V_R} w_v$.
The algorithm uses \(O(|V|+|E|)\) work and \(\poly(\log n)\) depth, where \(|V|\le n\).

Here, compact representation means that every vertex in both \(V_L\) and \(V_R\)
stores a compact list of all its neighbors.
\end{theorem}

\begin{definition}[Probability weight and distribution distance]
A pre-hitting-set instance consists of a bipartite graph $G=(V_L\cup V_R,E)$,
left probabilities $(p_u)_{u\in V_L}$, and right weights $(w_v)_{v\in V_R}$.
For each $v\in V_R$, recall
\[
c_v = \sum_{u\in N_G(v)} p_u.
\]
Define the probability weight
\[
W_G \coloneqq \sum_{v\in V_R} w_v\,c_v.
\]
and the total weight
\[
\hat{W}_G \coloneqq \sum_{v\in V_R} w_v.
\]
For two instances $G$ and $G'$ on the same vertex sets $(V_L,V_R)$ with the same
right weights $(w_v)_{v\in V_R}$ (but possibly different edge sets and left
probabilities), define the distribution distance
\[
d(G,G') \coloneqq \sum_{v\in V_R} w_v\,|c_v-c_v'|.
\]
\end{definition}

\begin{observation}\label{obs:distance-properties}
The distance $d(\cdot,\cdot)$ satisfies:
\begin{itemize}
\item (Triangle inequality.) If $G_1,G_2,G_3$ share the same $(V_L,V_R)$ and right weights, then
$d(G_1,G_3) \le d(G_1,G_2)+d(G_2,G_3).$
\item (Subadditivity under edge-disjoint union.)
Assume $G_1$ and $G_2$ have the same $(V_L,V_R)$, the same left probabilities, and the same right weights,
and their edge sets are disjoint; define $G_1\cup G_2$ by unioning edge sets (and keeping the common probabilities/weights).
Make the analogous assumptions for $G_3,G_4$.
Then
$d(G_1\cup G_2,\, G_3\cup G_4) \le d(G_1,G_3)+d(G_2,G_4).$
\end{itemize}
\end{observation}

Using the distance formulation, we also obtain a distance version of the hitting-set theorem.
At the beginning of the next section, we will show that the main theorem follows from it as a corollary.

\begin{theorem}[Distance Version of Hitting Set]\label{thm:distance_hitting_set}
Let \(\epsilon>0\) be a constant. Then there exists a constant \(c_0\) such that the
following holds.

There exists an algorithm that, given a pre-hitting-set instance
\(G=(V_L\cup V_R,E)\) in compact representation, computes in \(O(|V|+|E|)\) work
and \(\poly(\log n)\) depth a new pre-hitting-set
instance \(G'\) on the same bipartite graph and with the same right-node weights,
such that:
\begin{itemize}
    \item every left probability in \(G'\) is either \(0\) or at least \(2^{-c_0}\); and
    \item
    $d(G,G') \le \epsilon (W_G+\hat{W}_G).$
\end{itemize}
\end{theorem}
\subsection{Main reductions}

We first show how finding a hitting set can be reduced to finding a pre-hitting-set instance whose probabilities are lower bounded by a constant and whose distance from the original distribution is relatively small.

\begin{observation}[Reducing hitting sets to distance]\label{obs:hitting-to-distance}
Fix a constant $c_0$ and let $p_0 \coloneqq 2^{-c_0}$.
By deterministically including all left nodes with probability at least $p_0$
(and increasing the constant $C$ by a factor of $2^{c_0}$), we may assume that all left probabilities
satisfy $p_u \le p_0$.

Reweight the right nodes by
$w'_v \coloneqq \frac{w_v}{c_v+1},$
and let
\[
W \coloneqq \sum_{v\in V_R} w'_v(c_v+1)=\sum_{v\in V_R} w_v.
\]
For a set $H\subseteq V_L$, define an instance $G_H$ by setting the left probabilities to
$p_u=p_0$ for $u\in H$ and $p_u=0$ otherwise, and let $c_v(H)$ denote the resulting expected hit count.

Assume that $C$ is chosen so that $p_0 C \ge 2$.
If
$d(G,G_H) \le \varepsilon W,$
then $H$ makes at least a $(1-2\varepsilon)$-fraction of the total right weight happy.
\end{observation}

\begin{proof}
Let $h_v \coloneqq |N(v)\cap H|$. In $G_H$ we have $c_v(H)=p_0 h_v$.

If $v$ is unhappy, then either

\smallskip
\noindent
(i) \(h_v \leq  \frac{1}{2}(c_v-1)\), in which case
$c_v(H)=p_0 h_v \le h_v \leq \tfrac{1}{2}(c_v-1),$
so
$|c_v-c_v(H)| \ge c_v-\tfrac{1}{2}(c_v-1) = \tfrac{1}{2}(c_v+1);$

\smallskip
\noindent
or

\smallskip
\noindent
(ii) \(h_v > C(c_v+1)\), in which case
$c_v(H)=p_0 h_v > p_0 C(c_v+1)\ge 2(c_v+1),$
so
$|c_v(H)-c_v| \ge 2(c_v+1)-c_v = c_v+2 \ge \tfrac{1}{2}(c_v+1).$

Thus, every unhappy $v$ satisfies
$|c_v-c_v(H)|\ge \tfrac{1}{2}(c_v+1),$
and therefore contributes at least
$w'_v\cdot \tfrac{1}{2}(c_v+1)=\tfrac{1}{2}w_v$
to $d(G,G_H)$.
Hence
\[
\sum_{\text{$v$ unhappy}} w_v
\le
2\,d(G,G_H)
\le
2\varepsilon W
=
2\varepsilon \sum_{v\in V_R} w_v,
\]
so the happy weight is at least
\[
(1-2\varepsilon)\sum_{v\in V_R} w_v.
\qedhere
\]
\end{proof}

\paragraph{Overview of the job of one probability class.}
Using $\log n$ sorting (\cref{logn-sorting}) with minor modifications, we bucket left nodes by their exponent $k_u$
and split $G$ into $O(\log n)$ edge-disjoint subgraphs $G_i$, where in each $G_i$
all remaining left probabilities equal $2^{-i}$.
When we process the classes separately, we use subadditivity to argue that the total distance introduced is at most the sum of the distances introduced in the individual classes.

Fix one such subgraph $G_i$. Our goal is to construct a set $H_i\subseteq V_L(G_i)$ such that the induced instance $(G_i)_{H_i}$ satisfies $d\bigl(G_i,(G_i)_{H_i}\bigr) \le \epsilon\,W(G_i)+2^{-ci}\hat{W}(G)$, where $c>0$ is some constant.
The resulting instance will have left probabilities equal to $2^{-\lfloor i/2\rfloor}$, and its number of edges will be at most $2^{-ci}|E(G)|$.
We then combine the selected left nodes over all probability classes, and the triangle inequality guarantees that the total distance remains bounded.
After that, we perform a cleanup step on the rounded-up and shrunken instance. 
Finally, \cref{obs:hitting-to-distance} implies that this makes almost all right weight in $G$ happy.

\Cref{lem:graph-splitting} shows how to actually split the graph in linear work.

\begin{lemma}[Graph splitting]\label{lem:graph-splitting}
Let $G=(V_L\cup V_R,E)$ be a bipartite graph.
Each left vertex $u\in V_L$ has a key $\kappa(u)\in[\tau]$, where $\tau=O(\log n)$.
There is an algorithm that uses $O(|V|+|E|)$ work and $\poly(\log n)$ depth (where $|V| \leq n$) and outputs $\tau$ subgraphs
$G^{(1)},\dots,G^{(\tau)}$ in compact form, where for each $i\in[\tau]$,
$V_L^{(i)}=\{u\in V_L:\kappa(u)=i\},\qquad E^{(i)}=\{(u,v)\in E:u\in V_L^{(i)}\},$
and
$V_R^{(i)}=\{v\in V_R:\exists\,u\in V_L^{(i)}\text{ with }(u,v)\in E\}.$
\end{lemma}

\begin{proof}
First, sort the left vertices by key using $\log n$ sorting (\cref{logn-sorting}), and compute the $\tau$ bucket boundaries
(by a histogram and prefix sums over keys). This yields the compact lists $V_L^{(i)}$.

Next, bucket the edges by the key of their left endpoint, again using $\log n$ sorting (\cref{logn-sorting}).
This produces $\tau$ edge lists $E^{(i)}$ (still referencing original vertex IDs).

It remains to build $V_R^{(i)}$ and remap right endpoints to the compact IDs.
For each right vertex $v\in V_R$, scan its incident edges once and determine the set of keys
$K(v)\coloneqq \{\kappa(u):(u,v)\in E\}.$
Since $\tau=O(\log n)$, this can be maintained as a $\tau$-bit mask with $O(1)$ work per incident edge.
For each $i\in K(v)$, output one pair $(i,v)$.
The total number of such pairs is at most $|E|$ (each pair witnesses at least one edge). To put all the pairs into one compact list, we can count how many $v$ correspond to each $i$, and do a prefix sum to put them in the correct place for each $i$. 

Now apply $\log n$ sorting (\cref{logn-sorting}) to the pairs $(i,v)$ by their first component $i$. 
Within each bucket $i$, scan the resulting list to assign compact IDs to its right vertices,
thereby forming $V_R^{(i)}$.
Finally, scan $E^{(i)}$ and replace each right endpoint $v$ by its assigned compact ID in $V_R^{(i)}$.
All steps use $O(|V|+|E|)$ work, plus $poly(\log n)$ overhead.
\end{proof}

\subsection{$\tfrac12$ sampling methods}
\label{subsec:12-sampling}
This section shows how to implement the $\tfrac12$-sampling lemma from \cite{GG25} with linear work. The lemma deterministically selects half of the nodes on the left and doubles their probabilities, while maintaining a small distance from the original pre-hitting set.
\begin{lemma}[Approximate quadratic optimization]\label{lem:aqo}
Let $G=(V,E)$ be a (weighted) graph with vertex coefficients $c_v\in\mathbb{R}$ for all $v\in V$
and edge weights $w_e>0$ for all $e\in E$.
Let $(p_v)_{v\in V}$ be arbitrary probabilities with $p_v\in[0,1]$.
Fix any constant $c>0$, and assume $\varepsilon \ge 1/\log^c n$.
Then there is a deterministic parallel algorithm with $O(|V|+|E|)$ work and $\poly(\log n)$ depth that outputs a set
$U\subseteq V$ such that
\[
\sum_{u\in U} c_u \;-\!\!\sum_{e\in E[U]} w_e
\;\ge\;
\sum_{v\in V} p_v c_v \;-\!\!\sum_{e=\{u,v\}\in E} p_u p_v\, w_e
\;-\; \varepsilon \sum_{e\in E} w_e,
\]
where $E[U]:=\{\{u,v\}\in E : u,v\in U\}$.
\end{lemma}

\begin{proof}
Compute an $(N,\varepsilon)$ defective coloring of $G$ with
$N=\left\lceil 2/\varepsilon\right\rceil$, and let $E_{\mathrm{def}}$
be the set of defective edges (edges whose endpoints have the same color).
By definition, $\sum_{e\in E_{\mathrm{def}}} w_e \le \varepsilon \sum_{e\in E} w_e$.

Consider the modified objective that ignores defective edges:
\[
F(U)\;:=\;\sum_{u\in U} c_u \;-\!\!\sum_{e\in E[U]\setminus E_{\mathrm{def}}} w_e.
\]
Let $X$ be the random set obtained by independently including each $v$ with probability $p_v$.
Then
\[
\mathbb{E}[F(X)]
=
\sum_{v\in V} p_v c_v
-
\sum_{e=\{u,v\}\in E\setminus E_{\mathrm{def}}} p_u p_v\, w_e
\;\ge\;
\sum_{v\in V} p_v c_v
-
\sum_{e=\{u,v\}\in E} p_u p_v\, w_e
.
\]
Since $\sum_{e\in E_{\mathrm{def}}} w_e \le \varepsilon \sum_{e\in E} w_e$, it suffices to find
$U$ with $F(U)\ge \mathbb{E}[F(X)]$. Indeed, after finding such $X$, the defective edges will further cost at most $\sum_{e\in E_{def}}w_e \leq \epsilon \sum_{e\in E} w_e$.

We apply the method of conditional expectations, fixing the variables color class by color class.
Because $E_{\mathrm{def}}$ contains all edges inside a color class, the function $F(\cdot)$ is
linear in the variables of the current color class once previous colors are fixed, so we can decide
all vertices of that color in parallel without decreasing the conditional expectation.
This takes $N=O(1/\varepsilon)=poly(\log n)$ rounds and $O(|V|+|E|)$ total work.
The resulting set $U$ satisfies $F(U)\ge \mathbb{E}[F(X)]$, and thus the stated inequality.
\end{proof}

\begin{remark}
Unlike most previous derandomization methods, for example, those based on pairwise independence, our lemma requires the graph $G$ to be built explicitly. That is, it is not sufficient to keep $G$ in an implicit sum-of-squares form and claim a work bound in terms of the size of that representation alone. Furthermore, when constructing $G$, we must ensure that no extra $\log\log n$ factor is introduced by sorting.
\end{remark}
We next improve the work bound of the $\tfrac12$-sampling lemma from \cite{GG25} using \cref{lem:aqo}. The application of \cref{lem:aqo} here requires an auxiliary graph whose construction is not completely straightforward. \Cref{lem:near-independent-bundling} explains how to construct such a graph.

\begin{lemma}[Near-Independent Bundling]\label{lem:near-independent-bundling}
Suppose that we are given positive integers $d_1,d_2,\dots,d_k$ and a parameter $B=\operatorname{poly}(\log n)$.
Then there exists a deterministic parallel algorithm that forms groups with the following properties:
\begin{itemize}
    \item Each group has size in $[B,2B]$, and the multiplicity of each index in each group is at most $2$.
    \item For every $i$, the index $i$ appears at most $d_i$ times among all groups.
    \item For all but at most $B$ indices $i$, the index $i$ appears exactly $d_i$ times among all groups.
\end{itemize}
Let the groups be $S_1,\dots,S_s$. Then, we can construct a compact representation of the auxiliary graph corresponding to the approximate quadratic optimization objective
\[
\sum_j \left(x_j-\frac{|S_j|}{2}\right)^2,
\]
where $x_j$ denotes the number of chosen elements in $S_j$.
Moreover, the algorithm runs in
\[
O\!\left(B\sum_i d_i\right)
\]
work and $\operatorname{poly}(\log n)$ depth.
\end{lemma}

\begin{proof}
We first use $\log n$ sorting (\cref{logn-sorting}) to sort the indices $i$ by the value of $\lfloor \log_2 d_i \rfloor$.

Fix one degree class corresponding to degrees in $[2^t,2^{t+1})$. Split the indices in this class into consecutive blocks of size exactly $B$, leaving at most $B-1$ leftover indices from this class.

Now consider one such full block. We create exactly $2^t$ groups from this block. For each index $i$ in the block, write
$d_i = 2^t + r_i, \qquad 0 \le r_i < 2^t.$
We place one copy of $i$ into every one of the $2^t$ groups, and then place a second copy of $i$ into exactly $r_i$ of these groups. This is possible because $r_i<2^t$. Hence, the total number of copies of $i$ across these groups is exactly
$2^t+r_i=d_i.$
Also, in every group, each index appears with multiplicity at most $2$.

Since the block contains exactly $B$ indices and each contributes either one or two copies to each group, every group produced from this block has size in $[B,2B]$.

We now explain how to construct the auxiliary graph for the quadratic objective. Expanding one term $\left(x_j-\frac{|S_j|}{2}\right)^2$ produces linear terms and quadratic terms involving only pairs of indices that occur in the same group $S_j$. Since every group is supported on a single block of size $B$, all auxiliary edges produced by this group lie inside that block. Thus, for one block, the auxiliary graph consists of the union of $2^t$ clique-like gadgets on $B$ vertices, and therefore has $O(2^t B^2)$ edges in total. Because every index in this degree class has degree $\Theta(2^t)$, summing over all blocks in the class gives
\[
O\!\left(B\sum_{i \text{ in the class}} d_i\right)
\]
work to construct the auxiliary graph for that class. Summing over all degree classes yields
\[
O\!\left(B\sum_i d_i\right)
\]
total work.

It remains to handle the leftovers. There are at most $(B-1)\log n=\operatorname{poly}(\log n)$ leftover indices in total. On this set, we may simply use a sequential greedy cleanup. Repeatedly choose any $B$ active leftover indices of the smallest current residual demand, let
$\delta$ 
be the smallest residual demand among them, and create $\delta$ groups, each containing one copy of each of the chosen $B$ indices. Each such group has size exactly $B$, and every chosen index loses exactly $\delta$ units of residual demand. Therefore, at least one of the chosen indices becomes exhausted and disappears. Hence, the number of active leftover indices strictly decreases in each cleanup step, so the process terminates with at most $B-1$ active indices remaining.

Throughout the cleanup, every group has size exactly $B$, and each index appears at most once in each such group. Since the cleanup involves only $\operatorname{poly}(\log n)$ indices, its total depth is $\operatorname{poly}(\log n)$, and its total work is absorbed in the bound above.

This proves all three properties of the groups and the claimed construction bound for the auxiliary graph.
\end{proof}

\begin{lemma}[$\tfrac{1}{2}$ sampling]\label{lem:half-sampling}
There exists a universal constant $C$ such that the following holds. Let $G=(V_L\cup V_R,E)$ be a bipartite instance in which each right node $v\in V_R$ has weight $w_v>0$
and each left node has the same probability $p< \frac 12$.
Let $\gamma\in\left[\frac{1}{\log^4 n},\,\epsilon_0\right]$ be an error-related parameter, where $\epsilon_0$ is a desirably small constant. 
Then there is an algorithm that outputs a subset $S\subseteq V_L$ in $O\!\left((|V|+|E|)/\gamma\right)$ work and $\poly(\log n)$ depth such that, letting $G'$ be the instance obtained by setting
the left probabilities to $2p$ on $S$ and to $0$ on $V_L\setminus S$,
\begin{itemize}
\item
$d(G,G') \le C(\sqrt{\gamma}\,W(G)+\frac{p}{\gamma}\hat W(G)),$
\item
$|E(G')| \in \left[ \left(\frac12-C\sqrt{\gamma}\right)|E(G)|-\frac{C}{\gamma}\Delta,\, \left(\frac12+C\sqrt{\gamma}\right)|E(G)| \right],$
where $\Delta$ is the maximum degree on the left.
\end{itemize}
As a corollary, we have $|E(G')|\le \frac23\,|E(G)|$.
\end{lemma}

\begin{proof}
Let $B:=\lceil 1/\gamma\rceil$. Since $\gamma\le 0.01$, we have $B\ge 100$.

For each $v\in V_R$, partition $N(v)$ into disjoint groups of size exactly $B$, leaving at most
$B-1$ leftover neighbors. Let the resulting groups be $S_1,\dots,S_k$, where each group $S_i$
comes from a single right node and inherits its weight $w_i$.
For a subset $S\subseteq V_L$, define $x_i:=|S\cap S_i|$.

To control the number of edges, apply \cref{lem:near-independent-bundling} to the left degrees
$d_u\qquad (u\in V_L),$
with the same parameter $B$. This produces groups $T_1,\dots,T_\ell$ such that all but at most $B$
left vertices are represented exactly according to their degree demand. 
For the degree-based groups, we allow repeated occurrences of the same left vertex inside one group. 
Thus, each \(T_i\) is a multiset of left vertices rather than a set. 
For a chosen subset \(S\subseteq V_L\), define
\[
y_i \coloneqq \sum_{u\in S} m_i(u),
\]
where \(m_i(u)\in\{0,1,2\}\) denotes the multiplicity of \(u\) in the multiset \(T_i\).
Equivalently, \(y_i\) is the number of selected elements in \(T_i\), counted with multiplicity.

We delete the at most $B$ ungrouped left vertices. This changes the number of edges by at most
$B\Delta,$
and contributes at most
$Bp\,\hat W(G)\le \frac{2p}{\gamma}\hat W(G)$
to the distance.

Now consider the nonnegative objective
\[
Q(S)
:=
\frac{1}{M}\sum_{i=1}^k w_i\left(x_i-\frac{B}{2}\right)^2
+
\frac{1}{K}\sum_{i=1}^{\ell}\left(y_i-\frac{|T_i|}{2}\right)^2,
\]
where
\[
M:=\frac{B}{4}\sum_{i=1}^k w_i,
\qquad
K:=B\ell.
\]
If each left node is chosen independently with probability $1/2$, then
$\mathbb{E}\!\left[\left(x_i-\frac{B}{2}\right)^2\right]=\frac{B}{4},$
and, since $|T_i|\in[B,2B]$ and each element has multiplicity at most $2$,
$\mathbb{E}\!\left[\left(y_i-\frac{|T_i|}{2}\right)^2\right] = \Theta(B).$
Therefore
$\mathbb{E}[Q]=O(1).$

Expanding $Q(S)$ yields a quadratic polynomial in the indicator variables of the left vertices,
supported only on pairs of vertices that co-occur in some group.
This gives an auxiliary graph on $V_L$ with $O(|E|B)$ edges. The second term is handled by \cref{lem:near-independent-bundling}. For the first term,
the required groups come from the right neighborhoods, and can be expanded directly using the
stored adjacency-list positions.

Apply \cref{lem:aqo} to the objective $-Q(S)$ with parameter $\varepsilon:=\frac{1}{B}$.
Since \(\varepsilon\ge 1/\log^4 n\), this is within the range of \cref{lem:aqo}.
Because $\mathbb{E}[-Q]=-O(1)$ and the total edge weight of the auxiliary graph is $O(B)$,
we obtain a set $S$ with $Q(S)\le O(1)$.
In particular,
\[
\frac{1}{M}\sum_{i=1}^k w_i\left(x_i-\frac{B}{2}\right)^2 \le O(1),
\qquad
\frac{1}{K}\sum_{i=1}^{\ell}\left(y_i-\frac{|T_i|}{2}\right)^2 \le O(1).
\]

\paragraph{Distance bound.}
For each right node \(v\), let \(R_v\) be the set of leftover neighbors of \(v\), so \(|R_v|<B\).
Then
\[
|c_v-c_v'|
\le
2p\sum_{i:\,S_i\subseteq N(v)} \left|x_i-\frac{B}{2}\right|
+
p|R_v|.
\]
Summing over \(v\) gives
\[
d(G,G')
\le
2p\sum_{i=1}^k w_i\left|x_i-\frac{B}{2}\right|
+
pB\,\hat W(G).
\]
By Cauchy--Schwarz,
\[
\sum_{i=1}^k w_i\left|x_i-\frac{B}{2}\right|
\le
\sqrt{
\left(\sum_{i=1}^k w_i\right)
\left(\sum_{i=1}^k w_i\left(x_i-\frac{B}{2}\right)^2\right)
}.
\]
Let
\[
S_W:=\sum_{i=1}^k w_i.
\]
Since each right node contributes at most \(|N(v)|/B\) groups,
\[
S_W
\le
\frac{1}{B}\sum_{v\in V_R} w_v |N(v)|
=
\frac{1}{Bp}W(G).
\]
Also, from the bound on \(Q(S)\),
\[
\sum_{i=1}^k w_i\left(x_i-\frac{B}{2}\right)^2 \le O(M)=O(BS_W).
\]
Therefore,
$d(G,G') \le O\!\left( p\sqrt{S_W\cdot BS_W} \right) + \frac{2p}{\gamma}\hat W(G) \le O\!\left(\frac{1}{\sqrt{B}}W(G)\right) + \frac{2p}{\gamma}\hat W(G).$
Since \(B\asymp 1/\gamma\), this gives
$d(G,G') \le C\sqrt{\gamma}\,W(G)+\frac{Cp}{\gamma}\hat W(G).$

\paragraph{Edge reduction.}
Let $m:=\sum_{i=1}^{\ell}|T_i|$. Since at most \(B\) left vertices are ungrouped, we have
$|E(G)|-B\Delta \le m \le |E(G)|.$
Moreover,
\[
|E(G')|=\sum_{i=1}^{\ell} y_i,
\]
because all ungrouped left vertices were deleted.
Hence
\[
\left||E(G')|-\frac{m}{2}\right|
=
\left|\sum_{i=1}^{\ell} y_i-\frac12\sum_{i=1}^{\ell}|T_i|\right|
\le
\sum_{i=1}^{\ell}\left|y_i-\frac{|T_i|}{2}\right|.
\]
By Cauchy--Schwarz,
\[
\sum_{i=1}^{\ell}\left|y_i-\frac{|T_i|}{2}\right|
\le
\sqrt{
\ell\sum_{i=1}^{\ell}\left(y_i-\frac{|T_i|}{2}\right)^2
}
\le
O(\ell\sqrt{B}).
\]
Since each \(|T_i|\in[B,2B]\), we have \(m=\Theta(\ell B)\), and therefore
$\left||E(G')|-\frac{m}{2}\right| \le O\!\left(\frac{m}{\sqrt{B}}\right).$
As \(B\ge 100\), after adjusting constants we obtain
$|E(G')| \in \left[ \left(\frac12-C\sqrt{\gamma}\right)m,\, \left(\frac12+C\sqrt{\gamma}\right)m \right].$
Using \(m\le |E(G)|\) and \(m\ge |E(G)|-B\Delta\), we conclude
$|E(G')| \in \left[ \left(\frac12-C\sqrt{\gamma}\right)|E(G)|-\frac{C}{\gamma}\Delta,\, \left(\frac12+C\sqrt{\gamma}\right)|E(G)| \right].$
In particular, for sufficiently small \(\gamma\),
$|E(G')|\le \frac23\,|E(G)|$.\qedhere
\end{proof}

\subsection{Finishing}

\paragraph{Overview.}
The framework is similar to that of \cite{GG25}, but we use our reparameterization together with faster subroutines to achieve linear work.
We first deal with a fixed probability class.

Starting from a class whose probability is $2^{-i}$, we perform about $\log n$ boosting rounds: in each round, we sparsify the left side, keeping at most a $2/3$-fraction of the nodes, while doubling the probabilities of the surviving nodes. The decrease in the number of nodes creates room for additional work, allowing the total distance added over all rounds to form a geometric series.

Once the probability reaches $2^{-i/2}$, we stop this phase. The total distance incurred in this phase is of the form $\epsilon\bigl(W(G_i)+2^{-i/2}\hat{W}(G)\bigr)$.
By the triangle inequality, the total distance over all the probability classes over this phase is acceptable.

For the finishing phase after that, we repeatedly apply $\tfrac12$-sampling to the smallest remaining probability class. When the current smallest remaining class has probability $p=2^{-i}$, we know that it has at most a $2^{-i/2}$-fraction of its original size. This allows us to choose $\gamma=2^{-i/3}$ while still preserving the desired work bound.

Of course, these parameters must be combined with the appropriate \(\operatorname{poly}(\log n)\) floors and ceilings. We do not state these explicitly here, for the sake of keeping the overview simple.

\begin{lemma}[Fixed probability square root]\label{lem:fixed-prob-sampling}
Fix a constant $0<\epsilon<0.01$.
Let $G=(V_L\cup V_R,E)$ be a bipartite instance in which each $v\in V_R$ has weight $w_v>0$
and each left node has the same probability $p=2^{-t}$, where $c_0\le t\le 2\log n$ and
$c_0$ is a sufficiently large constant depending only on $\epsilon$.
Then there exists a deterministic parallel algorithm with $O(|V|+|E|)$ work and $\operatorname{poly}(\log n)$ depth that outputs a set $S\subseteq V_L$ such that,
letting $G'$ be the instance obtained by setting the left probabilities to $2^{-\lfloor t/2 \rfloor}$ on $S$ and to $0$
on $V_L\setminus S$ (keeping the same right weights),
$d(G,G') \le \epsilon\bigl(W(G) + 2^{-ct}\hat W(G)\bigr),$
where $c>0$ is an absolute constant.
Moreover,
$|E(G')| \le \left(\frac23\right)^{t/2}|E(G)|.$
\end{lemma}

Before the proof, we record a basic relation between $W(\cdot)$ and $d(\cdot,\cdot)$.

\begin{observation}\label{obs:W-lipschitz}
Suppose that $G$ and $G'$ are two instances with the same right part and the same right weights.
Then $|W(G)-W(G')|\le d(G,G')$.
\end{observation}

\begin{proof}
Write $W(G)=\sum_{v\in V_R} w_v c_v$ and $W(G')=\sum_{v\in V_R} w_v c_v'$.
Then
\[
|W(G)-W(G')|
=
\left|\sum_{v\in V_R} w_v(c_v-c_v')\right|
\le
\sum_{v\in V_R} w_v|c_v-c_v'|
=
d(G,G').
\qedhere
\]
\end{proof}

\begin{proof}[Proof of \cref{lem:fixed-prob-sampling}]
We invoke \cref{lem:half-sampling} repeatedly. Let
$r:=\left\lceil \frac{t}{2}\right\rceil,\qquad G_0:=G,\qquad p_i:=2^{i-t}\ \text{for } i=0,1,\dots,r.$
Thus $p_0=2^{-t}$, and $p_r$ is between $2^{-t/2}$ and $2^{-t/2+1}$, so after rounding we obtain the desired probability $2^{-\lfloor t/2\rfloor}$.

In round $i$ (from $G_{i-1}$ to $G_i$), \cref{lem:half-sampling} returns a subset of left nodes; we keep only those left nodes, remove isolated vertices, and set the new uniform probability to $p_i=2p_{i-1}$. The resulting instance is $G_i$, and $G_r$ is the desired $G'$.

\paragraph{Distance bound.}
In invocation $i$, choose
$\gamma_i := \max\!\left\{ \eta\epsilon^2\cdot 0.95^{\,i},\; \frac{1}{\log^4 n} \right\},$
where $\eta>0$ is a sufficiently small absolute constant.
For sufficiently large $c_0=c_0(\epsilon)$, all the instances covered by the lemma satisfy
$\gamma_i\le \epsilon_0$ for every $i$, so \cref{lem:half-sampling} applies.

By \cref{lem:half-sampling},
$d(G_{i-1},G_i)\le a\sqrt{\gamma_i}\,W(G_{i-1}) + a\frac{p_{i-1}}{\gamma_i}\hat W(G),$
for some absolute constant $a>0$.
By \cref{obs:W-lipschitz},
$|W(G_i)-W(G_{i-1})| \le d(G_{i-1},G_i),$
and hence
$W(G_i)\le \left(1+a\sqrt{\gamma_i}\right)W(G_{i-1}) + a\frac{p_{i-1}}{\gamma_i}\hat W(G).$

The multiplicative and additive error terms are both geometric. More precisely, let 
\[I:=\max\left\{i: 0.95^{\,i}\ge \frac{1}{\log^4 n}\right\}.\]
Then, we have
\[
\sum_{i=1}^r \sqrt{\gamma_i}
\le
\sqrt{\eta}\,\epsilon\sum_{i=1}^{\min\{r,I\}} 0.95^{\,i/2}
+
\sum_{i>I}\frac{1}{\log^2 n}
=
O(\epsilon),
\]
provided $\eta$ is sufficiently small. Likewise,
\[
\sum_{i=1}^r \frac{p_{i-1}}{\gamma_i}
\le
\frac{1}{\eta\epsilon^2}\sum_{i=1}^r \frac{2^{i-1-t}}{0.95^{\,i}}
=
O(2^{-ct})
\]
for some absolute constant $c>0$, since the sum is dominated by its last term and $r=\lceil t/2\rceil$.

Applying the standard inequality
\[
X_i\le (1+b_i)X_{i-1}+y_i
\quad\Longrightarrow\quad
X_i\le \exp\!\left(\sum_{j\le i} b_j\right)\left(X_0+\sum_{j\le i} y_j\right),
\]
with
$X_i=W(G_i),\qquad b_i=a\sqrt{\gamma_i},\qquad y_i=a\frac{p_{i-1}}{\gamma_i}\hat W(G),$
gives the following for every $i\le r$:
\[W(G_i)\le \exp\!\left(O(\epsilon)\right)\left(W(G)+O(2^{-ct})\hat W(G)\right).\]

Substituting this into the distance bound and summing over $i$, we get
\[
d(G,G')
\le
\sum_{i=1}^r d(G_{i-1},G_i)
\le
a\sum_{i=1}^r \sqrt{\gamma_i}\,W(G_{i-1})
+
a\sum_{i=1}^r \frac{p_{i-1}}{\gamma_i}\hat W(G).
\]
Using the bounds above,
$d(G,G') \le O(\epsilon)\left(W(G)+O(2^{-ct})\hat W(G)\right)+O(2^{-ct})\hat W(G).$
By choosing $\eta$ sufficiently small and then $c_0$ sufficiently large, the hidden constants can be absorbed into the stated form
$d(G,G') \le \epsilon\bigl(W(G)+2^{-ct}\hat W(G)\bigr).$

\paragraph{Edge bound.}
In each round, \cref{lem:half-sampling} guarantees
$|E(G_i)|\le \frac23 |E(G_{i-1})|.$
Therefore
$|E(G')| = |E(G_r)| \le \left(\frac23\right)^r |E(G)| \le \left(\frac23\right)^{t/2}|E(G)|.$

\paragraph{Work and depth bound.}
In round $i$, \cref{lem:half-sampling} uses
$O\!\left(\frac{|V(G_{i-1})|+|E(G_{i-1})|}{\gamma_i}\right)$
work and $\operatorname{poly}(\log n)$ depth.
After removing isolated vertices, we have
$|V(G_{i-1})|=O(|E(G_{i-1})|),$
and by the edge bound proved above,
$|E(G_{i-1})|\le \left(\frac23\right)^{i-1}|E(G)|.$
Therefore, the total work is
\[
\sum_{i=1}^r O\!\left(\frac{|E(G_{i-1})|}{\gamma_i}\right)
\le
O(|E(G)|)\sum_{i=1}^r \frac{(2/3)^{i-1}}{\gamma_i}.
\]

The work sum is geometric as well. For $i\le I$ we use $\gamma_i\ge \eta\epsilon^2\cdot 0.95^{\,i}$, giving
\[
\sum_{i=1}^{I} \frac{(2/3)^{i-1}}{\gamma_i}
\le
\frac{1}{\eta\epsilon^2}\sum_{i=1}^{I}\frac{(2/3)^{i-1}}{0.95^{\,i}}
=
O(1),
\]
since $\frac{2}{3\cdot 0.95}<1$. For $i>I$ we use $\gamma_i\ge 1/\log^4 n$, so
\[
\sum_{i>I}\frac{(2/3)^{i-1}}{\gamma_i}
\le
\log^4 n\sum_{i>I}(2/3)^{i-1}
=
o(1),
\]
because $0.95^I\asymp 1/\log^4 n$ implies $(2/3)^I=O(1/\log^C n)$ for some $C>4$. Therefore
\[
\sum_{i=1}^r \frac{(2/3)^{i-1}}{\gamma_i}=O(1),
\]
and the total work is $O(|E(G)|)$.

Including the initial vertex set contributes an additional $O(|V(G)|)$ term, so the total work is
$O(|V(G)|+|E(G)|).$

The depth is $\operatorname{poly}(\log n)$ per round, and the number of rounds is
$r=\left\lceil \frac{t}{2}\right\rceil \le \log n,$
so the total depth remains $\operatorname{poly}(\log n)$.
\qedhere
\end{proof}

We are now ready to prove the distance version, from which the main hitting-set theorem follows immediately by \cref{obs:hitting-to-distance}.

\begin{proof}[Proof of \cref{thm:distance_hitting_set}]
Fix a constant \(\epsilon>0\). Let \(\epsilon'>0\) be a sufficiently small constant fraction of \(\epsilon\), and let \(c_0\) be the constant given by \cref{lem:fixed-prob-sampling} for the error parameter \(\epsilon'\).

We first apply \cref{lem:graph-splitting} to decompose \(G\) into edge-disjoint subinstances
$G^{(1)},\dots,G^{(\tau)},$
where \(\tau=O(\log n)\), and within each \(G^{(i)}\) all left probabilities are equal.
For each subinstance \(G^{(i)}\), we apply \cref{lem:fixed-prob-sampling} to obtain
an instance \((G^{(i)})'\) on the same bipartite graph and with the same right-node
weights, such that every left probability in \((G^{(i)})'\) is either \(0\) or
\(2^{-\lfloor i/2 \rfloor}\), and
$d\!\left(G^{(i)},(G^{(i)})'\right)\le \epsilon'W\!\left(G^{(i)}\right)+ \epsilon' 2^{-ci} \hat W(G).$

Let \(G^{\ast}\) be the union of the instances \((G^{(i)})'\). This is well defined because
the subinstances \(G^{(i)}\) are edge-disjoint and share the same right-node weights.

Using the subadditivity of \(d(\cdot,\cdot)\) under edge-disjoint union, we obtain
\[
d(G,G^{\ast})
\le
\sum_{i=1}^{\tau} d\!\left(G^{(i)},(G^{(i)})'\right)
\le
\epsilon' \sum_{i=1}^{\tau} W\!\left(G^{(i)}\right)
+
\epsilon' \sum_{i=1}^{\tau} 2^{-ci}\hat W(G).
\]
Since the edge sets are disjoint,
\[
\sum_{i=1}^{\tau} W\!\left(G^{(i)}\right)=W(G),
\]
and since \(\sum_i 2^{-ci}=O(1)\), we get
$d(G,G^{\ast})\le \frac{\epsilon}{4}\bigl(W(G)+\hat W(G)\bigr)$
after choosing \(\epsilon'\) sufficiently small as a constant fraction of \(\epsilon\).

We now continue rounding up the remaining probabilities until all nonzero probabilities equal \(2^{-c_0}\).
At any stage, let \(2^{-u}\) be the smallest nonzero probability class, and let \(H_u\) denote the current subinstance consisting of the left vertices of probability \(2^{-u}\) together with all incident edges.

We apply \cref{lem:half-sampling} to \(H_u\) with
$\gamma_u=\max\!\left\{\frac{1}{\log^4 n},\,\eta\epsilon^2 2^{-cu}\right\},$
where \(\eta>0\) is a sufficiently small absolute constant.
This rounds the class \(2^{-u}\) up to probability \(2^{-u+1}\), after which we merge it into the existing class of that probability.

\paragraph{Distance bound for the finishing phase.}
By \cref{lem:half-sampling}, the distance introduced in the round for class \(u\) is at most
$a\sqrt{\gamma_u}\,W(H_u)+a\frac{2^{-u}}{\gamma_u}\hat W(G)$
for some absolute constant \(a>0\).

Both contributions are geometric. Since $\gamma_u\ge \eta\epsilon^2 2^{-cu}$,
\[
\sum_{u\ge c_0+1}\frac{2^{-u}}{\gamma_u}
\le
\frac{1}{\eta\epsilon^2}\sum_{u\ge c_0+1}2^{-(1-c)u}
=
O(2^{-(1-c)c_0}).
\]
So for sufficiently large $c_0$ the additive term contributes at most
$\frac{\epsilon}{8}\hat W(G).$

For the multiplicative term, we have
$\sqrt{\gamma_u}\le \sqrt{\eta}\,\epsilon\,2^{-cu/2}+\frac{1}{\log^2 n}.$
Also, throughout the finishing phase, the current instance stays within distance
$O\bigl(\epsilon(W(G)+\hat W(G))\bigr)$ of the original instance, so by \cref{obs:W-lipschitz} its total $W$-mass is always $O\bigl(W(G)+\hat W(G)\bigr)$. Therefore
\[
\sum_{u\ge c_0+1}\sqrt{\gamma_u}\,W(H_u)
\le
O\bigl(W(G)+\hat W(G)\bigr)
\left(
\sqrt{\eta}\,\epsilon\sum_{u\ge c_0+1}2^{-cu/2}
+
\sum_{u\ge c_0+1}\frac{1}{\log^2 n}
\right),
\]
and the right-hand side is at most $\frac{\epsilon}{8}\bigl(W(G)+\hat W(G)\bigr)$ once $\eta$ is chosen sufficiently small and then $c_0$ sufficiently large.

Hence, the entire finishing phase contributes distance at most
$\frac{\epsilon}{4}\bigl(W(G)+\hat W(G)\bigr).$

\paragraph{Work bound for the finishing phase.}
After the first stage, the current class of probability \(2^{-u}\) can only come from original classes \(G^{(i)}\) with \(i\ge 2u\). For each such original class, \cref{lem:fixed-prob-sampling} gives an edge bound of the form
$|E((G^{(i)})')|\le \left(\frac23\right)^{i/2}|E(G^{(i)})|.$
Therefore, summing over all \(i\ge 2u\),
$|E(H_u)|\le O(1)\left(\frac23\right)^u |E(G)|.$

The work of the round for class \(u\) is thus
$O\!\left(\frac{|E(H_u)|}{\gamma_u}\right) \le O(|E(G)|)\frac{(2/3)^u}{\gamma_u}.$
Using the definition of \(\gamma_u\), the same geometric-series calculation as in the proof of \cref{lem:fixed-prob-sampling} yields
\[
\sum_{u\ge c_0+1}\frac{(2/3)^u}{\gamma_u}=O(1),
\]
provided $c>0$ is chosen sufficiently small. Hence, the total work of the finishing phase is
$O(|E(G)|),$
and its total depth is \(\operatorname{poly}(\log n)\).

Combining the first stage and the finishing phase, and using the triangle inequality, we obtain
$d(G,G')\le \epsilon\bigl(W(G)+\hat W(G)\bigr).$
This proves \cref{thm:distance_hitting_set}.
\end{proof}

The main hitting-set theorem (\cref{thm:main-hitting-set}) is now immediate from \cref{thm:distance_hitting_set} and \cref{obs:hitting-to-distance}.
\section{$\Delta+1$ Coloring}

\subsection{Definition and Results}

\begin{definition}[$\Delta+1$ coloring]
A \(\Delta+1\) coloring of a simple undirected graph \(G=(V,E)\) is a labeling
\(f:V\to[\Delta+1]\) such that for every edge \(e=(u,v)\in E\), we have
\(f(u)\neq f(v)\), where \(\Delta\) denotes the maximum degree of \(G\).
\end{definition}

There is a trivial sequential greedy algorithm for the above problem. The main
result of this section is a parallel algorithm with the same asymptotic work as
the input size.

\begin{theorem}\label{thm:delta+1_coloring}
Given a simple undirected graph \(G=(V,E)\) with \(|V|\le n\), represented in
compact form and with maximum degree \(\Delta\), we can compute a
\(\Delta+1\) coloring of \(G\) in \(O(|V|+|E|)\) work and \(poly(\log n)\) depth.
\end{theorem}

\subsection{Black-Box Reduction through Dyadic Degree Splitting}

A natural generalization of this problem is list coloring, in which each vertex
must choose its color from a prescribed list. This setting arises naturally when
some vertices have already been colored, and we remove those colors from the
candidate lists of their neighbors. In this subsection, we reduce
\(\Delta+1\) coloring to a more restricted partial list-coloring problem, where
it suffices to color only a constant fraction of the vertices.

\begin{definition}[Partial \(\lbrack L+1\rbrack\) list coloring]
We define the partial \([L+1]\) list coloring problem as follows.

Let \(G=(V,E)\) be a compactly represented simple graph with \(|V|\le n\), and
suppose that every vertex has degree at most \(L\). Each vertex is given a
subset of \([L+1]\), represented compactly, as its list of candidate colors. It
is guaranteed that if a vertex has degree \(k\), then its list has size at least
\(k+1\). The goal is to compute a proper coloring of at least \(0.1|V|\) vertices,
using only colors from their lists, in \(O(|V|L)\) work and \(poly(\log n)\) depth.
\end{definition}

\cref{clm:reduction-delta-to-partial-list} reduces \(\Delta+1\) coloring to the
above partial list-coloring problem.

\begin{claim}\label{clm:reduction-delta-to-partial-list}
Suppose there is an algorithm for the partial \([L+1]\) list coloring problem.
Then we can solve the \(\Delta+1\) coloring problem in \(O(|V|+|E|)\) work and
\(poly(\log n)\) depth.
\end{claim}

\begin{proof}
For each \(i\), let \(V_i\) be the set of vertices whose degree in the original
graph lies in the range \([2^{i-1},2^i)\). Then
\[
m=\Theta\!\left(\sum_i 2^i |V_i|\right).
\]
Using \(\log n\) sorting (\cref{logn-sorting}), we can compute all sets \(V_i\)
in linear work.

We process the indices \(i\) from large to small and color the vertices in
\(V_i\). We do not remove colored vertices from the graph, so the sets \(V_i\)
remain fixed throughout the algorithm. For each \(i\), we repeatedly invoke the
partial \([L+1]\) list-coloring algorithm on the vertices of \(V_i\) that remain
uncolored. The graph passed to the oracle contains only these vertices and the
edges induced by them. We will use $L=\min(\Delta,2^i-1)$. The list of each such vertex is obtained by removing the
colors already used by its colored neighbors from the initial list $[L+1]$. This color-removal step uses the
CW property of the model: each colored neighbor tries to flag its color as
forbidden, and an arbitrary successful write still yields a correct flag.

The total work of that color-removal step for each vertex is bounded by its
original degree, namely \(O(2^i)\). The work for a fixed \(i\) is therefore
\(O(2^i |V_i|)\), since in the \(t\)-th iteration the number of vertices passed
to the oracle is at most \(0.9^t |V_i|\). Summing over all \(i\), the total work
is \(O(m)\).

For a fixed \(i\), the number of calls to the oracle is \(O(\log n)\), since in
each iteration the number of uncolored vertices decreases by a constant factor.
Therefore, the total depth is \(poly(\log n)\).
\end{proof}

\subsection{Partial List Coloring}

In this subsection, we prove \cref{thm:partial_list_coloring}, which solves the
partial \([L+1]\) list coloring problem.

\begin{theorem}\label{thm:partial_list_coloring}
Given an instance of the partial \([L+1]\) list coloring problem, we can solve it
in \(O(|V|L)\) work and \(poly(\log n)\) depth.
\end{theorem}

\paragraph{Method Overview.}
We use a gradual-rounding technique. Although our final algorithm is
deterministic, it is helpful to imagine first choosing, for each vertex, a
random color from its list. The expected number of monochromatic incidences is
upper bounded by the potential
\[
\sum_{u\in V}\frac{\deg(u)}{|\mathcal L_u|}.
\]

Our algorithm proceeds in \(O(\log n)\) stages. In each stage, we split the
current universe of candidate colors into two halves, and each vertex chooses
one of them. If the choice were made randomly, with probabilities proportional
to the intersection sizes, then the expected potential would remain unchanged.
We use defective coloring to derandomize this step while ensuring that the
potential increases by at most a sufficiently small multiplicative factor. After
enough stages, every set \(\mathcal L_u\) has been reduced to a singleton, while
the potential remains \(O(|V|)\). This implies that a constant fraction of the
vertices can be kept with a proper coloring.

\begin{proof}
We first construct a coloring with at most \(0.75|V|\) monochromatic edges, and
then extract the desired subset of properly colored vertices.

\paragraph{Constructing an almost-proper coloring.}
We may assume without loss of generality that \(L+1\) is a power of \(2\), and
write \(r=\log(L+1)\). The algorithm proceeds in \(r\) stages. At stage \(i\),
the universe \([L+1]\) is partitioned into \(2^i\) subsets of equal size
\(U_1,U_2,\dots,U_{2^i}\), and the vertex set \(V\) is partitioned into
\(V_1,V_2,\dots,V_{2^i}\). The invariant is that for every \(u\in V_j\), we have
\(\mathcal L_u\cap U_j\neq\emptyset\).

We define the potential
\[
P_i=\sum_{j=1}^{2^i}\sum_{u\in V_j}\frac{\deg_j(u)}{|\mathcal L_u\cap U_j|},
\]
where \(\deg_j(u)\) denotes the number of neighbors of \(u\) inside \(V_j\).
Clearly,
\[
P_0=\sum_{u\in V}\frac{\deg(u)}{|\mathcal L_u|}\le |V|,
\]
since \(|\mathcal L_u|\ge \deg(u)+1\) for every vertex \(u\). Let
\(\delta=\frac{1}{10\log n}\). We will ensure that \(P_i\le (1+\delta)P_{i-1}\)
for every \(i\).

Now consider stage \(i\ge 1\). For each \(j\), the set \(U_j\) from stage \(i-1\)
is split into two equal parts, which we denote by \(U_{2j-1}\) and \(U_{2j}\).
Correspondingly, we split \(V_j\) into two parts \(V_{2j-1}\) and \(V_{2j}\).

For a vertex \(u\in V_j\), define \(s_{u,0}=|\mathcal L_u\cap U_{2j-1}|\) and
\(s_{u,1}=|\mathcal L_u\cap U_{2j}|\). A randomized version of the algorithm
would place \(u\) into \(V_{2j-1}\) with probability
\(\frac{s_{u,0}}{s_{u,0}+s_{u,1}}\) and into \(V_{2j}\) with probability
\(\frac{s_{u,1}}{s_{u,0}+s_{u,1}}\).

We claim that under this distribution, \(\mathbb E[P_i]=P_{i-1}\). Fix a vertex
\(u\), and fix a neighbor \(v\) of \(u\) that lies in the same part \(V_j\) as
\(u\) at stage \(i-1\). Their contribution to \(P_{i-1}\) coming from the
endpoint \(u\) is \(1/(s_{u,0}+s_{u,1}) = 1/|\mathcal L_u\cap U_j|\). Now
condition on the event that \(v\) is placed into \(V_{2j-1}\). Then \(u\) is
also placed into \(V_{2j-1}\) with probability
\(\frac{s_{u,0}}{s_{u,0}+s_{u,1}}\), and in that case its contribution to \(P_i\)
is \(1/s_{u,0}\). Therefore the expected contribution from \(u\) is
\( \frac{1}{s_{u,0}} \cdot \frac{s_{u,0}}{s_{u,0}+s_{u,1}}
= \frac{1}{s_{u,0}+s_{u,1}} \), which is exactly its contribution to
\(P_{i-1}\). The same argument applies if \(v\) is placed into \(V_{2j}\).
Summing over all ordered pairs \((u,v)\) with \(uv\in E\) and using linearity of
expectation proves the claim.

It is convenient to rewrite the potential as
\[
P_i
=
\sum_{j=1}^{2^i}\ \sum_{e=(u,v)\in E[V_j]}
\left(
\frac{1}{|\mathcal L_u\cap U_j|}
+
\frac{1}{|\mathcal L_v\cap U_j|}
\right),
\]
where \(E[V_j]\) denotes the set of edges with both endpoints in \(V_j\).

To derandomize the \(i\)-th stage, we define an auxiliary weighted graph
\(H_{i-1}\) on the current vertex set \(V\). For every edge \(e=(u,v)\) with
both endpoints in the same current part \(V_j\) at stage \(i-1\), we include
the edge \(e\) in \(H_{i-1}\) with weight
\(w_e=\frac{1}{|\mathcal L_u\cap U_j|}+\frac{1}{|\mathcal L_v\cap U_j|}\). Then
the total edge weight of \(H_{i-1}\) is exactly \(P_{i-1}\).

Let \(\eta=\delta/2\). We run the
\(\left(\left\lceil \frac{2}{\eta}\right\rceil,\eta\right)\)
defective-coloring algorithm on \(H_{i-1}\). This produces
\(q=\left\lceil \frac{2}{\eta}\right\rceil
=O\!\left(\frac{1}{\delta}\right)=O(\log n)\) color classes
\(C_1,\dots,C_q\), such that the total weight of edges with both endpoints in
the same class is at most \(\eta P_{i-1}\). Call these edges \emph{bad}, and
call all other edges \emph{good}.

We now process the classes \(C_1,\dots,C_q\) one by one. When processing a class
\(C_t\), the choices made for earlier classes are already fixed, while the
choices for later classes are still viewed as random according to the above
distribution. For a vertex \(u\in C_t\) and a side \(c\in\{0,1\}\), let
\(A_u(c)\) denote the conditional expected contribution, from the \(u\)-endpoint,
of all good edges incident to \(u\), if we place \(u\) into the \(c\)-th child.
Also, let \(b_t(u)\) be the number of bad neighbors of \(u\) inside \(C_t\).

If we place \(u\) into side \(c\), then the total contribution of bad edges from
the \(u\)-endpoint is at most \(b_t(u)/s_{u,c}\). Indeed, each bad neighbor can
contribute at most \(1/s_{u,c}\) from the \(u\)-endpoint. Therefore, when
processing \(C_t\), we let each vertex \(u\in C_t\) choose the side
\(c\in\{0,1\}\) that minimizes \(A_u(c)+b_t(u)/s_{u,c}\).

These choices can be made in parallel within \(C_t\), because no good edge has
both endpoints inside \(C_t\): any such edge would have been bad by definition.
Hence, the quantities above are independent across the vertices of \(C_t\).

Now compare this deterministic choice with the randomized one. If \(X_u\) is the
random side chosen for \(u\), then \(\mathbb E\!\left[\frac{b_t(u)}{s_{u,X_u}}\right]
=
\frac{s_{u,0}}{s_{u,0}+s_{u,1}}\cdot \frac{b_t(u)}{s_{u,0}}
+
\frac{s_{u,1}}{s_{u,0}+s_{u,1}}\cdot \frac{b_t(u)}{s_{u,1}}
\le
\frac{2b_t(u)}{s_{u,0}+s_{u,1}}\). Therefore, the minimizing deterministic choice satisfies
\(A_u(c_u)+\frac{b_t(u)}{s_{u,c_u}}
\le \mathbb E[A_u(X_u)] + \frac{2b_t(u)}{s_{u,0}+s_{u,1}}\).

Summing over all \(u\in C_t\), we obtain that after fixing the choices of
vertices in \(C_t\), the sum of the conditional expected contribution of good
edges and the actual contribution of bad edges inside \(C_t\) increases by at
most
\[
2\sum_{u\in C_t}\frac{b_t(u)}{s_{u,0}+s_{u,1}}.
\]
Summing this over all classes \(C_t\), the additive loss is at most twice the
total weight of bad edges in \(H_{i-1}\), and hence at most
\(\delta P_{i-1}\).

Initially, the expected contribution of the good edges is at most
\(\mathbb E[P_i]=P_{i-1}\). After all classes have been processed, all random
choices have been fixed, so the conditional expectation becomes the actual
contribution of the good edges. Adding the contribution of the bad edges gives
\(P_i\le P_{i-1}+\delta P_{i-1}=(1+\delta)P_{i-1}\), as required.

\paragraph{Finishing the argument.}
After all \(r\) stages, we have \(2^r=L+1\) groups, each consisting of a single
color. Thus, the process defines a coloring of the graph. Since each vertex only
restricts its list and never introduces new colors, the final color assigned to
each vertex indeed belongs to its original list.

At this point, \(P_r\) is exactly the sum, over all vertices, of their
monochromatic degree in the final coloring. Therefore \(P_r\) is twice the number
of monochromatic edges. Since \(P_0\le |V|\), we obtain
\(P_r \le |V|(1+\delta)^r \le 1.5|V|\). Hence, the number of monochromatic edges
is at most \(P_r/2\le 0.75|V|\).

It follows that at most \(0.75|V|\) vertices can have monochromatic degree at
least \(2\); otherwise the sum of monochromatic degrees would exceed \(1.5|V|\).
Therefore at least \(0.25|V|\) vertices have monochromatic degree at most \(1\).

We keep every vertex of monochromatic degree \(0\). Among the vertices of
monochromatic degree \(1\), each connected component is a single edge, so we keep
exactly one endpoint of each such edge, say the endpoint with smaller ID. In this
way we retain at least half of the vertices of monochromatic degree at most \(1\),
and thus at least \(0.125|V|\) vertices. The retained vertices are properly
colored, as required.

\paragraph{Work and depth bounds.}
The depth bound is immediate from the fact that there are \(O(\log n)\) stages,
each of which has \(poly(\log n)\) depth.

For the work bound, the only nontrivial part is the gradual-rounding step.
Observe that \(P_i\le 1.5|V|\) for every stage \(i\). Therefore, the total number
of edge incidences in the graphs passed to the defective-coloring procedure at
stage \(i\) is bounded by
\[
\sum_{j=1}^{2^i}\sum_{u\in V_j}\deg_j(u)
\le
\max_{j,u} |\mathcal L_u\cap U_j|\cdot P_i
\le
\max_j |U_j|\cdot P_i
\le
1.5\cdot \frac{L+1}{2^i}\,|V|.
\]
Thus, the total work over all stages is
\[
\sum_{i=0}^{r} O\!\left(\frac{L}{2^i}|V|\right)=O(L|V|),
\]
as desired.
\end{proof}

\section{Maximal Matching}
\subsection{Definition and Results}

\begin{definition}[Maximal matching]
A \emph{maximal matching} of an undirected graph \(G=(V,E)\) is a subset
\(M\subseteq E\) such that
\begin{itemize}
    \item no two edges in \(M\) share an endpoint; and
    \item every edge in \(E\setminus M\) shares an endpoint with some edge in \(M\).
\end{itemize}
\end{definition}

In this section, we solve the maximal matching problem in linear work.

\begin{theorem}\label{thm:maximal_matching}
There exists a deterministic parallel algorithm that, given a simple undirected
graph \(G=(V,E)\) in compact representation, computes a maximal matching of \(G\)
in \(O(|V|+|E|)\) work and \(poly(\log n)\) depth.
\end{theorem}

\subsection{Somehow Maximal Matching}

\paragraph{Overview.}
In this section, we prove a key progress lemma toward maximal matching.
Let \(H\) be the set of vertices of degree at least \(d/2\), and let \(K:=|H|\).
We first restrict the graph to the edge set \(E_H:=\{e\in E:\text{$e$ is incident to at least one vertex of }H\}\).
Since every vertex of \(H\) has degree at least \(d/2\) and every edge of \(E_H\) is counted at
most twice when summing the degrees of vertices in \(H\), we have \(\frac{Kd}{4}\le |E_H|\le Kd\).

We now repeatedly perform a \emph{strict \(1/2\)-sampling} step on \((V,E_H)\).
In round \(i\), every surviving edge has current probability \(p_i\), and after the step the
surviving edges have probability \(2p_i\).
The step has two properties.
First, it keeps about half of the current edges.
Second, every vertex whose current degree is at least \(1/p_i\) has its degree reduced to at most
half.
Thus, after iterating this process until the probability becomes a fixed constant, the remaining
graph has maximum degree \(O(1)\), while still containing \(\Omega(K)\) edges.

Finally, we form the conflict graph whose vertices are the surviving edges and where two such
vertices are adjacent iff the corresponding edges in the original graph share an endpoint.
Because the surviving graph has maximum degree \(O(1)\), this conflict graph also has maximum
degree \(O(1)\).
We color it with \(O(1)\) colors and choose the largest color class.
This yields a matching of size \(\Omega(K)\).
Since every surviving edge is incident to a vertex of degree at least \(d/2\), deleting the
endpoints of this matching removes \(\Omega(Kd)\) edges from the original graph.

\begin{lemma}[Strict $\tfrac12$-sampling for edges]\label{lem:strict-half-sampling-edges}
There is an absolute constant \(C>0\) such that the following holds.

Let \(G=(V,E)\) be a simple graph represented in compact form, with \(|V|\le n\).
Let \(p\in(0,\tfrac14)\), and let \(\gamma\in[1/\log^4 n,\,0.01]\).
Then there is a deterministic parallel algorithm that computes a subgraph \(G'=(V,E')\), where \(E'\subseteq E\), in \(O((|V|+|E|)/\gamma)\) work and \(poly(\log n)\) depth such that:
\begin{enumerate}
    \item \(\left(\frac12-C\!\left(\sqrt\gamma+\frac{p}{\gamma}\right)\right)|E| \le |E'| \le \left(\frac12+\sqrt\gamma\right)|E|\).
    \item For every vertex \(v\in V\) with \(\deg_G(v)\ge 1/p\), we have \(\deg_{G'}(v)\le \frac12\,\deg_G(v)\).
\end{enumerate}
\end{lemma}

\begin{proof}
Let \(T:=\{v\in V:\deg_G(v)\ge 1/p\}\). If \(T=\emptyset\), then item \(2\) is vacuous, and we may simply keep any \(\lfloor |E|/2\rfloor\) edges. Thus assume \(T\neq\emptyset\), so \(|E|\ge 1/p\).

We build a bipartite instance \(B=(V_L\cup V_R,F)\) as follows.
\begin{itemize}
    \item The left side \(V_L\) consists of the edges of \(G\), one left node for each \(e\in E\).
    \item The right side \(V_R\) consists of the vertices in \(T\).
    \item A left node corresponding to an edge \(e=\{u,v\}\) is adjacent in \(B\) to each endpoint
    that lies in \(T\). Thus every left node has degree \(0\), \(1\), or \(2\) in \(B\).
    \item Every left node has probability \(p\).
    \item Every right node has weight \(1\).
\end{itemize}

Let \(S\subseteq V_L\) be the output of \cref{lem:half-sampling}, let \(E_1:=\{e\in E:\text{the left node corresponding to }e\text{ lies in }S\}\), and write \(G_1:=(V,E_1)\).
Since left nodes of \(B\) are in bijection with edges of \(G\), the edge-count guarantee of
\cref{lem:half-sampling} gives \(\left(\frac12-\sqrt\gamma\right)|E|-\frac{2C_0}{\gamma}\le |E_1|\le \left(\frac12+\sqrt\gamma\right)|E|\), where \(C_0\) is the constant in \cref{lem:half-sampling} and we used \(\Delta=2\) on the left side of the bipartite graph.
Because \(T\neq\emptyset\), we have \(|E|\ge 1/p\), and hence \(1/\gamma=O(p|E|)\).

We now repair \(G_1\) so that all vertices of \(T\) satisfy the desired degree bound.
For each \(v\in T\), define its overflow by \(o(v):=\max\{0,\deg_{G_1}(v)-\lfloor \deg_G(v)/2\rfloor\}\).
Using the compact representation of \(G_1\), for each \(v\in T\) we scan the adjacency list of
\(v\) and mark an arbitrary set of exactly \(o(v)\) incident edges for deletion.
Let \(M\subseteq E_1\) be the union of all marked edges, and define \(E':=E_1\setminus M\) and \(G':=(V,E')\).

Since an edge is deleted if it is marked by at least one endpoint, the arbitrary-write CRCW model
is sufficient for the marking step.
By construction, for every \(v\in T\), \(\deg_{G'}(v)\le \lfloor\deg_G(v)/2\rfloor\le \frac12\,\deg_G(v)\), which proves item \(2\).

It remains to bound \(|M|\).
Fix \(v\in T\), and let \(d_v:=\deg_G(v)\) and \(d_v':=\deg_{G_1}(v)\).
If \(o(v)>0\), then \(d_v'>\lfloor d_v/2\rfloor\), and a direct check gives \(o(v)\le \frac12|d_v-2d_v'|+\frac12\).
Hence, 
\[
\sum_{v\in T} o(v)
\le
\frac12\sum_{v\in T}\bigl|\deg_G(v)-2\deg_{G_1}(v)\bigr|
+\frac{|T|}{2}.
\]
Since every edge in \(M\) is counted in the overflow of at least one endpoint, \(|M|\le \sum_{v\in T}o(v)\).

Now apply the distance bound from \cref{lem:half-sampling}.
For every \(v\in T\), we have \(c_v=p\,\deg_G(v)\) and \(c_v'=2p\,\deg_{G_1}(v)\), so \(|c_v-c_v'|=p\,|\deg_G(v)-2\deg_{G_1}(v)|\).
Therefore
\[
p\sum_{v\in T}\bigl|\deg_G(v)-2\deg_{G_1}(v)\bigr|
\le
\left(2\sqrt\gamma+\frac{2p}{\gamma}\right)W(B),
\]
where
\[
W(B)=\sum_{v\in T} p\,\deg_G(v).
\]
Dividing by \(p\), we get \(\sum_{v\in T}|\deg_G(v)-2\deg_{G_1}(v)| \le C(\sqrt\gamma+p/\gamma)\sum_{v\in T}\deg_G(v)\) for a suitable absolute constant \(C\).

Also, since \(\deg_G(v)\ge 1/p\) for every \(v\in T\), we have \(|T|\le p\sum_{v\in T}\deg_G(v)\).
Combining the previous bounds yields \(|M|\le C(\sqrt\gamma+p/\gamma)\sum_{v\in T}\deg_G(v)\).
Finally, \(\sum_{v\in T}\deg_G(v)\le 2|E|\), so \(|M|\le C(\sqrt\gamma+p/\gamma)|E|\).

Therefore, 
\(|E'| \ge |E_1|-|M| \ge \left(\frac12-\sqrt\gamma-C\!\left(\sqrt\gamma+\frac{p}{\gamma}\right)\right)|E|,\)
which is \(\left(\frac12-C(\sqrt\gamma+p/\gamma)\right)|E|\) after adjusting \(C\).
The upper bound \(|E'|\le |E_1|\le \left(\frac12+\sqrt\gamma\right)|E|\) is immediate.

The work is that of one application of \cref{lem:half-sampling}, namely \(O((|V|+|E|)/\gamma)\), plus \(O(|V|+|E|)\) additional work for constructing \(G_1\), scanning adjacency lists, marking deletions, and compacting.
The depth remains \(poly(\log n)\).
\end{proof}

\begin{lemma}[Fixed-probability strict sampling]\label{lem:fixed-prob-strict-sampling}
There exist absolute constants \(c_0\in \mathbb N\) and \(c_1>0\) such that the following holds.

Let \(G=(V,E)\) be a simple graph represented in compact form, with \(|V|\le n\) and maximum
degree at most \(d=2^k\).
Let \(H:=\{v\in V:\deg_G(v)\ge d/2\}\), let \(K:=|H|\), and let \(E_H:=\{e\in E:\text{$e$ is incident to at least one vertex of }H\}\).
Then there is a deterministic parallel algorithm that computes a subgraph \(G'=(V,E')\), where \(E'\subseteq E_H\), in \(O(|V|+|E|)\) work and \(poly(\log n)\) depth such that:
\begin{enumerate}
    \item \(\Delta(G')\le 2^{c_0+1}\).
    \item \(|E'|\ge c_1 K\).
\end{enumerate}
\end{lemma}

\begin{proof}
Let \(G_0:=(V,E_H)\).
Since every vertex in \(H\) has degree at least \(d/2\), and every edge of \(E_H\) is counted
at most twice when summing the degrees of vertices in \(H\), we have \(\frac{Kd}{4}\le |E_H|\le Kd\).
Indeed,
\[
\sum_{v\in H}\deg_{G_0}(v)\ge K\cdot \frac d2,
\]
while each edge of \(E_H\) contributes at most \(2\) to this sum.

Choose \(t\) so that \(p_0:=2^{-t}\in [1/(32d),\,1/(16d)]\).
Fix a sufficiently large absolute constant \(c_0\), let \(r:=t-c_0\), and for \(i=0,1,\dots,r\) define \(p_i:=2^i p_0\).
Then \(p_r=2^{-c_0}\) is a fixed positive constant.

We define a sequence of subgraphs \(G_0,G_1,\dots,G_r\) iteratively.
For \(i=1,2,\dots,r\), let \(\gamma_i:=\max\{\eta\cdot \max(\sqrt{p_{i-1}},\,0.95^i),\,1/\log^4 n\}\), where \(\eta>0\) is a sufficiently small absolute constant to be fixed later, and obtain
\(G_i\) from \(G_{i-1}\) by applying \cref{lem:strict-half-sampling-edges} with parameters
\(p_{i-1}\) and \(\gamma_i\).

\paragraph{Maximum-degree bound.}
We claim that \(\Delta(G_i)\le 2/p_i\) for every \(i=0,1,\dots,r\).
This is immediate for \(i=0\), since \(\Delta(G_0)\le d\le 2/p_0\).
For the inductive step, fix \(v\in V\).
If \(\deg_{G_{i-1}}(v)\ge 1/p_{i-1}\), then \cref{lem:strict-half-sampling-edges} gives \(\deg_{G_i}(v)\le \deg_{G_{i-1}}(v)/2\le 1/p_i\).
Otherwise \(\deg_{G_i}(v)\le \deg_{G_{i-1}}(v)<1/p_{i-1}=2/p_i\).
Thus \(\Delta(G_i)\le 2/p_i\) for all \(i\), and in particular \(\Delta(G_r)\le 2/p_r=2^{c_0+1}\), proving item \(1\).

\paragraph{Number of surviving edges.}
For \(i=1,2,\dots,r\), \cref{lem:strict-half-sampling-edges} gives \(|E(G_i)|\ge (\frac12-\alpha_i)|E(G_{i-1})|\), where \(\alpha_i:=C(\sqrt{\gamma_i}+p_{i-1}/\gamma_i)\).
Multiplying by \(p_i=2p_{i-1}\), we get \(p_i|E(G_i)|\ge (1-2\alpha_i)\,p_{i-1}|E(G_{i-1})|\), and hence
\(p_r|E(G_r)| \ge \left(\prod_{i=1}^r (1-2\alpha_i)\right)p_0|E(G_0)|.\)

It therefore suffices to show \(\sum_{i=1}^r \alpha_i=O(1)\), in fact smaller than \(0.01\).
We bound the two parts of \(\alpha_i\) separately.

For \(\sum_i \sqrt{\gamma_i}\), we use \(\sqrt{\max(a,b,c)}\le \sqrt a+\sqrt b+\sqrt c\), giving
\[
\sum_{i=1}^r \sqrt{\gamma_i}
\le
\sqrt{\eta}\sum_{i=1}^r p_{i-1}^{1/4}
+
\sqrt{\eta}\sum_{i=1}^{I}0.95^{\,i/2}
+
O\!\left(\frac{r-I}{\log^2 n}\right),
\]
where \(I:=\max\{i:0.95^{\,i}\ge 1/\log^4 n\}\).
The first two sums are geometric, and the last term is \(O(1/\log n)\), so \(\sum_i \sqrt{\gamma_i}=O(\sqrt{\eta})+O(2^{-c_0/4}\sqrt{\eta})+o(1)\).

For \(\sum_i p_{i-1}/\gamma_i\), we split into the same three regimes as in the hitting-set analysis.
If \(\sqrt{p_{i-1}}\ge 0.95^i\), then \(\gamma_i\ge \eta\sqrt{p_{i-1}}\), so the contribution is at most
\[
\frac{1}{\eta}\sum_{i=1}^r \sqrt{p_{i-1}}
=
O\!\left(\frac{2^{-c_0/2}}{\eta}\right).
\]
If \(\sqrt{p_{i-1}}<0.95^i\) and \(i\le I\), then \(\gamma_i\ge \eta\cdot 0.95^i\), so
\(\frac{p_{i-1}}{\gamma_i} \le \frac{2^{-t-1}}{\eta}\left(\frac{2}{0.95}\right)^i.\)
This occurs on an initial interval of indices, since \(\sqrt{p_{i-1}}/0.95^i\) is multiplied each round by \(\sqrt2/0.95>1\), so the whole contribution is bounded by its last term, namely \(O(2^{-c_0}/\eta)\).
Finally, for \(i>I\), we use \(\gamma_i\ge 1/\log^4 n\), and the contribution is
\[
\log^4 n\sum_{i>I}2^{i-1-t}
=
O(2^{-c_0}\log^4 n).
\]

Combining the three regimes, we obtain
\[
\sum_{i=1}^r \alpha_i
\le
C\!\left(
O(\sqrt{\eta})
+
O(2^{-c_0/4}\sqrt{\eta})
+
O\!\left(\frac{2^{-c_0/2}}{\eta}\right)
+
O\!\left(\frac{2^{-c_0}}{\eta}\right)
+
O(2^{-c_0}\log^4 n)
+
o(1)
\right).
\]
We first choose \(\eta>0\) sufficiently small, and then choose \(c_0\) sufficiently large, so that \(\sum_{i=1}^r \alpha_i\le 0.01\).
Hence \(\prod_{i=1}^r (1-2\alpha_i)\ge 1-2\sum_i \alpha_i\ge 0.98\), and therefore \(p_r|E(G_r)|=\Omega(p_0|E(G_0)|)\).

Now \(p_0|E(G_0)|=\Theta(|E_H|/d)=\Theta(K)\), because \(p_0=\Theta(1/d)\) and \(|E_H|=\Theta(Kd)\).
Since \(p_r=2^{-c_0}\) is a fixed positive constant, we conclude that \(|E(G_r)|=\Omega(K)\).
Setting \(G':=G_r\) proves item \(2\).

\paragraph{Work and depth.}
Each application of \cref{lem:strict-half-sampling-edges} to \(G_{i-1}\) uses \(O((|V(G_{i-1})|+|E(G_{i-1})|)/\gamma_i)\) work and \(poly(\log n)\) depth.
After deleting isolated vertices and compacting, we have \(|V(G_{i-1})|=O(|E(G_{i-1})|)\).
Also, by item \(1\) of \cref{lem:strict-half-sampling-edges}, \(|E(G_i)|\le (\frac12+\sqrt{\gamma_i})|E(G_{i-1})|\le 0.7\,|E(G_{i-1})|\), since \(\gamma_i\le 0.01\).
Finally, \(\gamma_i\ge \eta\cdot 0.95^i\).
Therefore, the total work is
\[
\sum_{i=1}^r O\!\left(\frac{|E(G_{i-1})|}{\gamma_i}\right)
\le
O(|E|)\sum_{i\ge 1}\frac{0.7^{\,i-1}}{\eta\cdot 0.95^{\,i}}
=
O(|E|),
\]
because \(0.7/0.95<1\).
The initial restriction from \(G\) to \(G_0\) also uses \(O(|V|+|E|)\) work, so the total work is
\(O(|V|+|E|)\).
The number of rounds is \(r=O(\log d)\le O(\log n)\), and each round has \(poly(\log n)\) depth,
so the total depth remains \(poly(\log n)\).
\end{proof}

\begin{lemma}\label{lem:somehow-mm}
There exists a constant \(c>0\) such that the following holds.

Suppose \(G=(V,E)\) is a simple graph represented in compact form, with \(|V|\le n\) and maximum
degree at most \(d=2^k\).
Let \(H:=\{v\in V:\deg_G(v)\ge d/2\}\) and \(K:=|H|\).
Then there is a deterministic parallel algorithm that computes a matching \(M\subseteq E\) such
that deleting all endpoints of \(M\) removes at least \(cKd\) edges from \(G\).
The algorithm uses \(O(|V|+|E|)\) work and \(poly(\log n)\) depth.
\end{lemma}

\begin{proof}
Apply \cref{lem:fixed-prob-strict-sampling} to \(G\).
Let \(G'=(V,E')\), where \(E'\subseteq E_H\), be the resulting subgraph, with \(H:=\{v\in V:\deg_G(v)\ge d/2\}\) and \(K:=|H|\).
By \cref{lem:fixed-prob-strict-sampling}, we have \(\Delta(G')\le 2^{c_0+1}\) and \(|E'|\ge c_1K\) for absolute constants \(c_0\in \mathbb N\) and \(c_1>0\).

We now form the conflict graph \(Q\) of \(G'\).
The vertex set of \(Q\) is \(E'\), and two vertices of \(Q\) are adjacent iff the corresponding
edges of \(G'\) share an endpoint.
Since every vertex of \(G'\) has degree at most \(2^{c_0+1}\), every edge of \(G'\) is incident
to fewer than \(2^{c_0+2}\) other edges, and hence \(\Delta(Q)=O(1)\).

By \cref{thm:delta+1_coloring}, we can compute a proper \((\Delta(Q)+1)\)-coloring of \(Q\) in \(O(|V(Q)|+|E(Q)|)\) work and \(poly(\log n)\) depth.
Since \(\Delta(Q)=O(1)\), this is \(O(|E'|)\) work.
Let the color classes be \(C_1,C_2,\dots,C_q\), where \(q=O(1)\).
Every color class is an independent set in \(Q\), and therefore the corresponding set of edges in
\(G'\) is a matching.

Since the classes partition \(E'\), the largest class has size at least \(|E'|/q=\Omega(|E'|)=\Omega(K)\).
Let \(M\) be the matching in \(G'\) corresponding to this largest class.
Then \(|M|=\Omega(K)\).

Finally, every edge of \(E'\) lies in \(E_H\), so every edge of \(M\) is incident to at least one
vertex of \(H\), that is, to at least one vertex of degree at least \(d/2\) in the original graph
\(G\).
Because \(M\) is a matching, these high-degree endpoints are all distinct.
Therefore, deleting all endpoints of \(M\) removes at least \((d/2)\cdot |M|=\Omega(Kd)\) edges from the original graph.

For the complexity, the call to \cref{lem:fixed-prob-strict-sampling} uses \(O(|V|+|E|)\) work and \(poly(\log n)\) depth.
The conflict graph \(Q\) has \(|V(Q)|=|E'|\) and \(|E(Q)|=O(|E'|)\), since \(\Delta(Q)=O(1)\).
Thus, constructing \(Q\), coloring it, and extracting the largest color class all take \(O(|E'|)\subseteq O(|E|)\) work and \(poly(\log n)\) depth.
Hence the whole algorithm runs in \(O(|V|+|E|)\) work and \(poly(\log n)\) depth.
\end{proof}

\subsection{Wrap Up}

With the extractor above and the data structure in \cref{big-degree-ds}, we can now
complete the proof of maximal matching.

\begin{proof}[Proof of \cref{thm:maximal_matching}]
The algorithm proceeds in rounds.
In each round, let the current maximum degree be \(\Delta\in[2^k,2^{k+1})\).
We use \cref{big-degree-ds} to extract the induced subgraph on the left vertices of degree at least \(2^k\), and then apply \cref{lem:somehow-mm} to this subgraph.
Although \cref{big-degree-ds} is stated for bipartite graphs, we may reduce the general case to the bipartite case by replacing every edge \((u,v)\) with two bipartite edges \((u_L,v_R)\) and \((u_R,v_L)\); one side of the extracted bipartite graph then recovers the extracted graph on the original vertex set.
The resulting matching is added to the global matching, and we delete all matched vertices together with their incident edges.

Consider one round, and let \(H\) be the extracted subgraph.
If \(K\) denotes the number of vertices of \(H\) whose degree is at least \(2^k\), then every such vertex has degree in \([2^k,2^{k+1})\), since the current maximum degree lies in that range.
Hence \(|E(H)|=\Theta(K\,2^k)\).
By \cref{lem:somehow-mm}, the matching computed on \(H\) removes \(\Omega(K\,2^k)=\Omega(|E(H)|)\) edges from the current graph.
The work of the round is \(O(|V(H)|+|E(H)|+R)=O(|E(H)|+R)\), where \(R\) is the redundancy term from \cref{big-degree-ds} and we use \(|V(H)|=O(|E(H)|)\) after deleting isolated vertices.

Therefore, the total work over all rounds is
\[
\sum O(|E(H)|) + \sum O(R).
\]
The first sum is \(O(|E|)\), since each round removes a constant fraction of the edges
of the extracted subgraph, and the second sum is \(O(|V|+|E|)\) by
\cref{big-degree-ds}. Hence the total work is \(O(|V|+|E|)\).

For the depth bound, note that each degree bucket survives for at most \(O(\log n)\)
rounds, since every round removes a constant fraction of the edges in the current top
bucket. There are only \(O(\log n)\) buckets, so the total number of rounds is
\(O(\log^2 n)\). As each round has \(poly(\log n)\) depth, the overall depth remains
\(poly(\log n)\).

Finally, the matching produced is maximal, since in each round we delete every edge
incident to the newly chosen matching edges, and the algorithm terminates only when no
edges remain.
\end{proof}


\section{Set Cover}

\subsection{Definition and Results}

\begin{definition}[Set cover]
A set-cover instance consists of a family of sets \(S_1,S_2,\dots,S_\ell\) and a universe of
elements \(x_1,x_2,\dots,x_t\). It is guaranteed that every element belongs to at least one set.
A \emph{set cover} is a subfamily of the \(S_i\)'s whose union contains all elements. The
\emph{optimal set cover} is a set cover of minimum cardinality.
\end{definition}

\begin{theorem}\label{thm:set-cover}
There exists a deterministic parallel algorithm that, given a set-cover instance with sets
\(S_1,S_2,\dots,S_\ell\) and elements \(x_1,x_2,\dots,x_t\), where \(\ell,t\le n\), computes an
\(O(\log s\log\log n)\)-approximation to the optimal set cover in \(O(\ell+t+|E|)\) work and
\(poly(\log n)\) depth. Here \(|E|\) is the total size of the sets, the instance is represented as
a bipartite graph in compact form, and \(s:=\max_i |S_i|\).
\end{theorem}

\subsection{Reducing to Fractional Set Cover}

We first define the fractional relaxation of set cover.

\begin{definition}[Fractional set cover]
A fractional set cover for a set-cover instance with sets \(S_1,S_2,\dots,S_\ell\) and elements
\(x_1,x_2,\dots,x_t\) is a vector \(z=(z_1,z_2,\dots,z_\ell)\in[0,1]^\ell\) such that for every
element \(x_j\),
\[
\sum_{i:\,x_j\in S_i} z_i \ge 1.
\]
Its cost is \(\sum_{i=1}^{\ell} z_i\).
\end{definition}

We now show how a linear-work algorithm for approximately solving the fractional relaxation implies
a linear-work algorithm for integral set cover, losing only an additional \(O(\log\log n)\) factor.

\begin{lemma}\label{lem:frac-to-int-set-cover}
Suppose there is a deterministic parallel algorithm that, given a compactly represented set-cover
instance with \(\ell,t\le n\), computes a feasible fractional set cover of cost at most
\(A\cdot \mathrm{OPT}\) in \(O(\ell+t+|E|)\) work and \(poly(\log n)\) depth, where \(\mathrm{OPT}\)
denotes the optimum integral set-cover size. Then there is a deterministic parallel algorithm that
computes an integral set cover of size \(O(A\log\log n)\cdot \mathrm{OPT}\) in \(O(\ell+t+|E|)\)
work and \(poly(\log n)\) depth.
\end{lemma}

\begin{proof}
\paragraph{One iteration.}
Fix a residual set-cover instance with \(\ell\) sets, \(t\) elements, and \(|E|\) incidences, and
let \(z=(z_1,\dots,z_\ell)\) be a feasible fractional set cover of cost at most
\(A\cdot \mathrm{OPT}\).

For each set \(S_i\), round \(z_i\) up to the next power of two by setting
\(p_i:=0\) if \(z_i=0\), and \(p_i:=2^{\lceil \log_2 z_i\rceil}\) otherwise. Then
\(z_i \le p_i < 2z_i\) for every \(i\). Hence \(\sum_i p_i \le 2\sum_i z_i \le 2A\cdot \mathrm{OPT}\),
and for every element \(x_j\),
\[
\sum_{i:\,x_j\in S_i} p_i \ge \sum_{i:\,x_j\in S_i} z_i \ge 1.
\]

We now build a pre-hitting-set instance. The left side consists of the sets
\(S_1,\dots,S_\ell\), with probabilities \(p_i\). The right side consists of the elements
\(x_1,\dots,x_t\), together with one additional special node \(r^\star\). The special node is
adjacent to every set. Each element \(x_j\) receives weight \(\deg(x_j)/|E|\), so the total weight
of all element-nodes is \(1\). The special node \(r^\star\) also receives weight \(1\).

Apply \cref{thm:main-hitting-set} to this instance with a sufficiently small constant
\(\epsilon\), and let \(H\) be the resulting hitting set of left nodes. For an element-node
\(x_j\), its expected hit count is at least \(1\), so happiness implies that it is hit at least
once. For the special node \(r^\star\), the total right-node weight is \(2\), so the hitting-set
theorem guarantees happy weight at least \(2(1-\epsilon)\). For sufficiently small \(\epsilon\),
this is strictly larger than \(1\), hence \(r^\star\) itself must be happy. Consequently,
\(|H| = O(\sum_i p_i+1)=O(A\cdot\mathrm{OPT})\).

Moreover, the total weight of happy element-nodes is still a positive constant. Since the weight of
an element-node \(x_j\) is \(\deg(x_j)/|E|\), this implies that the selected sets cover a positive
constant fraction of the current incidences. By shrinking the constant in the theorem statement if
needed, we may assume they cover at least half of the current incidences.

Thus one application of \cref{thm:main-hitting-set} returns \(O(A\cdot\mathrm{OPT})\) sets that
cover at least half of the current incidences.

\paragraph{Iteration and cleanup.}
We iterate the previous step. In each iteration, we:
\begin{enumerate}
    \item compute an \(A\)-approximate fractional set cover on the current residual instance;
    \item round it up to power-of-two probabilities as above;
    \item apply \cref{thm:main-hitting-set};
    \item add the selected sets to the solution, delete all elements covered by them, and remove
    empty sets.
\end{enumerate}

Since each iteration removes at least half of the remaining incidences, after \(O(\log\log n)\)
iterations, the number of remaining incidences drops by an arbitrary polylogarithmic factor; for
concreteness, after \(100\lceil \log\log n\rceil\) iterations it is at most
\(|E|/\log^{100} n\). The number of sets chosen during these iterations is
\(O(A\log\log n)\cdot \mathrm{OPT}\).

We then finish on the residual instance using the deterministic parallel algorithm of
Berger, Rompel, and Shor~\cite{BergerRompelShor94}. Its approximation factor is \(O(\log s)\), and
its work is \(O((\ell'+t'+|E'|)\log^5 n)\) on the residual instance. After deleting covered
elements and removing empty sets, every remaining set and every remaining element is incident to at
least one residual edge, so \(\ell'+t'=O(|E'|)\). Since \(|E'|\le |E|/\log^{100}n\), this
contributes only \(O(\ell+t+|E|)\) total work. The number of additional sets selected is at most
\(O(A)\cdot \mathrm{OPT}\), and hence the final cover has size
\(O(A\log\log n)\cdot \mathrm{OPT}\).

Each iteration has \(poly(\log n)\) depth, and there are only \(O(\log\log n)\) iterations, so the
overall depth remains \(poly(\log n)\).
\end{proof}

\subsection{Fractional Set Cover}

\begin{lemma}\label{lem:fractional-set-cover}
Given a compactly represented set-cover instance with \(\ell,t\le n\), there is a deterministic
parallel algorithm that computes a feasible fractional set cover of cost
\(O(\log s)\cdot \mathrm{OPT}\) in \(O(\ell+t+|E|)\) work and \(poly(\log n)\) depth, where
\(s:=\max_i |S_i|\).
\end{lemma}

\begin{proof}
We view the instance as a bipartite graph, with the sets on the left and the elements on the
right. The algorithm proceeds in phases.

\paragraph{Method overview.}
In one phase, let the current maximum set size lie in \([2^k,2^{k+1})\). Using
\cref{big-degree-ds}, we extract the subinstance \(H\) consisting of all sets of current size at
least \(2^k\), together with the elements contained in those sets. Since the maximum remaining set
size is below \(2^{k+1}\), every set in \(H\) has size in \([2^k,2^{k+1})\).

During this phase, we increase the fractional values of the sets in \(H\). At the end of the
phase, every surviving set of \(H\) has size below \(2^{k-1}\). We then delete all elements whose
total assigned fractional value among all phases has reached at least \(1\), and continue to the
next phase.

We first describe one phase, then prove the approximation ratio, and finally analyze the work and
depth.

\paragraph{One phase.}
Fix one phase and let \(H\) be the extracted subinstance. Let \(D\) be the maximum degree on the
right side of \(H\), rounded up to the next power of \(2\). Initially, every set in \(H\) is
active and has current probability \(1/D\). In round \(r\), every active set has current
probability \(2^r/D\). The phase lasts for \(R:=\log D\) rounds, so in the last round the current
probability reaches \(1\).

For an element \(x\), let \(c_x\) be its degree in \(H\), and define its activation round by
\(\tau(x):=\max\{0,\lceil \log_2(D/c_x)\rceil\}\). Before round \(\tau(x)\), the total
contribution of the active sets containing \(x\) is less than \(1\), so \(x\) cannot yet be
removed. We therefore keep \(x\) asleep until round \(\tau(x)\), and only then start processing
it.

For each set \(S\), we maintain:
\begin{itemize}
    \item a status bit saying whether \(S\) is still active,
    \item a current list \(A(S)\) of awake elements still contained in \(S\), and
    \item a counter \(\sigma(S)\) equal to the number of sleeping elements still contained in
    \(S\).
\end{itemize}
Thus the current number of remaining elements of \(S\) is \(|A(S)|+\sigma(S)\). A set is asleep
when \(A(S)\) is empty. In that case, it does no work in the current round, but it still knows
\(\sigma(S)\), so whenever new elements wake up, it can correctly decide whether it should remain
active.

Using \(\log n\) sorting (\cref{logn-sorting}) once at the beginning of the phase, we bucket all
incidences \((x,S)\) by the activation round \(\tau(x)\). Thus, in round \(r\), each set \(S\) can
access the block of its incident elements that wake up in that round.

The update in round \(r\) is as follows.
\begin{enumerate}
    \item We wake up the sets that need to wake up in this round. Specifically, we consider all
    sets that contain an element whose activation round is \(r\). This list can be precomputed
    using \(\log n\) sorting (\cref{logn-sorting}). We prune it by checking the status bits and
    keep only those sets that are currently sleeping. We then mark these sets as awake and append
    them to the current awake list of sets.

    \item For each awake active set \(S\), scan the block of incidences \((x,S)\) with
    \(\tau(x)=r\). Every still-alive such element \(x\) is appended to \(A(S)\), and
    \(\sigma(S)\) is decreased by \(1\).

    \item Every awake element \(x\) scans its current list of incident active sets, removes those
    that have already become inactive, and adds their frozen probabilities to a fixed scalar
    \(f_x\). If the total \(f_x+a_x\cdot 2^r/D\) reaches at least \(1\), where \(a_x\) is the
    number of still-active incident sets of \(x\), then \(x\) is marked dead and deleted from the
    residual instance.

    \item Every awake active set \(S\) scans its list \(A(S)\) and removes the dead elements. If
    \(|A(S)|+\sigma(S)<2^{k-1}\), then \(S\) becomes inactive permanently and stores its current
    probability as its final contribution. Otherwise, if \(A(S)\) becomes empty, then \(S\) goes
    to sleep.

    \item We compact the awake-set list so that only sets that remain awake are carried to the next
    round.
\end{enumerate}

At the end of the phase, every set that stayed active throughout all \(R\) rounds has current
probability \(1\), so every remaining awake element incident to it is deleted. Every set that
became inactive did so only after its number of remaining elements had dropped below \(2^{k-1}\).
Hence every surviving set in \(H\) has size below \(2^{k-1}\), as required.

\paragraph{Approximation analysis.}
We charge the increase in the fractional value of a set to the elements still present when the
increase occurs. More precisely, suppose that during this phase a set \(S\) changes its current
probability from \(p\) to \(2p\). This contributes an additional amount \(p\) to the fractional
cost of \(S\). At that moment \(S\) is still active, so it contains at least \(2^{k-1}\) remaining
elements. We therefore distribute the charge \(p\) evenly among any \(2^{k-1}\) remaining elements
of \(S\), charging each of them \(p/2^{k-1}\). We also treat the initial assignment \(0\to 1/D\)
in the same way, charging \(1/(D2^{k-1})\) to \(2^{k-1}\) remaining elements of the set.

Fix an element \(x\), and let \(c_x'\) be the total fractional value incident to \(x\) at the last
moment in this phase when it is still present. If \(x\) is deleted during the phase, then
\(c_x' < 2\), since it was below \(1\) in the previous round and all active probabilities only
double from one round to the next. If \(x\) survives the phase, then certainly \(c_x' < 1\). Thus
in all cases \(c_x' < 2\). Since every increase in an incident set of size class \(k\) charges
\(x\) by at most \(p/2^{k-1}\), the total charge received by \(x\) during this phase is at most
\(c_x'/2^{k-1}=O(2^{-k})\).

Now fix an optimal integral set cover \(\mathcal O\), and assign each element to one set of
\(\mathcal O\) containing it. Consider a fixed set \(A\in \mathcal O\). During a phase with size
class \([2^k,2^{k+1})\), the number of still-present elements assigned to \(A\) is at most
\(2^{k+1}\), since the current maximum set size is below \(2^{k+1}\). Each such element receives
charge \(O(2^{-k})\) in the phase, so the total charge assigned to \(A\) in that phase is \(O(1)\).

There are only \(O(\log s)\) phases, so the total charge assigned to \(A\) over the whole
execution is \(O(\log s)\). Summing over all sets in \(\mathcal O\), we obtain that the total
fractional cost produced by the algorithm is \(O(\log s)\cdot |\mathcal O|
= O(\log s)\cdot \mathrm{OPT}\).

\paragraph{Work and depth.}
Fix one phase and let \(H\) be the extracted subinstance. Using \(\log n\) sorting
(\cref{logn-sorting}), we bucket the incidences of \(H\) by activation round. This costs
\(O(t_H+\ell_H+|E(H)|)\) work.

Inside the phase, each incidence \((x,S)\) is appended to the awake list \(A(S)\) at most once,
namely when \(x\) wakes up, and is removed from that list at most once, namely when \(x\) dies.
Similarly, each incidence is deleted from the active-set list of \(x\) at most once, namely when
\(S\) becomes inactive. Thus, all list maintenance over the entire phase costs \(O(|E(H)|)\) work.
Including the preprocessing, one phase uses \(O(t_H+\ell_H+|E(H)|)\) work, which is \(O(|E(H)|)\)
after removing isolated vertices.

To sum over the phases, note that if a set survives a phase, then by construction its number of
remaining elements drops by at least a factor of \(2\). Therefore, for every set \(S_i\), the sum
of its sizes over all phases is \(O(|S_i|)\). Summing over all sets yields
\[
\sum_{\text{phases }H} |E(H)| = O(|E|).
\]
The extractor \cref{big-degree-ds} contributes an additional total redundancy of
\(O(\ell+t+|E|)\). Hence the overall work is \(O(\ell+t+|E|)\).

For the depth bound, one phase contains \(R+1=O(\log D)\le O(\log n)\) rounds. Each round uses
only scans, compactions, prefix sums, and access to the precomputed activation blocks, all of which
have \(poly(\log n)\) depth. Since there are only \(O(\log s)\le O(\log n)\) phases, the total
depth is \(poly(\log n)\).
\end{proof}

\section{Maximal Independent Set}

\subsection{Definition and Results}

\begin{definition}[Maximal Independent Set]
A maximal independent set is a vertex set $V' \subseteq V$ for a graph $G=(V,E)$ such that:
\begin{itemize}
    \item the vertices in $V'$ are pairwise nonadjacent; and
    \item every vertex in $V$ is either in $V'$ or adjacent to a vertex in $V'$.
\end{itemize}
\end{definition}

The main result of this section is a linear-work parallel algorithm for maximal independent set. It follows the main derandomization idea of \cite{GG25}, but uses an improved hitting-set lemma with a quadratic objective.

\begin{theorem}\label{thm:mis}
Given a graph $G=(V,E)$ in compact representation, there is a deterministic parallel algorithm that finds a maximal independent set of $G$ in $O(|V|+|E|)$ work and $poly(\log n)$ depth.
\end{theorem}

\subsection{Reducing to Hitting Set via Luby's Algorithm}

To prove \cref{thm:mis}, it suffices to give an algorithm for finding an independent-ish set, exactly as defined in \cite{GG25}.

To simplify the analysis, we orient the edges of $G$. For each edge, we orient it from the lower-degree endpoint to the higher-degree endpoint, breaking ties by vertex IDs.

\begin{definition}[Independent-ish set]
Fix a small constant $c>0$ and a constant $C$. An independent-ish set of a graph $G=(V,E)$ is a subset $S\subseteq V$ such that the total degree of the vertices in $S\cup N(S)$ is at least $c|E|$, and every vertex in $S$ has out-degree at most $C$ in the induced subgraph $G[S]$. Here, $N(S)$ is the neighbor set of $S$.
\end{definition}

Indeed, given such a set, we can color the induced subgraph on $S$ using $2C+1$ colors in $O(|V|+|E|)$ work, as explained in the next lemma.

\begin{lemma}\label{lem:independent-ish-coloring}
Let \(G=(V,E)\) be a graph, and let \(S\subseteq V\) be an independent-ish set with parameter \(C\).
Then, the induced subgraph \(G[S]\) can be properly colored with colors from \([4C+1]\) in \(O(|V|+|E|)\) work and \(poly(\log n)\) depth.
\end{lemma}

\begin{proof}
Let \(H:=G[S]\). Every vertex of \(S\) has out-degree at most \(C\) in \(H\), so \(|E(H)|=\sum_{v\in S}\deg_H^+(v)\le C|S|\). Hence the average undirected degree in \(H\) is at most \(2C\). The same argument applies to every induced subgraph \(H'\) of \(H\), so every induced subgraph of \(H\) contains a vertex of degree at most \(4C\).

We therefore repeatedly peel away all vertices of degree at most \(4C\). In every nonempty induced subgraph, at least half of the vertices have degree at most \(4C\), since otherwise the average degree would exceed \(2C\). Thus the process removes a constant fraction of the remaining vertices in each step. Using repeated subgraph selection, we obtain a decomposition \(S_1,S_2,\dots,S_t\) with \(t=O(\log n)\), where every vertex in \(S_i\) has degree at most \(4C\) in the subgraph induced by \(S_i\cup S_{i+1}\cup \cdots \cup S_t\).

We now color the layers in reverse order. When we color a layer \(S_i\), all later layers have already been colored. For each vertex \(v\in S_i\), let its list be the set of colors in \([4C+1]\) not used by its already-colored neighbors. Since \(v\) has current degree at most \(4C\), the list size is at least \(k+1\), where \(k\) is the current degree of \(v\). Thus the hypothesis of the partial \([L+1]\) list-coloring problem is satisfied with \(L=4C\).

We apply the partial \([4C+1]\) list-coloring algorithm to the current layer. Each application colors a constant fraction of the remaining vertices of that layer, so \(O(\log n)\) repetitions color the whole layer. Since \(C\) is a constant, each application uses linear work in the current layer, and the uncolored part shrinks geometrically. Summing over all repetitions and layers gives \(O(\sum_i |S_i|)=O(|S|)\) work, plus the linear work needed to maintain the induced subgraphs and adjacency information. Thus the total work is \(O(|V|+|E|)\), and the total depth remains \(poly(\log n)\).

This yields a proper coloring of \(H=G[S]\) using colors from \([4C+1]\).
\end{proof}

\paragraph{Finishing with an independent-ish set.}
Let \(S_1,\ldots,S_q\) be the color classes of \(G[S]\), where \(q=O(C)\).
Since the closed neighborhoods \(S_i\cup N(S_i)\) cover \(S\cup N(S)\), there is
a color \(i^\star\) such that
\[
\sum_{v\in S_{i^\star}\cup N(S_{i^\star})}\deg_G(v)
\ge
\frac{1}{q}
\sum_{v\in S\cup N(S)}\deg_G(v)
=
\Omega(|E|).
\]
The set \(S_{i^\star}\) is independent, so we add it to the MIS and delete
\(S_{i^\star}\cup N(S_{i^\star})\). This removes \(\Omega(|E|)\) edges from the
current graph. Repeating on the remaining graph, while adding isolated vertices
to the MIS, gives a maximal independent set in total linear work.

\paragraph{Finding the independent-ish set.}
The next step is to define a hitting-set objective that directly yields an independent-ish set. This part is essentially the same as in \cite{GG25}, except that here we use a stronger version of the hitting-set lemma. We repeat the details for completeness.

Suppose that we intend to select each vertex \(v\) with probability \(p_v=\min\!\left(1,\frac{3}{\deg(v)}\right)\), rounded up to the nearest power of \(2\). Call a vertex \(v\) \emph{good} if its out-degree is at most twice its in-degree. A degree-counting argument shows that the total degree of the good vertices is at least \(|E|/2\).

For each good vertex \(v\), the total probability mass of its in-neighbors is at least \(1\). Indeed, \(v\) has at least \(\deg(v)/3\) in-neighbors, every in-neighbor \(u\) satisfies \(\deg(u)\le \deg(v)\), and therefore \(p_u \ge 3/\deg(v)\). Hence we may choose a subset \(IN^*(v)\subseteq N^{-}(v)\) such that \(1\le \sum_{u\in IN^*(v)} p_u \le 2\).

Fix a sufficiently large constant \(C\). We keep every selected vertex whose out-degree in the selected subgraph is at most \(C\), and discard every selected vertex whose out-degree in the selected subgraph exceeds \(C\).

To quantify the effect of discarding a selected vertex \(u\), define \(S_u \coloneqq \sum_{v:\,u\in IN^*(v)} \deg(v)\). Thus \(S_u\) is the total degree mass of good vertices \(v\) for which \(u\) is one of the vertices that can claim \(v\). If a selected vertex \(u\) is discarded, then the total degree contribution lost because of deleting \(u\) is at most \(S_u\).

This leads us to consider the quadratic objective \(\Phi (A)\coloneqq \sum_{u\to w, u, w \in A} S_u\). Its expectation is \(O(|E|)\), because
\[
\mathbb E[\Phi(A)]
=
\sum_v \deg(v)\sum_{u\in IN^*(v)} p_u \sum_{u\to w} p_w
\le
\sum_v \deg(v)\sum_{u\in IN^*(v)} p_u \cdot 3
\le
6\sum_v \deg(v)
=
O(|E|).
\]

Now, suppose that we have the following stronger version of the hitting-set lemma.

\begin{lemma}[Hitting set with auxiliary quadratic objective]\label{lem:strong-hitting-set}
Fix a constant \(\epsilon>0\). Let \(G=(V_L\cup V_R,E)\) be a pre-hitting-set instance, and let \(\Psi(S)=\sum_{e\in F[S]} b_e\) be an auxiliary objective on the left side, where \(F\) is a weighted multigraph on \(V_L\), and the quadratic coefficients \(b_e\) are nonnegative. Then there is a deterministic parallel algorithm that chooses a subset \(S\subseteq V_L\) such that at least a \((1-\epsilon)\)-fraction of the total right weight is happy, and \(\Psi(S)\le O\!\left(\mathbb E[\Psi(X)]\right)\), where \(X\) denotes the random subset obtained by independently selecting each left vertex according to its prescribed probability.
\end{lemma}

We apply this lemma to the hitting-set instance in which each right vertex \(v\) can only be hit by vertices in \(IN^*(v)\), together with the quadratic objective \(\Phi\). After doing so, we still have \(\Phi \le C_0|E|\) for some absolute constant \(C_0\).

Let \(H\) be the set of happy good vertices, and let \(D\) be the set of selected vertices that are later discarded because their out-degree in the selected subgraph exceeds \(C\). Then the total degree of the resulting independent-ish set and its neighbors is bounded below by \(\sum_{v\in H}\deg(v)-\sum_{u\in D} S_u\). Indeed, the first term measures the total degree mass contributed by happy good vertices, while the second upper-bounds the total loss caused by deleting selected vertices of excessively large out-degree.

Finally, if \(u\in D\), then its selected out-degree is at least \(C\), so its contribution to \(\Phi\) is at least \(C S_u\). Hence \(\sum_{u\in D} S_u \le \Phi/C\le (C_0/C)|E|\). Choosing \(C\) sufficiently large compared to \(C_0\), the second term becomes a small fraction of \(|E|\), whereas the first term is \(\Omega(|E|)\). Therefore, the resulting set is an independent-ish set.

\subsection{Hitting Set with Quadratic Auxiliary Objective}

\paragraph{Proof overview.}
We follow essentially the same proof strategy as for the original hitting-set lemma. We split the graph into probability classes, perform repeated \(\tfrac12\)-sampling inside each single probability class in order to obtain a transition \(p\to \sqrt p\), and then finish by repeatedly rounding up the smallest remaining class.

The difference is that here, in addition to controlling the hitting-set distance, we must also control an auxiliary quadratic objective, as well as one-sided shrinkage for both the number of edges of the current bipartite instance and the quantity \(\sum_{u\in V_L} d_u\). Here, the parameter \(d_u\) may be arbitrary, but its main intended use is to represent the degree of \(u\) in the full instance \(F\), when the current left side \(V_L\) is only a subset of that full instance. The required one-step statement is given by the next strengthened \(\tfrac12\)-sampling lemma. Once that lemma is proved, the rest of the argument is the same as before: first apply the \(p\to \sqrt p\) transition separately to each probability class, and then perform a finishing phase by repeatedly rounding up the smallest remaining class. The resulting distance and work bounds again come from summing geometric series.

\begin{lemma}[Strict \(\tfrac12\)-sampling with auxiliary objective]\label{lem:strict-half-sampling-aux}
There is an absolute constant \(C>0\) such that the following holds.

Let \(G=(V_L\cup V_R,E)\) be a bipartite instance in compact representation, with \(|V_L|,|V_R|\le n\). Assume every left vertex has the same probability \(p\in(0,\tfrac14)\), and let \(\gamma\in[1/\log^A n,\epsilon_0]\), where \(\epsilon_0>0\) is a sufficiently small absolute constant and \(A\) is a fixed positive constant.

In addition, let \(H=(V_L,F)\) be a weighted multigraph on the left side, with nonnegative vertex coefficients \((a_u)_{u\in V_L}\), positive edge weights \((b_e)_{e\in F}\), and positive integers \((d_u)_{u\in V_L}\). Define \(\Phi(S)\coloneqq \sum_{u\in S} a_u + \sum_{e\in F[S]} b_e\) for \(S\subseteq V_L\), where \(F[S]\) denotes the set of auxiliary edges whose two endpoints both lie in \(S\).

Then there is a deterministic parallel algorithm that outputs a subset \(S\subseteq V_L\) using \(poly(\log n)\) depth and \(O\!\left((|V_L|+|V_R|+|E|+|F|)/\gamma\right)\) work such that, letting \(G'\) be the instance obtained by setting the left probabilities to \(2p\) on \(S\) and to \(0\) on \(V_L\setminus S\), the following hold:
\begin{enumerate}
    \item \(d(G,G') \le C\gamma^{1/6}\,W(G) + C\frac{p}{\gamma}\,\hat W(G)\).
    \item \(|E(G')|\le \frac23 |E(G)|\) and \(\sum_{u\in S}d_u \le \frac23 \sum_{u\in V_L}d_u\).
    \item If \(\mathbb{E}_{X\sim \mathrm{Ber}(1/2)}[\Phi(X)]\) denotes the value of the auxiliary objective when each left vertex is independently chosen with probability \(1/2\), then
    \[
    \Phi(S)
    \le
    C\,\mathbb E_{X\sim \mathrm{Ber}(1/2)}[\Phi(X)]
    +
    C\gamma^{1/3}\sum_{e\in F} b_e.
    \]
\end{enumerate}
\end{lemma}

\begin{proof}
Let \(B:=\lceil 1/\gamma\rceil\). Since \(\gamma\le \epsilon_0\), we may assume \(B\) is sufficiently large.

\paragraph{Distance potential.}
For each \(v\in V_R\), partition \(N(v)\) into disjoint groups of size exactly \(B\), leaving at most \(B-1\) leftover neighbors. Let the resulting groups be \(S_1,\dots,S_k\), where each \(S_i\) comes from a single right vertex and inherits its weight \(w_i\). For a subset \(S\subseteq V_L\), define \(x_i:=|S\cap S_i|\), and set
\[
M_{\mathrm{dis}}:=\frac{B}{4}\sum_{i=1}^k w_i,
\qquad
\Phi_{\mathrm{dis}}(S):=
\frac{1}{M_{\mathrm{dis}}}\sum_{i=1}^k
w_i\left(x_i-\frac{B}{2}\right)^2.
\]
Under independent \(\tfrac12\)-sampling, \(\mathbb E[(x_i-B/2)^2]=B/4\), so \(\mathbb E[\Phi_{\mathrm{dis}}]=1\).

\paragraph{Edge-count potential.}
Apply \cref{lem:near-independent-bundling} to the left degrees in \(G\). This gives multisets \(T_1,\dots,T_{\ell_G}\), each of size in \([B,2B]\), such that all but at most \(B\) left vertices are represented exactly according to their degree demand. For each \(i\), let \(y_i\) be the number of selected elements of \(T_i\), counted with multiplicity. Set
\[
K_G:=B\ell_G,
\qquad
\Phi_G(S):=
\frac{1}{K_G}\sum_{i=1}^{\ell_G}
\left(y_i-\frac{|T_i|}{2}\right)^2.
\]
Since \(|T_i|\in [B,2B]\), the variance of \(y_i\) under independent \(\tfrac12\)-sampling is \(\Theta(B)\), so \(\mathbb E[\Phi_G]=O(1)\).

\paragraph{\(d_u\)-mass potential.}
Apply \cref{lem:near-independent-bundling} to the demand sequence \((d_u)_{u\in V_L}\). This gives multisets \(U_1,\dots,U_{\ell_D}\), again each of size in \([B,2B]\), such that all but at most \(B\) vertices are represented exactly according to their demand. For each \(i\), let \(z_i\) be the number of selected elements of \(U_i\), counted with multiplicity. Set
\[
K_D:=B\ell_D,
\qquad
\Phi_D(S):=
\frac{1}{K_D}\sum_{i=1}^{\ell_D}
\left(z_i-\frac{|U_i|}{2}\right)^2.
\]
Again \(\mathbb E[\Phi_D]=O(1)\).

\paragraph{Combined objective.}
Apply \cref{lem:aqo} to the objective \(-(\gamma^{1/3}(\Phi_{\mathrm{dis}}+\Phi_G+\Phi_D)+\Phi(S))\). Since all coefficients in \(\Phi\) are nonnegative, the AQO output \(S\) satisfies
\[
\gamma^{1/3}\bigl(\Phi_{\mathrm{dis}}(S)+\Phi_G(S)+\Phi_D(S)\bigr)+\Phi(S)
\le
C_1\gamma^{1/3}+\mathbb E[\Phi(X)]+C_1\gamma\sum_{e\in F}b_e
\]
for some absolute constant \(C_1\), where \(X\) denotes independent \(\tfrac12\)-sampling. In particular, \[\Phi_{\mathrm{dis}}(S), \Phi_G(S), \Phi_D(S)\le O(\gamma^{-1/3}),\] and \(\Phi(S)\le C\,\mathbb E[\Phi(X)] + C\gamma^{1/3}\sum_{e\in F}b_e\) for a suitable absolute constant \(C\).

\paragraph{Distance bound.}
Let \(S_W:=\sum_{i=1}^k w_i\). Exactly as in the proof of \cref{lem:half-sampling}, after deleting the leftover vertices not covered by the distance groups, we have
\[
d(G,G')
\le
2p\sum_{i=1}^k w_i\left|x_i-\frac{B}{2}\right|
+
O\!\left(\frac{p}{\gamma}\hat W(G)\right).
\]
By Cauchy--Schwarz,
\[
\sum_{i=1}^k w_i\left|x_i-\frac{B}{2}\right|
\le
\sqrt{\left(\sum_{i=1}^k w_i\right)
\left(\sum_{i=1}^k w_i\left(x_i-\frac{B}{2}\right)^2\right)}.
\]
Now \(\sum_i w_i(x_i-B/2)^2=M_{\mathrm{dis}}\Phi_{\mathrm{dis}}(S)\le O(BS_W)\cdot O(\gamma^{-1/3})\), so
\[
\sum_{i=1}^k w_i\left|x_i-\frac{B}{2}\right|
\le
O\!\left(S_W\sqrt{B\gamma^{-1/3}}\right).
\]
Also \(S_W\le W(G)/(Bp)\), as before. Substituting this yields \(d(G,G')\le O(\gamma^{1/6})W(G)+O\!\left((p/\gamma)\hat W(G)\right)\).

\paragraph{Edge bound.}
Let \(m_G:=\sum_{i=1}^{\ell_G}|T_i|\). Since at most \(B\) vertices are left uncovered by the bundling, and we delete all of them, we have \(|E(G')|=\sum_{i=1}^{\ell_G} y_i\) and \(m_G\le |E(G)|\). By Cauchy--Schwarz,
\[
\left||E(G')|-\frac{m_G}{2}\right|
=
\left|\sum_{i=1}^{\ell_G} y_i-\frac12\sum_{i=1}^{\ell_G}|T_i|\right|
\le
\sum_{i=1}^{\ell_G}\left|y_i-\frac{|T_i|}{2}\right|
\le
\sqrt{\ell_G\sum_{i=1}^{\ell_G}\left(y_i-\frac{|T_i|}{2}\right)^2}.
\]
Using \(\sum_i (y_i-|T_i|/2)^2=K_G\Phi_G(S)\le O(B\ell_G)\cdot O(\gamma^{-1/3})\), we get \(\left||E(G')|-m_G/2\right|\le O(\ell_G\sqrt{B\gamma^{-1/3}})\). Since \(m_G=\Theta(\ell_G B)\), this is \(O(\gamma^{1/6})m_G\). Hence \(|E(G')|\le (1/2+O(\gamma^{1/6}))m_G\le (1/2+O(\gamma^{1/6}))|E(G)|\). Choosing \(\epsilon_0\) sufficiently small yields \(|E(G')|\le \frac23 |E(G)|\).

\paragraph{\(d_u\)-mass bound.}
Let \(D:=\sum_{u\in V_L} d_u\) and \(m_D:=\sum_{i=1}^{\ell_D}|U_i|\). Since we delete all leftover vertices not covered by the bundling, \(\sum_{u\in S} d_u = \sum_{i=1}^{\ell_D} z_i\), and by construction \(m_D\le D\). The same calculation as above gives
\[
\left|\sum_{u\in S} d_u-\frac{m_D}{2}\right|
\le
O(\gamma^{1/6})\,m_D.
\]
Hence \(\sum_{u\in S} d_u \le (1/2+O(\gamma^{1/6}))m_D \le (1/2+O(\gamma^{1/6}))D\). Choosing \(\epsilon_0\) sufficiently small gives \(\sum_{u\in S} d_u \le \frac23 \sum_{u\in V_L} d_u\).

\paragraph{Work and depth.}
The three control potentials are built by one copy of the distance grouping and two applications of near-independent bundling. The resulting auxiliary graph has \(O((|E|+|F|)/\gamma)\) edges and can be constructed in \(O((|V_L|+|V_R|+|E|+|F|)/\gamma)\) work. One application of \cref{lem:aqo} then yields the same asymptotic work bound and \(poly(\log n)\) depth.
\end{proof}

With this version and the aforementioned scheduling, we can finish the proof of the hitting-set lemma with the auxiliary objective.

\begin{proof}[Proof of \cref{lem:strong-hitting-set}]
\paragraph{Outline.}
The algorithm has two phases.

In the first phase, we treat each probability class separately and repeatedly apply \cref{lem:strict-half-sampling-aux} in order to transform \(p\) into \(\sqrt p\). This phase is responsible for creating geometric shrinkage in both the size of the bipartite graph and the auxiliary graph \(F\).

In the second phase, we combine all remaining classes and repeatedly round up the smallest remaining probability class until every nonzero probability equals \(2^{-c_0}\). We then deterministically keep all vertices with nonzero probability.

By the triangle inequality, it suffices to show that each phase contributes \(O(\epsilon)(W(G)+\hat W(G))\) to the distance, while the auxiliary objective \(\Phi\) increases by at most a constant factor. The work bound follows from the same geometric shrinkage.

Throughout the proof, let \(L:=(\log n)^A\), where \(A>0\) is a sufficiently large absolute constant. Everywhere below, the floor \(1/\log^4 n\) from the earlier lemmas may be replaced by \(1/L\).

\paragraph{The square-root phase.}
We first split the bipartite instance into probability classes \(G^{(1)},G^{(2)},\dots,G^{(\tau)}\), where \(\tau=O(\log n)\), and in \(G^{(i)}\) every left probability equals \(2^{-i}\).

Fix one class \(G^{(i)}\). Let \(r_i:=\lceil i/2\rceil\) and \(p_{i,j}:=2^{j-i}\) for \(0\le j\le r_i\). Thus \(p_{i,0}=2^{-i}\), and \(p_{i,r_i}\) is between \(2^{-i/2}\) and \(2^{-i/2+1}\), so after rounding we obtain the target probability \(2^{-\lfloor i/2\rfloor}\).

At round \(j\), we apply \cref{lem:strict-half-sampling-aux} to the current subinstance consisting of the left vertices of class \(2^{j-i}\) and all right vertices adjacent to them. The auxiliary graph used in this call is the restriction of the current auxiliary graph \(F\) to the left vertices of this class. Thus:
\begin{itemize}
    \item if an edge of \(F\) has both endpoints in the current class, it contributes a quadratic term in this round;
    \item if it has exactly one endpoint in the current class, then after fixing the other endpoint it contributes a nonnegative linear term in this round;
    \item if neither endpoint lies in the current class, it plays no role in this round.
\end{itemize}
For the parameters \(d_u\), we use the degree of \(u\) in the current auxiliary graph \(F\).

For the \(j\)-th round inside class \(i\), choose
\[
\gamma_{i,j}:=
\max\!\left\{
\eta \epsilon^6\cdot 0.95^{\,j},\;
\frac{1}{L}
\right\},
\]
where \(\eta>0\) is a sufficiently small absolute constant.

\paragraph{Distance and work.}
The proof is the same as for the fixed-probability square-root lemma: replace every appearance of \(\sqrt{\gamma}\) there by \(\gamma^{1/6}\), and replace \(\epsilon^2\) by \(\epsilon^6\) in the definition of \(\gamma_{i,j}\). This yields
\[
d\bigl(G^{(i)},(G^{(i)})'\bigr)
\le
\frac{\epsilon}{4}\,W\!\left(G^{(i)}\right)
+
\frac{\epsilon}{4}\,2^{-ci}\hat W(G)
\]
for some absolute constant \(c>0\), provided \(\eta\) is chosen sufficiently small and then \(c_0\) sufficiently large.

The same proof also gives \(|E((G^{(i)})')|\le (2/3)^{i/2}|E(G^{(i)})|\) and, because \cref{lem:strict-half-sampling-aux} gives the same one-sided shrinkage for the parameters \(d_u\),
\[
\sum_{u\in V_L((G^{(i)})')} d_u
\le
\left(\frac23\right)^{i/2}
\sum_{u\in V_L(G^{(i)})} d_u.
\]
The work bound is still linear in the size of \(G^{(i)}\), and the depth remains \(poly(\log n)\).

\paragraph{Auxiliary objective.}
To control the total increase of \(\Phi\), we use a weighted potential \(\Phi^*\). For a term of \(\Phi\), let \(u\in\{0,1,2\}\) denote the number of endpoints whose probability class has already been processed in the square-root phase. Define \(\Phi^*\) by multiplying that term by \((1+\epsilon_1)^{-u}\), where \(\epsilon_1>0\) is a sufficiently small absolute constant chosen later.

We claim that \(\Phi^*\) is nonincreasing throughout the square-root phase. Consider one round acting on a fixed class. Only the terms touching that class can change. By the one-step lemma, the restriction of \(\Phi\) to those terms increases by at most a factor \(1+\epsilon_1\), after normalizing the parameters appropriately. Indeed, the multiplicative loss in one round is \(1+O(\gamma_{i,j}^{1/3})\), so the total multiplicative loss over all rounds relevant to one endpoint is at most
\[
\prod_j \bigl(1+O(\gamma_{i,j}^{1/3})\bigr)
\le
\exp\!\left(O\!\left(\sum_j \gamma_{i,j}^{1/3}\right)\right).
\]
Now \(\sum_j \gamma_{i,j}^{1/3}=O(\eta^{1/3}\epsilon^2)+o(1)\) by the choice of \(\gamma_{i,j}\). Therefore, by choosing \(\eta\) sufficiently small, the above product is at most \(1+\epsilon_1\). On the other hand, every such term has its value in \(\Phi^*\) multiplied by an additional factor \((1+\epsilon_1)^{-1}\), because exactly one more endpoint has now been processed. Hence, the contribution of all terms touching the current class does not increase in \(\Phi^*\), while all other terms remain unchanged. Therefore \(\Phi^*\) is nonincreasing.

At the beginning of the square-root phase, \(\Phi^*=\Phi\). At the end, every term has had at most two endpoints processed, so \(\Phi \le (1+\epsilon_1)^2 \Phi^* \le (1+\epsilon_1)^2 \Phi_{\mathrm{initial}}\). Thus, the total auxiliary objective increases by at most a constant factor during the square-root phase.

Summing the distance bounds over all classes and using \(\sum_i W(G^{(i)})=W(G)\) and \(\sum_i 2^{-ci}=O(1)\), we conclude that the total distance introduced in the square-root phase is at most \(\frac{\epsilon}{2}(W(G)+\hat W(G))\), while \(\Phi\) has increased by at most a constant factor. The total work of this phase is \(O(|V|+|E|+|F|)\), and the depth is \(poly(\log n)\).

\paragraph{The cleanup phase.}
After the square-root phase, the remaining nonzero probabilities are \(2^{-u}\) for \(u\ge c_0\). We now repeatedly round up the smallest remaining class until every nonzero probability equals \(2^{-c_0}\).

Fix a stage in which the smallest remaining probability is \(2^{-u}\). Let \(H_u\) be the current subinstance formed by that probability class. We apply \cref{lem:strict-half-sampling-aux} to \(H_u\) with parameter
\[
\gamma_u
:=
\max\!\left\{
\eta\epsilon^6 2^{-\beta u},\;
\frac{1}{L}
\right\},
\]
where \(0<\beta<\log_2(3/2)\) is a sufficiently small absolute constant, and \(\eta>0\) is the same small constant as above.

\paragraph{Distance.}
By the one-step lemma, the distance added at stage \(u\) is at most \(C\gamma_u^{1/6}W(H_u)+C2^{-u}\hat W(G)/\gamma_u\). The first term sums to \(O(\epsilon)(W(G)+\hat W(G))\), because \(\sum_{u\ge c_0}\gamma_u^{1/6}\le C\epsilon\sum_{u\ge c_0}2^{-\beta u/6}+O(\log n/L^{1/6})=O(\epsilon)\). For the second term,
\[
\sum_{u\ge c_0}\frac{2^{-u}}{\gamma_u}
\le
\frac{1}{\eta\epsilon^6}\sum_{u\ge c_0}2^{-(1-\beta)u}
=
O(2^{-c'c_0})
\]
for some \(c'>0\). Choosing \(c_0=c_0(\epsilon)\) sufficiently large makes this contribution at most \(\frac{\epsilon}{4}\hat W(G)\). Thus, the total distance introduced in the cleanup phase is at most \(\frac{\epsilon}{2}(W(G)+\hat W(G))\).

\paragraph{Auxiliary objective.}
Exactly the same stagewise decomposition as in the square-root phase shows that if \(\Phi_t\) is the current auxiliary objective, then one cleanup step at level \(u\) multiplies it by at most \(1+\delta_u\), where \(\delta_u:=C_2\gamma_u^{1/3}\) for some absolute constant \(C_2\). Therefore, the total multiplicative increase over the cleanup phase is at most
\[
\prod_{u\ge c_0}(1+\delta_u)
\le
\exp\!\left(\sum_{u\ge c_0}\delta_u\right).
\]
Now \(\sum_{u\ge c_0}\delta_u\le C_2\sum_{u\ge c_0}\gamma_u^{1/3}\le C_2\eta^{1/3}\epsilon^2\sum_{u\ge c_0}2^{-\beta u/3}+O(\log n/L^{1/3})=O(\epsilon^2)+o(1)\). Hence, the cleanup phase also increases \(\Phi\) by at most a constant factor.

\paragraph{Work.}
By the square-root phase, when the smallest remaining probability is \(2^{-u}\), both the number of edges of the current class and the total \(d_u\)-mass of that class are at most \(O((2/3)^u)\) times their original values. Thus the work of the cleanup step at level \(u\) is bounded by \(O((2/3)^u(|E|+|F|)/\gamma_u)\). Since \(\beta<\log_2(3/2)\), we have \((2/3)\cdot 2^\beta<1\), and therefore \(\sum_{u\ge c_0}(2/3)^u/\gamma_u=O(1)\). It follows that the total work of the cleanup phase is \(O(|E|+|F|)\), and the depth is again \(poly(\log n)\).

\paragraph{Conclusion.}
At the end of the cleanup phase, every nonzero probability equals \(2^{-c_0}\). We now keep all left vertices whose probability is nonzero. By \cref{obs:hitting-to-distance} and the total distance bound proved above, at least a \((1-\epsilon)\)-fraction of the total right weight is happy. The total multiplicative increase of \(\Phi\) over the two phases is bounded by an absolute constant. The total work is \(O(|V|+|E|+|F|)\), and the total depth is \(poly(\log n)\).
\end{proof}

\bibliographystyle{alpha}
\bibliography{refs, ref2}
\appendix
\end{document}